\documentclass[conference]{IEEEtran}

\pdfoutput=1

\makeatletter

 \let\proof\@undefined
 \let\endproof\@undefined
 \makeatother

\makeatletter
\def\ps@headings{%
\def\@oddhead{\mbox{}\scriptsize\rightmark \hfil \thepage}%
\def\@evenhead{\scriptsize\thepage \hfil \leftmark\mbox{}}%
\def\@oddfoot{}%
\def\@evenfoot{}}
\makeatother
\usepackage{cite}
\usepackage{amsmath}
\usepackage{amssymb}
\usepackage{graphicx} 
\usepackage{amsthm}
\usepackage{amsfonts}
\usepackage{color}
\usepackage{bbm}
\usepackage{verbatim}
\usepackage{dsfont}
\usepackage{subfigure}

\newcommand{\beqa}{\begin{eqnarray}}
 \newcommand{\eeqa}{\end{eqnarray}}
\newcommand{\beqano}{\begin{eqnarray*}}
\newcommand{\eeqano}{\end{eqnarray*}}

\newcommand{\beit}{\begin{itemize}}
\newcommand{\eeit}{\end{itemize}}

\newtheorem{definition}{Definition}

\newtheorem{lemma}{Lemma}

\newtheorem{proposition}{Proposition}
\newtheorem{assumption}{Assumption}

\newcommand{\al}{\alpha}
\newcommand{\be}{\beta}
\newcommand{\eps}{\epsilon}

\newcommand{\varep}{\varepsilon}

\newcommand{\setT}{\mathcal{T}}
\newcommand{\setN}{\mathcal{N}}
\newcommand{\setL}{\mathcal{L}}

 \newcommand{\setE}{\mathcal{E}}

\newcommand{\setI}{\mathcal{I}}
\newcommand{\setA}{\mathcal{A}}

\newcommand{\setC}{\mathcal{C}}
\newcommand{\setB}{\mathcal{B}}
\newcommand{\setH}{\mathcal{H}}
\newcommand{\setD}{\mathcal{D}}

\newcommand{\retheta}{\boldsymbol{\theta}_{\mathbf{e}=\mathbf{0},\mathbf{B}=\mathbf{0}}}

\newcommand{\bftheta}{\boldsymbol{\theta}_{L}}

\newcommand{\bfe}{\mathbf{e}}
\newcommand{\bfI}{\mathbf{I}}
\newcommand{\bfx}{\mathbf{x}}

\newcommand{\bfA}{\mathbf{A}}
\newcommand{\bfB}{\mathbf{B}}
 
 \newcommand{\bfV}{\mathbf{V}}
\newcommand{\vz}{\mathbf{z}}
\newcommand{\bfy}{\mathbf{y}}

\newcommand{\ev}{\mathds{E}}

\begin{document}

\title{ Throughput-Optimal Random Access with Order-Optimal Delay}

\author{Mahdi Lotfinezhad and Peter Marbach \\
Department of Computer Science, University of Toronto  \\
E-mail: \{mahdi,marbach\}@cs.toronto.edu \\
\vspace{-.1in}}


\maketitle

\begin{abstract}
In this paper, we consider CSMA policies for scheduling of multihop wireless networks with one-hop
traffic. The main contribution of this paper is to propose Unlocking CSMA (U-CSMA) policy that enables to obtain 
high throughput with low (average) packet delay for large wireless networks. In particular, the delay 
under U-CSMA policy becomes \emph{order-optimal}. For one-hop traffic, delay is
 defined to be order-optimal if
it is $O(1)$, i.e., it stays bounded, as the 
network-size increases to infinity. Using mean field theory techniques, we analytically show that 
for torus (grid-like) interference topologies with one-hop traffic, to achieve a network load of $\rho$, the delay under U-CSMA policy becomes $O(1/(1-\rho)^{3})$ as the network-size increases, and hence, delay becomes order optimal. We conduct simulations for 
general random geometric interference topologies under U-CSMA policy combined with 
congestion control to maximize
a network-wide utility. These simulations 
confirm that order optimality holds, and that we can use U-CSMA policy jointly with congestion control to operate close to the 
optimal utility with a low packet delay in arbitrarily large random geometric topologies. To the best
of our knowledge, it is for the first time that a simple distributed scheduling policy is proposed that in 
addition to throughput/utility-optimality exhibits delay order-optimality.

\end{abstract}


\section{Introduction}  

One of the most intriguing challenges in the context of wireless networking is the design of a
scheduling policy that has the following properties:
\begin{itemize}
\item [a)] throughput-optimality,
\item [b)] low packet delay\footnote{Throughout this paper, by delay we mean average packet 
delay.}, and
\item [c)] simple and fully distributed implementation.
\end{itemize}

From a complexity theoretic viewpoint, unless $\bf{NP\subseteq BPP} $ or $\bf{P=NP}$, 
there does not exist \cite{Shah09hardness} a universal scheduling policy
that has the above three properties for all possible network topologies. However, it is still
 possible to design a policy that has the above properties for a subset class of network topologies. 
 This seems to be true for geometric networks \cite{mazumdar:shroff:06,joo:shroff:infocom:08}, in which only links that are geometrically 
close interfere with each other. These networks closely approximate a wide range 
of practical wireless networks, and yet are known to admit
Polynomial-Time Approximation Scheme (PTAS) for 
several NP-hard optimization problems (see e.g., \cite{Hunt1998, mazumdar:shroff:06}). For this reason, our focus in this paper is to design a scheduling policy for large geometric wireless networks. 

There are two main approaches into the design of scheduling policies in wireless networking: either through matching policies 
\cite{tassiulas:ephremides:jac:92,buche:kushber:04,eryilmaz:sirkant:ton:05,neely:modiano:infocom:05,
neely:jsac05,mazumdar:shroff:06,dimakis:warland:06, lin:shroff:ton06,neely:info08, joo:shroff:infocom:08, zussman:modiano:08, yi:proutiere:08,Le:modiano:09,gupta:shroff:09}, or through random access policies \cite{abramson:70,haejk:82,kelly:macphee:87, hastad:96, bianchi:00,goldberg:00,Baccelli06analoha, marbach:CDC:07,lee:chiang:07, stolyar_dynamic_2008, durvy:08,garetto:08, Jiang:09,ni-2009,rajag:09,mahdi:peter:al09,rad:09,jiang:mobihoc:09, proutiere:10,ghaderi:10}. 
Despite the past efforts that have significantly
advanced our understanding of these policies and their performance, to the best of our knowledge,
there is no instance of these policies that realizes all of the three properties mentioned
 earlier, even for geometric networks.

On one hand, we have matching policies that can be throughput-optimal 
\cite{tassiulas:ephremides:jac:92, buche:kushber:04,eryilmaz:sirkant:ton:05,neely:jsac05} and can provide 
order-optimal low delay \cite{Le:modiano:09}. However, these optimalities are obtained assuming that an NP-hard
problem can be solved in each scheduling round. At the same time, reducing the complexity of
matching policies, in general, comes at the price of losing throughput-optimality of these
policies \cite{lin:shroff:ton06,dimakis:warland:06,zussman:modiano:08,neely:info08} or 
a large delay \cite{Shah09hardness}. This leaves the design of a matching policy with all the three properties as an 
open research challenge (see Section~\ref{sec:related} for further discussion on matching policies).

On the other hand, we have random access policies that are naturally simple and can be
implemented in distributed manner. 
Among these policies, the classical CSMA policy, see \cite{marbach:CDC:07,durvy:08,Jiang:09,ni-2009,rajag:09} for variants of this policy, is the one that features throughput-optimality in a wide range of settings \cite{marbach:CDC:07,Jiang:09,ni-2009,rajag:09}. This encourages to use the classical CSMA policy for scheduling of
wireless networks. However, as will be discussed shortly, the delay performance under this policy can be very poor. 
As a result, the current random access policies do not possess all of the three properties mentioned
earlier. 


As a motivating example, 
consider an $n\times n$ torus (see Fig.~\ref{fig:lattice_g})
\emph{interference graph} \cite{mazumdar:shroff:06,Jiang:09,ni-2009,rajag:09} with $L=n^{2}$ nodes where each node interferes with the four closest neighbouring nodes. 
Suppose packet arrival rate is uniform, i.e., it is the same for all nodes, and that is equal to $\lambda$. 
Let $\rho$ be the corresponding load\footnote{In the limit of large toruses, the maximum
uniform throughput is $0.5$, and load $\rho$ in the limit becomes $\frac{\lambda}{0.5}$. See Section~\ref{sec:mod:lattice:uaa} for the definition of $\rho$.}, and define
$$\eps=1-\rho.$$
For this simple topology, a mixing-time analysis \cite{rajag:09} upperbounds the packet delay under the classical CSMA policy as 
$$O\Big(\Big[ \frac{1}{\eps}\Big]^{c_{u}L}\Big),$$
where $c_{u}>1$ is a constant.
For small $\eps$, a similar analysis \cite{borgs_torpid_1999} lowerbounds the packet 
delay under the classical CSMA policy 
as
$$\Omega(e^{c_{l}L/(logL)^{2}}),$$
for some constant $c_{l}>0$.

The above delay-bounds show that the classical CSMA policy exhibits
 a \emph{threshold behaviour} in the sense in order to achieve 
a high throughput, i.e., to make $\eps$ small, one has to tolerate a delay that
 \emph{exponentially grows}
 with the network-size $L$. The threshold behaviour and the exponential growth are related to the
 \emph{phase transition}
phenomenon\footnote{Phase transition has also been reported as the cause of border
 effects that persist in 2D under the classical CSMA policy \cite{durvy:08}.} in the hard-core lattice gas model \cite{gaunt:65,aldous:prob:book}. 
Due to such threshold behaviours, even in mid-sized simple topologies, the classical CSMA policy cannot support
 a high throughput with low delay (see Section~\ref{sec:motivation}).


In this paper, we propose Unlocking CSMA (U-CSMA) as a new CSMA policy that overcomes the 
threshold behaviour of the classical 
CSMA policy. While being simple and distributed, U-CSMA policy has the following properties for
 geometric networks with one-hop traffic \cite{mazumdar:shroff:06, joo:shroff:infocom:08}.
 \begin{itemize}
 \item [a)] It enables to achieve a high throughout/utility arbitrarily close to the optimal
 with a low (average) packet delay.
\item [b)] The (average) packet delay under this policy is order-optimal, i.e., it stays bounded
as the network-size increases to infinity.
 \end{itemize}

 We provide analytical results for the torus interference topology with uniform 
packet arrival rate as considered earlier, and show that for large network-size $L$,
the average delay under U-CSMA policy is order-optimal and is
\begin{align}
  O\Big(\Big[ \frac{1}{\eps}\Big]^{3}\Big). \label{eq:dbound:ucsma}
\end{align}
It is important to note that the above delay bound is independent of the network size $L$, 
in sheer contrast to
 the delay under the classical CSMA
policy that exponentially increases with the network-size $L$. This means that U-CSMA 
policy does not suffer from the threshold behaviour and is indeed able to provide high throughput with low delay for arbitrarily large torus
topologies.

In our simulation study, we use U-CSMA policy jointly with a congestion control algorithm 
to maximize a network-wide utility in large random geometric networks. We show that
using U-CSMA, we can assign packet arrival rates closely to the optimal
 with a low packet delay that stays bounded as the network-size increases, and hence, a delay
that exhibits order-optimality. As far as we are aware, it is for 
the first time that a simple distributed scheduling policy is proposed that 
can operate close to the optimal with order-optimal low packet delay.

We believe that the design principle of U-CSMA policy and the novel approach taken to study
its performance open up a new direction into the design and study of scheduling policies 
for large-scale wireless networks. The main significance of our study in this paper
is that it realizes 
the possibility of having large-scale wireless multihop networks that can be maintained
in a simple distributed manner and that can provide high throughput/utility, arbitrarily
close to the optimal, with 
order-optimal low packet delay.

A key step to obtain the delay bound in \eqref{eq:dbound:ucsma} is where we show that the
schedule under the classical CSMA policy quickly converges to 
a maximum schedule in geometric networks. Using techniques from mean field theory \cite{Leboudec:08}, 
we show that for large torus and lattice topologies with large uniform attempt-rates, the \emph{distance} (see Section~\ref{sec:result:1}) to the maximum
schedules as a function of time $t$ drops as $\frac{1}{\sqrt{t}}.$
To the best of our knowledge, our result is the first that analytically characterizes the fast 
convergence behaviour of the classical CSMA policy. As this convergence is independent of network-size $L$, 
it is fundamentally different
 than the convergence time to the steady-state (i.e., the mixing time) of the dynamics of the classical CSMA policy, 
which can be exponentially large in $L$ \cite{borgs_torpid_1999}.

The rest of the paper is organized as follows. In the next section, we briefly review the related work. In Section~\ref{sec:model}, we present the network model and the classical CSMA policy model. 
In Section~\ref{sec:overview}, we provide an overview of our main results, including the 
description of U-CSMA policy and simulation
results. In Section~\ref{sec:dis:impl}, we provide one example to implement U-CSMA 
policy in a fully distributed and asynchronous manner.
In section~\ref{sec:main}, we provide a formal statement of our analytical results
in this paper. In Section~\ref{sec:dynamics:of:csma}, 
we elaborate on the dynamics of schedules under the classical CSMA policy whose
characterization is required to derive analytical results. In 
Section~\ref{sec:assumptionss}, we formally state two assumptions that allow the formal analysis
developed in this paper. In Section~\ref{sec:flow:control}, we provide details on how to
 use U-CSMA policy jointly with a congestion control algorithm for general topologies. Finally, we conclude the paper in Section~\ref{sec:con}.

\section{Related Work}\label{sec:related}

 In this section, we provide a brief, by no means exhaustive, overview of the work in the area of wireless scheduling that
is closest to ours in this paper. We consider two main classes, i.e., the matching policies
and random access policies.

\underline{\emph{Matching Policies:}}
Maximum Weight Matching (MWM) policy was first proposed in the seminal work in \cite{tassiulas:ephremides:jac:92}. This policy 
is perhaps the first policy that is throughput-optimal in a wide range of settings \cite{tassiulas:ephremides:jac:92,buche:kushber:04,eryilmaz:sirkant:ton:05,neely:jsac05}. MWM policy at any timeslot maximizes a weighted summation of queue-sizes in the network, 
which can be an NP-hard optimization problem \cite{mazumdar:shroff:06}. Despite its complexity,
simulations \cite{gupta:shroff:09} show that MWM policy is close to the optimal in terms of delay for one-hop traffic. 
For multihop traffic, the delay under MWM policy is $O(\frac{L}{\eps})$, and for 
one-hop traffic is order-optimal as $O(\frac{1}{\eps})$, under certain conditions \cite{Le:modiano:09} that hold for geometric
 networks. The delay bound in our paper for one-hop traffic is $O([\frac{1}{\eps}]^{3})$, which 
includes a multiplicative factor of $[\frac{1}{\eps}]^{2}$ as well as $\frac{1}{\eps}$. This factor can be interpreted as 
the scheduling-time needed to find schedules that are $\eps$ close to the optimality. However, we note that 
the delay performance in \cite{gupta:shroff:09} and
the $O(\frac{1}{\eps})$ bound in \cite{Le:modiano:09} are obtained assuming that the NP-hard problem of MWM policy can be solved at every timeslot. 

Greedy Maximal Matching (GMM) policy is a simple and distributed alternative for MWM policy, see e.g., \cite{lin:shroff:ton06,mazumdar:shroff:06}.
 While GMM policy is not 
throughput-optimal in general, a number of local pooling results \cite{dimakis:warland:06, joo:shroff:infocom:08 ,zussman:modiano:08}
indicate that for a noticeable subset of topologies, GMM
policy is indeed throughput-optimal. However, GMM requires message 
passing, and it is an open area to investigate the delay performance of GMM policy. Maximal Matching (MM)
policy is simpler than GMM policy and has order-optimal delay of $O(\frac{1}{\eps})$ for one-hop traffic \cite{neely:info08}.
However, this policy is not throughput-optimal and is guaranteed to stabilize only half of the capacity region. In \cite{sanghavi:skrikant:07}, a matching policy is proposed that can stabilize arbitrarily close
to $100\%$ of the capacity region in expense of increasing an overhead that is constant in network-size. However, this
policy is limited for networks with primary interference. See \cite{yi:proutiere:08} for a comparison of different matching policies.


\underline{\emph{Random Access Policies:}}
Random access policies started with the classical Aloha protocol \cite{abramson:70}, for which an 
optimality result was first established in \cite{haejk:82}. The capacity of random access 
policies under collision detections, acknowledgements, or backoff schemes have been studied in 
\cite{kelly:macphee:87, hastad:96, goldberg:00}. The recent work in \cite{stolyar_dynamic_2008} 
chooses access probabilities in an Aloha-like policy
 based on queue backlogs to achieve the capacity region of slotted Aloha. 
In \cite{lee:chiang:07,rad:09}, distributed protocols are proposed that assign
 access probabilities to maximize a network utility under an Aloha-like protocol.
 Due to their simplicity, Aloha-like protocols
have been also used in mobile networks \cite{Baccelli06analoha}. These protocols however are not 
 throughput-optimal \cite{stolyar_dynamic_2008}.


CSMA policies are a special class of random access policies that assume nodes can sense whether 
their neighbours are transmitting. Performance of these policies as defined in 802.11 standard for a specific network setup is studied in \cite{bianchi:00}. For an interesting but special class of networks with primary interference, it is known that 1) CSMA polices are throughput-optimal
\cite{marbach:CDC:07}, and 2) for a subclass of these networks such as the $n\times n$ switch, the delay to access
 the channel becomes memoryless under CSMA policies, leading to an $O(\frac{1}{\eps})$ (normalized) 
packet delay
\cite{mahdi:peter:al09}. 

 Throughput-optimality of CSMA policies extends to networks with arbitrary interference
 graphs \cite{Jiang:09,ni-2009,rajag:09}. The throughput-optimal CSMA policies in \cite{Jiang:09,ni-2009,rajag:09} are based on a continuous time Markov chain that prevents collisions. 
This is addressed by considering contention resolution \cite{jiang:mobihoc:09,ni-2009}.

Both in \cite{Jiang:09} and \cite{ni-2009}, it is assumed
that there is a time-scale separation and, hence, CSMA dynamics quickly converges to its steady-state
faster than the rate by which queues change over time. The authors of \cite{rajag:09} and later those of \cite{ghaderi:10}
show that as long as attempt rates of nodes change sufficiently slowly, throughput optimality can be achieved. A related work 
\cite{proutiere:10} divides the time axis into frames, and updates  
parameters of CSMA policy only at the beginning of each frame. 
However, delay performance under the above throughput-optimal schemes is not investigated, and
 the upperbound on the delay
inferred from these papers increases with the network-size.

Before concluding this section, we note that there are numerous results that study link starvation under CSMA policies, e.g., see
\cite{garetto:08} and references therein. In particular, the work in \cite{durvy:08} 
shows that in 2D, the phase transition phenomenon makes the CSMA policy \emph{lock into} a 
certain similar set of states for a long time, causing large
packet delays. Using this insight, we propose U-CSMA policy that benefits from a novel unlocking mechanism. 
In cotrast to previous matching or random access policies, U-CSMA is a simple 
and distributed policy that provides high throughput with low delay that features order-optimality.

\section{Network and Classical CSMA Policy Model}\label{sec:model}

In this section, we introduce the network and classical CSMA policy model that we use in this paper.
\subsection{Network Model}\label{sec:netw:mod}
We consider a fixed wireless network consisted of a set $\setN$ of nodes, and
a set $\setL$ of links with cardinality $L.$ We refer to $L$ as the \emph{network size}. A
link $l= (n, m)\in\setL$ indicates that transmitter node $n$ and receiver node $m$ 
are within transmission range of each other and can exchange data packets. Each link $l=(n,m)$
corresponds to a \emph{queue} that is maintained by its transmitter node $n$.

We model the contention between links by an \emph{interference graph} $G(\setL, \setE)$ 
\cite{mazumdar:shroff:06,Jiang:09,ni-2009,rajag:09,proutiere:10}, where $\setL$ is the set of links  and $\setE$ is the set of edges. An edge $e = (l,l') \in \setE$ in the graph $G(\setL, \setE)$ indicates that the two links $l$ and $l'$, $l,l' \in \setL$ interfere with each other. In the following, we will refer to $\setL$ as the node set of the interference graph, and to the set $\setE$ as its edge set. We define a 
\emph{geometric interference graph} \cite{Hunt1998,mazumdar:shroff:06, joo:shroff:infocom:08} to be a 
graph whose vertices can be considered as points on the plane, and where two
vertices are connected
by an edge if and only if the distance between them is less than the \emph{interference range}
 $r$ where $r>0$. We define a \emph{geometric network} as a network with
geometric interference graph. We define a \emph{random geometric network} as
a geometric network for which the vertices of its interference graph are points that are 
distributed according to a uniform stochastic process
over a convex region in the plane.

We define a \emph{valid schedule} to be a subset of links in $\setL$ no two of which interfere with each other. We define a \emph{maximum schedule} to 
be a valid schedule with the largest number of links in $\setL$.
We also define a link to be \emph{active} at time $t$, if the link is transmitting at time $t$. We define a 
\emph{scheduling policy} to be an algorithm, randomized or deterministic, that determines which links are active at any given time.

Throughout the paper, we assume that traffic is one-hop. Let $\lambda_l$ be the packet arrival
 rate for transmission over link $l$, which corresponds to a queue in the network, and let
$$\boldsymbol{\lambda} = (\lambda_{l})_{l \in \setL}$$
 be the arrival rate vector for a given network. We assume that the rate of transmission is the same for all links, and
it takes one unit of time to transmit any one packet.

To characterize the arrival process in further detail, for $t_{2}>t_{1}\geq 0$, let $A_{l}(t_{1},t_{2})$ be the number of packets that arrive for transmission to 
link $l$ in the time interval $(t_{1},t_{2}]$. We assume that the number of packets that arrive in a unit time interval
to any link $l$ is bounded by a constant $A_{max}$, i.e.,
\begin{align}
  A_{l}(t,t+1)\leq A_{max}. \label{ineq:max:arrival:l}
\end{align}
Moreover, we assume for any $\eps>0$, there exists an integer $k_{\eps}\geq 1$ such that 
for $t_{2}-t_{1}\geq k_{\eps}$ and for all $l\in \setL$, we have
\beqa
\left|\ev\bigg[\frac{A_{l}(t_{1},t_{2})}{t_{2}-t_{1}}  \ \Big | \setH_{s}(t_{1}) \bigg]- \lambda_{l} \right|<\eps
\label{eq:def:arrival:exp}
\eeqa
where $\setH_{s}(t_{1})$ is the system history up to and including time $t_{1}$. The above intuitively means that the expected time-average number of packets
that arrive to a link $l$ converges to its arrival rate $\lambda_{l}$.

\subsection{Classical CSMA Policy}\label{sec:mod:ideal:csma}

For our analysis, we define the classical CSMA policy as follows, similar to the ones
 presented in~\cite{durvy:08,Jiang:09,ni-2009,rajag:09}. Given a wireless network with interference
graph $G(\setL, \setE)$, every link $l \in \setL$ independently of others
senses transmissions of any conflicting link in the interference graph $G(\setL, \setE)$, i.e. of any link $l'$ such that the
edge $e = (l,l')$ is contained in the edge set $\setE$. A link $l$ senses the channel as \emph{idle} at time $t$ if all 
of its conflicting (interfering) links are not active and not transmitting at time $t$. 
If link $l$ senses that any of its interfering links is transmitting, then it waits until all of its interfering links become silent. Once this happens, link $l$ sets a backoff timer with a value that is exponentially distributed with mean $1/z_{l}$, $z_{l}>0$, and starts to reduce the backoff timer. If the timer reaches
zero before any of its interfering links start a transmission, then link $l$ starts a transmission. Otherwise, link $l$ simply
waits until all of its interfering links become silent again,  and repeats the above process. 
We define $z_{l}$ to be the transmission \emph{attempt-rate} of link $l$.
We assume that all
transmission times are independently and exponentially distributed with unit mean.

The above models an \emph{idealized CSMA policy} in which 1) any link can always sense
transmissions of all of its interfering links, and 2) there is no hidden-terminal
 problem that can create packet collisions as in \cite{durvy:08,Jiang:09, rajag:09}.
These assumptions can be removed using the methods of \cite{jiang:mobihoc:09,ni-2009}. Hence, we continue assuming that 
the above two assumptions hold.

We characterize a classical CSMA policy by the vector $\vz = (z_l)_{l \in \setL}$
 where $z_l$ is the transmission attempt-rate of link $l$.
Given vector $\vz$, the network dynamics as which links are active over time 
can be represented by a Markov process \cite{Jiang:09}. Using this, we can define 
$\mu_l(\vz)$, $l \in \setL$, as the service rate of link $l$ under $\vz$, i.e., $\mu_l(\vz)$ is
 the fraction of time that link $l$ is active under the CSMA policy $\vz$.

We say that the classical CSMA policy $\vz$ \emph{stabilizes} the
network for a given packet arrival rate vector $\boldsymbol{\lambda}$ if \cite{tassiulas:ephremides:jac:92}
\begin{align}
\lambda_l < \mu_l(\vz), \qquad  l\in\setL . \label{eq:stab}  
\end{align}
This commonly used stability criteria~\cite{tassiulas:ephremides:jac:92} requires that for each link $l \in \setL$, the link service rate $\mu_l(\vz)$ is larger than the arrival rate $\lambda_l$. Given a fixed network, we then define the \emph{achievable rate region} $\setC$ of the classical CSMA policy as 
$$\setC=\{\boldsymbol{\lambda}: \exists \vz \text{ s.t. \eqref{eq:stab} holds.} \},$$
i.e., as the set of all rate vectors
$\boldsymbol{\lambda}$ for which there exists a vector $\vz$ that stabilizes the network for $\boldsymbol{\lambda}$.

It is well-known that the classical CSMA policy is \emph{throughput optimal} \cite{marbach:CDC:07,Jiang:09,ni-2009,rajag:09}, i.e., the set
$\setC$ contains all arrival rate vectors $\boldsymbol{\lambda}$ that are inside the capacity region $\Gamma$, where
$\Gamma$ is the set of all $\boldsymbol{\lambda}$'s
that can be stabilized by any scheduling policy, CSMA or not, including those with the full knowledge of future packet arrivals.


\subsection{Lattice and Torus Interference Graphs with Uniform Attempt  and Packet Arrival Rates}\label{sec:mod:lattice:uaa}

To obtain analytical results, we consider wireless networks with grid-like interference graphs.
In particular, we consider the \emph{lattice interference graph}
 $G_{L}=G_L(\setL,\setE)$ and the \emph{torus interference graph} $\setT_{L}=\setT_{L}(\setL,\setE_{\setT})$.
In both cases, the set $\setL$ is the set of all links where each link $l \in \setL$ can be
 represented by coordinates $(i,j), \ i,j \in \{0,...,n\},$
on the plane.
See Fig.~\ref{fig:lattice_g} for an illustration. Hence, the network-size, i.e., the total number of links, is given by 
$L=(n+1)^{2}.$

It remains to specify which links interfere with each other. For the lattice interference graph $G_{L}$, we assume that there exists an edge $e \in \setE$ 
between any two links $l = (i,j)$ and $l'=(i',j')$, $l,l'\in\setL$, iff link $l$ and link $l'$ differ in exactly one coordinate, i.e., we have that
$$| i - i'| + |j - j'|  =1.$$

For the torus interference graph $\setT_{L}$, the edge set $\setE_{\setT}$ contains all edges 
defined for the lattice interference graph $G_{L}$. In addition,
the set $\setE_{\setT}$ contains an edge between link $l=(i,0)$ and link $l'=(i,n)$, for $0\leq i \leq n$, and also contains an edge 
between link $l=(0,j)$ and link $l'=(n,j)$, for $0\leq j \leq n$. As a result, the torus interference graph $\setT_{L}$ is the same as $G_{L}$ with additional
edges around the boundary of $G_{L}$ so that every link has exactly four interfering links.

\begin{figure}[tp]
\centering
\includegraphics[width=.35\textwidth]{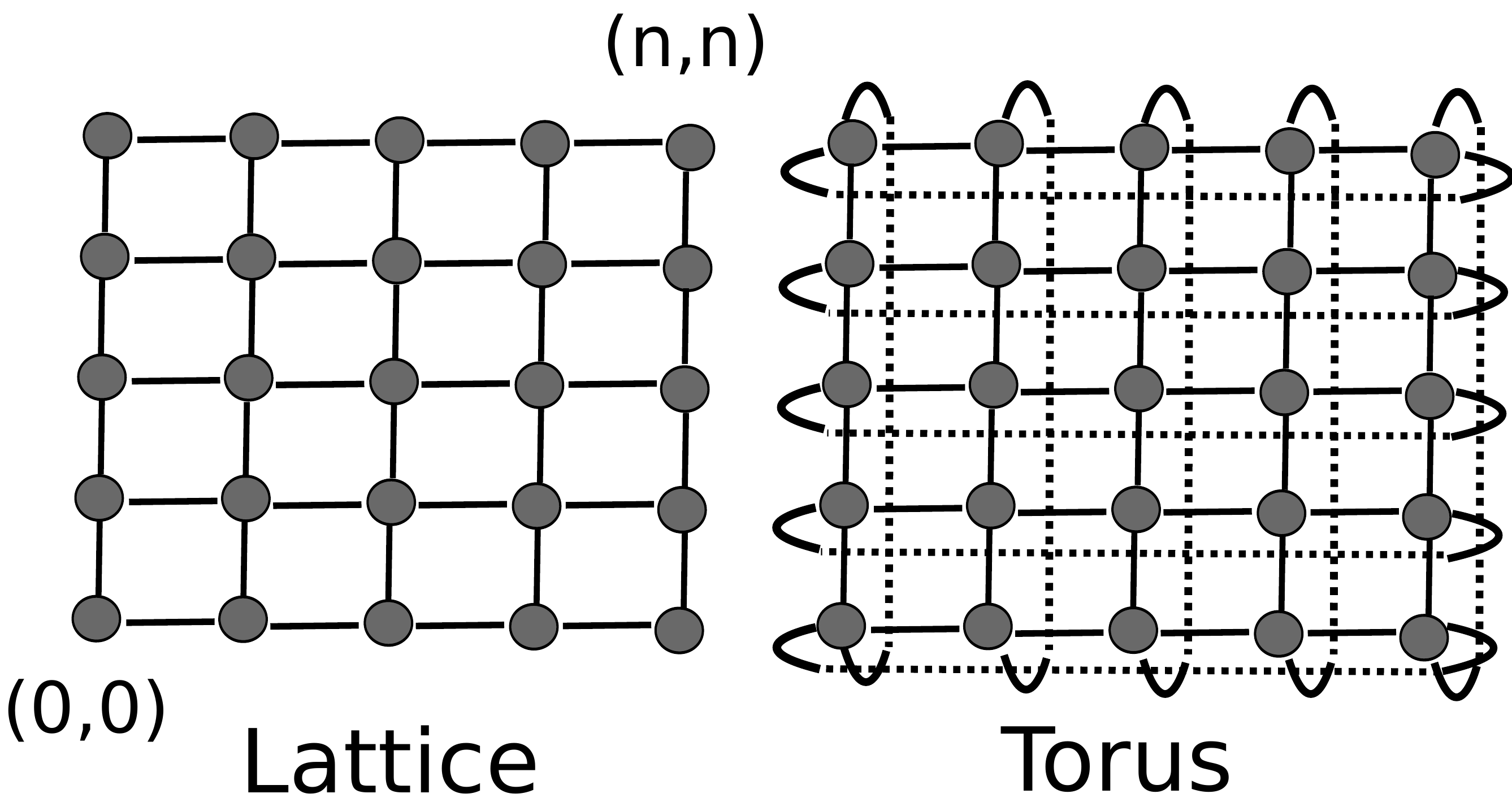}
\caption{Lattice and torus interference graphs. Each dark circle represents a link, and an edge between 
two dark circles shows that their corresponding links interfere with each other.}
\label{fig:lattice_g} 
\vspace{-.2in}
\end{figure}

 Given a lattice or torus interference graph, we define a link $l=(i,j)\in \setL$ as an {\em even link} iff $i+j$ is an even number. We
define $\setL^{(e)}$ as the set of all such even links. Similarly, we define a link $l=(i,j)\in \setL$ as an {\em odd link} iff $i+j$ is an odd
number, and define $\setL^{(o)}$ as the set of all odd links. 


For the lattice and torus interference graphs $G_{L}$ and $\setT_{L}$, we focus on CSMA policies $\{\vz\}$ with uniform transmission attempt-rates so that
$$z_l = z, \qquad  \ l \in \setL,$$
for some $z>0$. In addition, we focus on the case of uniform packet arrival rates, i.e., we let
\beqa
\lambda_{l} = \lambda , \qquad 0< \lambda<\mu_{max}(L), \ l \in \setL.   \label{eq:def:unif:lambda}
\eeqa
where $\mu_{max}(L)$ is the \emph{maximum uniform-throughput}, i.e., the maximum throughput
that can be provided for \emph{all} links by any policy in the network. 
For lattice interference graph $G_{L}$, we have that 
$$\mu_{max}(L)=0.5.$$
This throughput can be achieved, for instance, by alternating between two valid schedules $\setL^{(o)}$ and $\setL^{(e)}$ every unit of time, 
which allows every link to be active half of the time. 
For torus interference graph $\setT_{L}$, due to boundaries being wrapped around, $\setL^{(o)}$ and $\setL^{(e)}$ are not valid
schedules, but we can show that
$$\lim_{L\to \infty} \mu_{max}(L)=0.5.$$

Having defined $\mu_{max}(L)$, for a given lattice or torus interference graph with $L$ links, we define the network \emph{load factor} or simply 
\emph{load} $\rho$ as
\begin{align}
  \rho=\rho(\lambda)=\frac{\lambda}{\mu_{max}(L)}. \label{def:loadf}
\end{align}
We also define $\eps$ to be the distance to maximum load of $\rho=1$:
\begin{align}
  \eps=\eps(\lambda)=1-\rho(\lambda). \label{eq:def:eps:load}
\end{align}

We next provide an overview of our main results. 

\section{Overview of Main Results}\label{sec:overview}
In this section, we provide an overview of our main results. We first investigate the performance of the
classical CSMA policy as defined in Section~\ref{sec:mod:ideal:csma}, and explain why under
 this policy it is impractical to obtain both high throughput and low delay. We explain that
a \emph{locking-in} behaviour leads to an exponentially increasing delay that for the torus interference graph is upperbounded as
$$O\Big( \Big[\frac{1}{\eps}\Big]^{c_{u}L}\Big) $$
where $\eps$ is defined in \eqref{eq:def:eps:load}, and $c_{u}>1$ is a constant. We also explain that for small $\eps$, the delay is
$$\Omega(e^{c_{l}L/(\log{L})^{2}}),$$
for some constant $c_{l}>0$.

We then propose and describe the U-CSMA policy as the main contribution of this paper. We show that for geometric networks, U-CSMA 
policy overcomes the shortcomings of the classical CSMA policy and allows
 to obtain high throughput or utility, arbitrarily close to the optimal, with low packet delay that is order-optimal, i.e., stays bounded as the network-size increases to infinity. In
particular, we analytically show that for large networks with torus interference graph and with uniform packet arrival rates, 
the average delay under U-CSMA policy is upperbounded as
$$O\Big( \Big[\frac{1}{\eps}\Big]^{3}\Big), $$
independent of the network-size $L$. 

Using a simulation study, we show that the same general delay behaviour also holds for the practical case where 1) the arrival rates are determined by
a congestion control algorithm used on top of the U-CSMA policy to maximize a network-wide
utility, and 2) the interference graph is geometric 
(see Section~\ref{sec:netw:mod}) and constructed in a randomized manner.

\subsection {Performance of Classical CSMA Policy}\label{sec:motivation}

In this section, we provide a motivating example to examine the performance of the classical CSMA policy, and
explain why even for simple topologies, this 
policy fails to support a high throughput with low delay. 


Consider a fixed wireless network with torus interference graph, as defined in Section~\ref{sec:mod:lattice:uaa}, having $L$ links and a uniform packet arrival rate $\lambda$ to each link, as defined in \eqref{eq:def:unif:lambda}. It is well-known that \cite{gaunt:65} if all links use the
 same rate $z$, then the following holds for the achieved uniform throughput $\mu(z,L)$:
\begin{align}
   \mu_{max}(L)-\mu(z,L)=\boldsymbol{\Theta}(z^{-1}).\label{eq:un:thput}
\end{align}
This means that to be $\boldsymbol{\Theta}(\eps)$ away from the maximum uniform throughput $\mu_{max}(L)$,
an attempt rate $z$ of order $\frac{1}{\eps}$ is needed.

For the above network, two threshold behaviours exist, as explained in the following.

\emph{\underline{Threshold Behaviour as a Function of Attempt-rate $z$:} }
It is well-known that for a fixed network size $L$, as the attempt rate $z$ increases beyond a threshold, 
the delay of classical CSMA policy on the torus interference graph increases substantially. 
This increase is related to a \emph{phase transition} phenomenon, in terms of the existence of more than one Gibbs measures for the infinite torus \cite{aldous:prob:book}. 


The currently best explicit characterization of the delay of the classical CSMA policy in terms of $z$ shows that the delay is
(see, e.g., the mixing time analysis in \cite{rajag:09})
\begin{align}
  O( z^{c_{u}L}), \label{bound:d:u}
\end{align}
for some constant $c_{u}>1$. While for $z<1$, the above bound can be moderate for a moderate network size $L$, for $z>1$, there will a rapid increase even for moderate values of
$L$. Since by \eqref{eq:un:thput}, a large attempt-rate is needed to support a high throughput, this explains why the classical CSMA policy cannot provide high throughput without incurring a large delay. 

We note that by \eqref{eq:un:thput}, the classical CSMA policy needs to 
use an attempt rate of order $1/\eps$ to support the load $\rho=1-\eps$, which can be used to write the delay bound in \eqref{bound:d:u} as
\begin{align}
  O\Big( \Big[\frac{1}{\eps}\Big]^{c_{u}L} \Big). \label{bound:d:u:eps}
\end{align}

\emph{\underline{Threshold Behaviour as a Function of Network-size $L$:}} 
Depending on the value of a given attempt $z$, as we increase the network size $L$, the delay of the classical CSMA policy 
shows an undesirable threshold behaviour. 

On one hand, there exists a constant $z_{c,1}>0$ such that for all attempt-rates $z<z_{c,1}$, the delay is upperbounded as \cite{Vigoda01anote}
\begin{align}
  O(\log(L)). \label{bound:d:log}
\end{align}
This bound states that for low attempt rates resulting in low uniform-throughputs, the delay increases only logarithmically in the network size $L$.   

On the other hand, there exists a constant $z_{c,2}>0$ such for any attempt-rate $z>z_{c,2}$, the delay is lowerbounded as~\cite{borgs_torpid_1999}
\begin{align}
\Omega(e^{c_{l}L/(\log{L})^{2}}),\label{bound:d:l}  
\end{align}
for some constant $c_{l}>0$. Hence, for large attempt-rates required to support high throughputs, the delay grows exponentially with 
the network-size $L$, which results in a threshold behaviour as $L$ increases. It is this exponential increase in the delay that prevents
the classical CSMA policy to provide high throughput with low packet delay as the network-size $L$ increases.

\emph{\underline{Simulation:}}
To illustrate the threshold behaviours, we have simulated a torus of size $L\in\{100, 400, 1600\}$
under the classical CSMA policy with uniform attempt rate $z$. 
We have assumed i.i.d packet arrivals where 
every unit of time one packet arrives for link $l$, $l\in\setL$, with 
probability $\lambda$, independent of any other packet arrival event. 
For a given network size $L$, to support the uniform arrival rate $\lambda$
 (see Section~\ref{sec:mod:lattice:uaa}) where
 \begin{align}
   \lambda = (1-\epsilon) \mu_{max}(L),  \qquad \epsilon >0, \label{eq:arrival:sim}
 \end{align}
and consequently a load factor (as defined in \eqref{def:loadf}) of $\rho=(1-\eps)$,
we have chosen the attempt rate $z$ such that the resulting uniform throughput $\mu(z,L)$ is given by 
\begin{align}
  \mu(z,L)=\mu_{max}(L)(1-\frac{\eps}{2})>\lambda.
\end{align}

 Fig.~\ref{fig:Q}(a) shows the resulting average queue size per link as a function of $\rho$ in 
linear scale. This figure clearly illustrates the two threshold behaviours.

First, we see that for a given network-size $L$, for a small load $\rho $ less than $0.3$, the queue-sizes are small. However, as the load $\rho$ increases towards $0.5$,
 which requires a larger attempt-rate $z$, the queue-size increases from only few packets to thousands. While the classical CSMA policy is 
throughput-optimal and in principle can support a load $\rho$ close to $1$, we see that in practice, it cannot support loads as low as $0.5$, i.e., it 
cannot reach the $50\%$ utilization without incurring a large delay. For instance, for the $20 \times 20$ torus, the large delay becomes more
 than $1$sec for a packet length of $2346$ bytes and a channel rate of $54$Mbs as in 802.11 standards.

Second, we see that for a given $\rho$, the queue-size shows two different behaviours. If $\rho<0.4 $, the queue-size is small and hardly changes with the 
network size. In contrast, for $\rho>0.4$, the queue-size shows a threshold behaviour and drastically and exponentially increases with the network size. 
For instance, at $\rho=0.44$, the queue-size almost doubles every time that the network size $L$ increases by a factor of $4$.

\begin{figure*}[htp]
\centering
\subfigure[Illustration of the threshold behaviours under classical CSMA policy for torus 
interference graph.]{
\includegraphics[width=.31\textwidth]{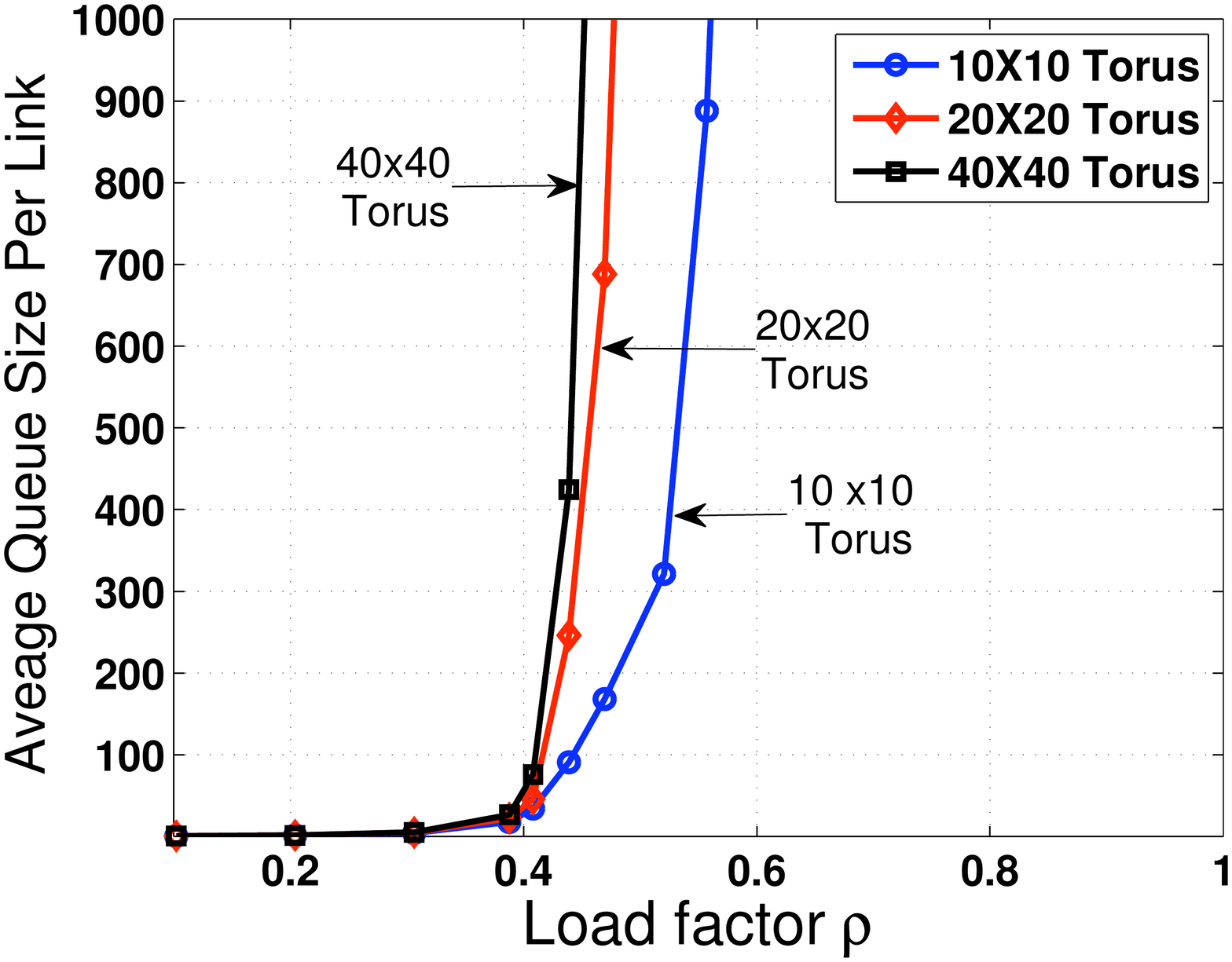}
\label{fig:snapshot1}
}
\hfill
\hspace{-.4 in}
\subfigure[Illustration of elimination of the threshold behaviours under U-CSMA policy
for torus interference graph.]{
\includegraphics[width=.31\textwidth]{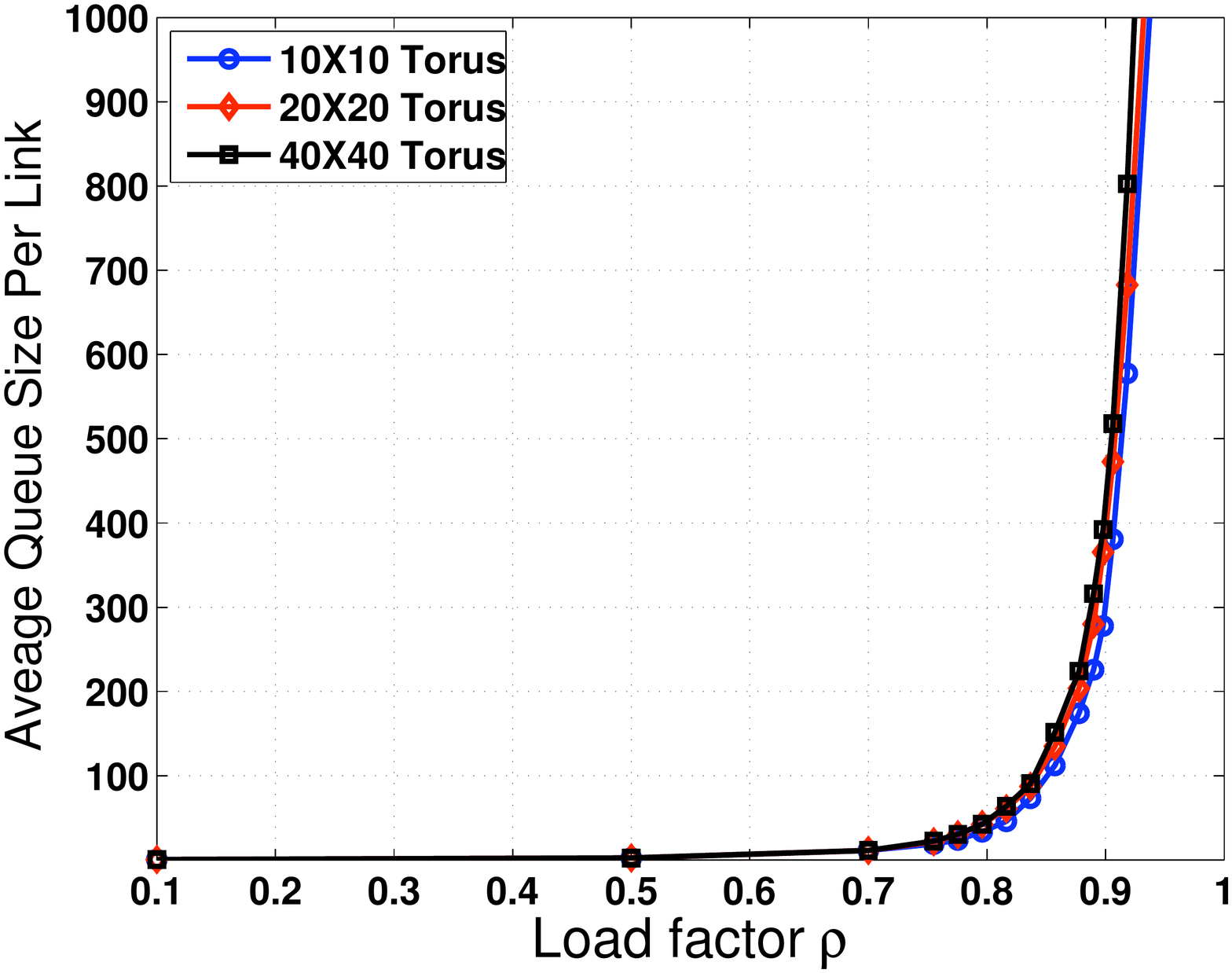}
\label{fig:snapshot2}
}
\hfill \hspace{-.4 in}
\subfigure[Performance under U-CSMA policy combined with congestion control
 in random geometric interference graphs, as a function of utility ratio $\rho_{u}$.]{
\includegraphics[width=.31\textwidth]{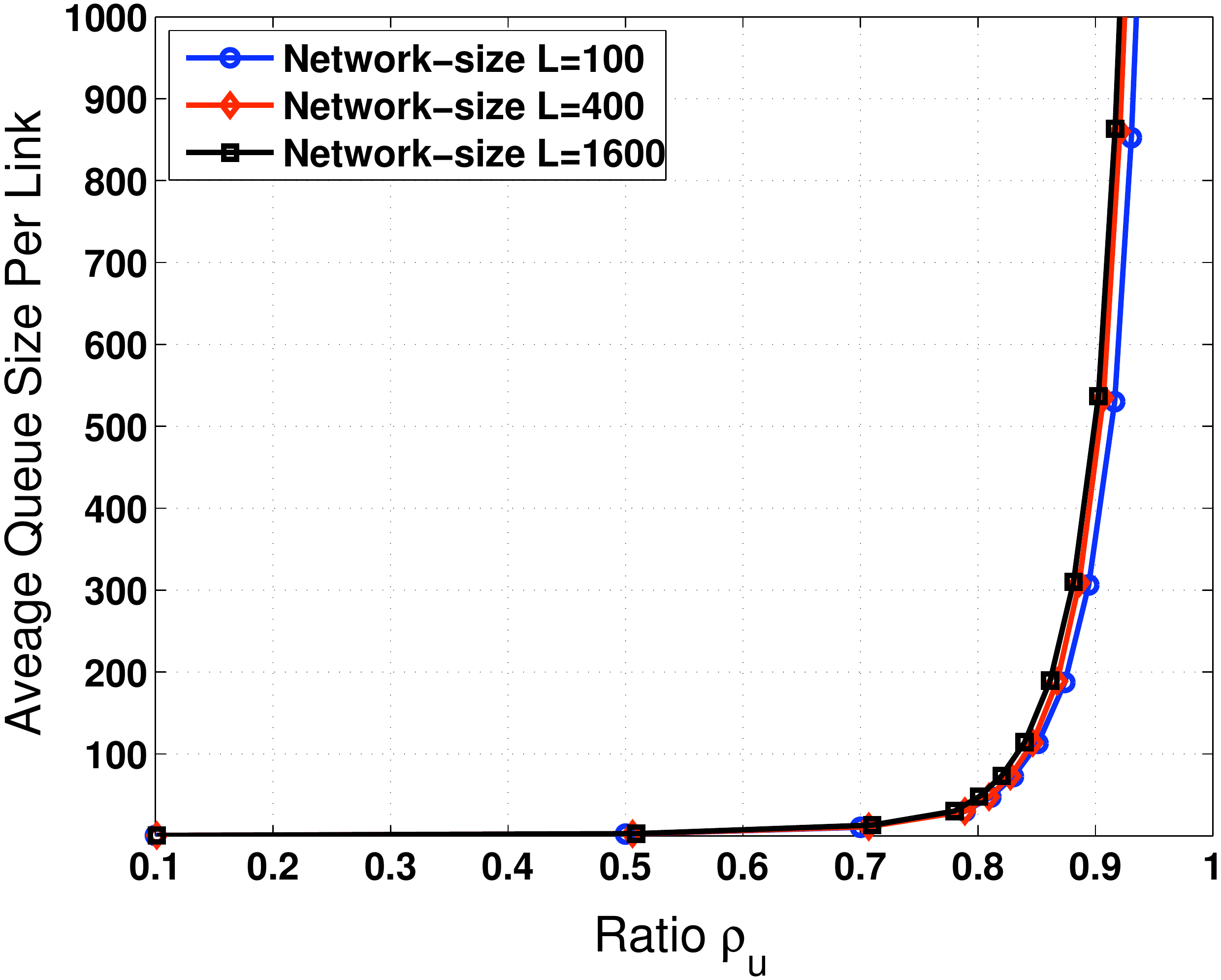}
\label{fig:snapshot3}
}
\caption{Illustration of performance under the classical CSMA policy and U-CSMA policy.}
\label{fig:Q}
\vspace{-.2in}
\end{figure*}


\emph{\underline{Intuition:}}
By \eqref{eq:un:thput}, in order to support a high uniform throughput, the classical CSMA policy needs to use a large attempt rate $z$. 
 For a large attempt rate $z$, the network state will mainly alternate between two types of transmission patterns (valid schedules) 
where either mostly links in the set of even links $\setL^{(e)}$, or links in the set of odd links $\setL^{(o)}$, are active (see Section~\ref{sec:mod:lattice:uaa}). However, as $z$ and $L$ increase, transitions between these two types of patterns occur very infrequently. 
This implies that the classical CSMA policy
 tends to \emph{lock into} one type of transmission patterns
for a very long time before it switches to the other type of patterns \cite{aldous:prob:book}.

This \emph{locking-in} behaviour of the CSMA policy immediately implies that while one type of links, e.g., even links, are active for a long time, the other type of links, e.g., odd links, cannot transmit for a long time. As a result, this locking-in behaviour leads to large queue-sizes and hence a 
large packet delay. 

We next describe U-CSMA policy and provide theoretical and simulation results characterizing its performance.


\subsection{U-CSMA Policy and Its Performance}\label{sec:approach}
The main contribution of this paper is to propose U-CSMA policy that overcomes
the threshold behaviours faced by the classical CSMA policy. As such,
U-CSMA enables to obtain a high throughput with low delay that is order-optimal.




\emph{\underline{U-CSMA Policy:}}
The basic idea behind our proposed U-CSMA policy is very simple. U-CSMA policy uses
a classical CSMA policy $\vz$ as described in Section~\ref{sec:mod:ideal:csma}. However, periodically, i.e., at times 
$$T_{i}=iT , \ i\in\{0,1,2,\cdots\},$$
 U-CSMA policy resets, or \emph{unlocks}, the
transmission pattern of the classical CSMA policy by requiring all links to become silent, and then 
immediately restarts the classical CSMA protocol to operate as usual. In the rest, we refer to parameter $T$ as the \emph{unlocking period}. We note the in the limit of large $T$, U-CSMA policy reduces to the
throughput-optimal classical CSMA policy.

 The intuition behind the above unlocking mechanism is to prevent the threshold behaviour by
 preventing the 
policy from locking into a particular
 transmission pattern for too long. In Section~\ref{sec:dis:impl}, we provide
one approach to implement U-CSMA policy
 in a fully distributed and asynchronous manner.


\underline{\emph{Analytical Results:}}  
In order to characterize the performance of U-CSMA policy, we first need to know how to choose 
the unlocking period $T$. While a smaller $T$ helps employ the unlocking mechanism more frequently
leading to a smaller delay, it
may also prevent the underlying classical CSMA policy used by U-CSMA policy from converging to a maximum schedule that is necessary to obtain a high throughput. Hence, as the first step, we need to study how fast the classical CSMA policy
converges to a maximum schedule.

Our first analytical result (see Proposition~\ref{result:1} in Section~\ref{sec:main}) shows that for the lattice and torus 
interference graphs with uniform attempt rate $z$,
valid schedules under the classical CSMA policy quickly converge to a maximum schedule at a rate that becomes independent of 
network-size $L$ for large networks and attempt-rates. Remarkably, this result shows that the \emph{distance} to the maximum schedules roughly drops 
as $$\frac{1}{\sqrt{t}}.$$

Our second analytical result (see Proposition~\ref{result:2} in Section\ref{sec:result:2}) uses the above convergence result to stabilize networks with torus interference topology and uniform
packet arrival rate $\lambda$. In particular this result shows
 that U-CSMA policy with unlocking period
\begin{align}
  T(\eps)=\boldsymbol{\Theta}\Big( \Big[\frac{1}{\eps}\Big]^{2}\Big),
\end{align}
and with large uniform attempt rate\footnote{ Large attempt rates can
be implemented using Glauber dynamics as in \cite{rajag:09,ni-2009}. } $z$ stabilizes the load $\rho=(1-\eps)$ for large networks with torus interference graph. Hence,
by the above choice for the unlocking period, U-CSMA policy stabilizes queues in the network, all of
which have packet arrival rate of $\lambda= (1-\eps)\mu_{max}(L)$.

Further, this result shows that by the above choice for the unlocking period $T(\eps)$, the average queue-size per link and, hence, average delay become order-optimal and 
independent of the network size $L$ in the sense that for large $L$ and attempt-rate $z$, they are upperbounded as
\begin{align}
  O\Big( \Big[\frac{1}{\eps}\Big]^{3}\Big). \label{eq:bound:d:np}
\end{align}

Comparing the above delay bound with the ones in \eqref{bound:d:u:eps} and \eqref{bound:d:l} for the classical CSMA policy, we see that
U-CSMA policy does not suffer from the threshold behaviours. Specifically, we see that
as a function of $1/\eps$, the queue-size under U-CSMA policy increases at most with exponent $3$ as opposed to the exponent
$L$ under classical CSMA policy, as suggested by the bound in \eqref{bound:d:u:eps}. Moreover, U-CSMA policy has changed a queue-size that
exponentially grows with the network size $L$ (see \eqref{bound:d:l}) to a queue-size that does not depend on the network size $L$.

\underline{\emph{Simulation Results:}} To illustrate the performance of the U-CSMA policy and compare it
with the analytical results, we have simulated a torus
of size $L\in\{100, 400, 1600\}$ under the U-CSMA policy. We have assumed i.i.d. packet arrivals  
  where every unit of time, one packet arrives to link $l$, $l\in\setL$, 
 with probability $\lambda$ independent of any other arrival event. 
We have set the uniform attempt rate at $z=50$, and for a given uniform arrival 
rate $$\lambda =(1-\eps) \mu_{max}(L),$$ 
or load $\rho=1-\eps$, we have chosen the unlocking period $T$ as 
\begin{align}
T= \frac{1.2}{\eps^{2}} .
\end{align}

Fig.~\ref{fig:Q}(b) shows the resulting average queue-sizes as a function of load $\rho$. We make the following two observations. 
First, comparing Fig.~\ref{fig:Q}(b) with Fig.~\ref{fig:Q}(a), we see that while
 the classical CSMA ``hits the wall'' and its queue-size becomes on the order of thousands of packets before reaching
 a low load of $\rho=0.5$, 
the U-CSMA policy can indeed get much closer to the maximum load of $1$. In practical terms, for a packet length of 
$2346$ bytes and a channel rate of $54$Mbs as in 802.11 standards, the average packet delay under U-CSMA policy becomes $30$ms and $90$ms for
 $80\%$ and $85\%$ channel utilizations, respectively, while the average delay under the classical CSMA becomes more than $1$sec before even reaching the $50\%$ 
utilization. In addition, replotting the queue-size as a function of $\eps=1-\rho$ in $\log$-$\log$ scale (see 
Fig.~\ref{fig:torus_q_log}),
 we see that the average exponent by which queue-size increases as a function of $1/\eps $ is 3.02, which closely matches 
the exponent 3 as predicted by the analysis in \eqref{eq:bound:d:np}.

Second and as remarkably predicted by the analysis, the average queue-size does not change significantly with the network size. In fact, for $20\times20$ and $40\times 40$ toruses the average queue-sizes are hardly distinguishable. This confirms that 1) U-CSMA eliminates the threshold behaviours that exist for the classical CSMA policy, and 2) the delay under U-CSMA 
is order-optimal in that it stays bounded as the network size increases.

To investigate whether the insight gained through the analysis for the torus interference graph carries over to general network setups, we have 
simulated a random geometric interference graph \cite{Hunt1998,mazumdar:shroff:06, joo:shroff:infocom:08}, 
see Section~\ref{sec:netw:mod}, in which $L\in\{100, 400, 1600\}$ links are randomly distributed over a square area of $10\times 10$, $20\times 20$, and
$40\times 40$, respectively. We have chosen the interference range $r$ so that every link on the average interferes with six other links. As in \cite{neely:modiano:infocom:05,lin:shroff:ton06, rad:09}, we have implemented a 
congestion control algorithm to tune the arrival rate to each link so 
that a \emph{network-wide} logarithmic utility
 function $U_{net}$ is maximized. This algorithm operates on top of the U-CSMA policy (see 
Section\ref{sec:flow:control} for further details).

Fig.~\ref{fig:Q}(c) plots the average queue-size as a function of $\rho_{u}$ where 
$$\rho_{u}=\frac{U_{net}}{U_{opt}},$$
i.e., $\rho_{u}$ is the ratio of the achieved network-wide utility to the optimum maximal utility $U_{opt}$. 
Remarkably, the delay behaviour is similar to the one illustrated by Fig.~\ref{fig:Q}(b).

The main observation here is that the that the insight gained through the analysis for the torus interference graph also holds for
 the general case considered here. First, we observe that even in random topologies
under a congestion control algorithm, we can use U-CSMA policy to assign arrival rates
 closely to the optimal without incurring a large delay. For instance,
for a packet length of $2346$ bytes and a channel rate of $54$Mbs, the delay becomes 40ms to get to $80\%$ of optimality. Interestingly, the exponent by which queue-size increases as a function 
of $1/(1-\rho_u)$ approaches 3, the same exponent in the delay bound of torus graph in \eqref{eq:bound:d:np} (see Fig.\ref{fig:rand_q_log} for the corresponding $\log$-$\log$ plot).

Second, we observe
that the average queue-size and hence the delay slightly change with the network-size. This means order-optimality of delay is preserved, and therefore, we can use U-CSMA policy jointly with congestion control to assign arrival rates 
close to the optimal with low packet delay in arbitrarily large networks. 

  In the next section, we provide one example to implement U-CSMA policy in
a fully distributed and asynchronous manner.

\section{Distributed Implementation}\label{sec:dis:impl} 

In this section, we provide an algorithm to implement the unlocking mechanism 
of U-CSMA policy, as described in Section~\ref{sec:approach}, in a fully distributed
and asynchronous manner. 
 We assume links can send busy tones to \emph{initiate}, or \emph{relay}, the unlocking process. 
Busy tones as opposed to control packets are propagated much 
faster and their detection is easier. 

Fix the unlocking period at $T$.
Let $\Delta_{b}\ll T$ be the maximum delay from the time a link broadcasts 
one busy tone until all its interfering links detect the busy tone.
Further, every link keeps track of the last time $t_{last}^{(1)}$ and the second last time $t_{last}^{(2)}$ 
that either initiated 
or relayed a busy tone. In addition, every link maintains a counter that determines the
next time after $t_{last}^{(1)}$ that it may send a busy tone to locally initiate the unlocking process. The value of the counter
is reset to $T+t_{b}$ at times $\{t_{last}^{(1)}\}$, where $t_{b}$ is a r.v. Links choose $t_{b}$ as follows to maintain the length of 
periods close to 
$T$. If $t_{last}^{(1)}-t_{last}^{(2)}\leq T$, $t_{b}$ is chosen 
uniformly distributed from $[0,2\Delta]$, otherwise from $[-2\Delta,0]$, where $\Delta\geq\Delta_{b}$. A link that joins the network
for the first time at time $t$, sets its $t_{last}^{(1)}=t$ and its counter value to $T$.

Every link $l$ implements the following. If at time $t$, link $l$ detects a busy tone or its
counter reaches zero, it broadcasts one busy tone, to all of its interfering links, only if it has not done so in the 
last $0.5 T$ time-units. This ensures that busy tones will not go back to the link where initiated them. 
After broadcasting a busy tone, the link $l$ stops transmitting and can start competing for the channel, using the classical CSMA policy as usual,
only after a time that is uniformly distributed in $[0, 2\Delta]$. This ensures that 
one link does not always transmit first. 

It is clear that for a fixed $T$ and $\Delta$, as the delay $\Delta_{b}$ approaches zero, i.e., when busy tones
propagate very fast, the distributed approach 
converges to the ideal unlocking mechanism. For large $\Delta_{b}$, however, transmission patterns are unlocked 
locally. Nevertheless, our simulations for both the torus and the random geometric networks, as
simulated in 
 Section~\ref{sec:overview}, show that the changes in the queue size and utility are less than $1.5 \%$ when $T=154$, $\Delta_{b} \in [0.1,2]$,
and $\Delta=\Delta_{b}$, all in units of time. Hence, with moderate values of busy tone delay, the distributed unlocking mechanism performs close to its ideal.

 In the next section, we provide formal statements of the analytical results in this
paper.

\section{Performance Analysis}\label{sec:main}

In this section, we formally state the analytical results developed in this paper for lattice
and torus interference graphs.
These results characterize the rate by which the schedule under classical CSMA policy converges 
to maximum schedules, and characterize the delay-throughput tradeoff under U-CSMA policy.
These results use two assumptions that are formally stated in Section~\ref{sec:assumptionss}.

Even by making these assumptions, the analysis of CSMA convergence is by no means trivial. 
This analysis requires techniques often used to develop mean-field results \cite{Leboudec:08}, characterizing the properties of ODEs, and also large deviation results. Simulation results presented in Section~\ref{sec:overview} and this section verify that
these assumptions indeed lead to correct qualitative results, not only for lattice and 
torus topologies, but also for random geometric networks under congestion control.



\subsection{Convergence to Maximum Schedules Under Classical CSMA Policy}\label{sec:result:1}
Our first result characterizes the rate by which the schedule under classical CSMA policy converges to maximum
 schedules. We consider the lattice or torus interference graph with $L$ links, and a classical 
CSMA policy with uniform attempts rate $z$, as described in Section~\ref{sec:model}.

To state our first result, we use the following notation.
Let $\theta_{L}(t,z)$ be the \emph{density}, i.e., fraction, of links that are active at time $t$, $t > 0$. Hence, if
 $N_{a}(t,z)$ is the total number of links that are active at the time $t$ under a classical 
CSMA policy with uniform attempt rate $z$, then $\theta_{L}(t,z)$ is given by
$$\theta_{L}(t,z) = \frac{N_a(t,z)}{L}$$
We assume that the system is idle at time $t=0$ such that
$$ \theta_{L}(0,z) = 0, \qquad z>0.$$

Let $\delta_{L}(t,z)$ be 
\begin{align}
  \delta_{L}(t,z)=0.5-\theta_{L}(t,z). \label{eq:def:delta:l:n}
\end{align} 
Since $0.5$ is the fraction of links that can be active under a maximum schedule
in lattice or torus interference graphs in the limit of large $L$, we see that 
$\delta_{L}(t,z)$ can represent the \emph{distance} between the schedule at time 
$t$ and the limit maximum schedules.

Proposition~\ref{result:1} characterizes how fast the distance $ \delta_{L}(t,z)$
 approaches $0$, or in other words, how fast the distance to maximum schedules drops to $0$, in 
the limit of large $L$ and $z$. 
\begin{proposition}\label{result:1}
Suppose the interference graph is given by the lattice (or torus) interference graph $G_{L}$ (or $\setT_{L}$).
 Under Assumptions~\ref{assum:1}-\ref{assum:2} for $G_L$ (or $\setT_L$), there
exists a positive constant $C_{1}$, independent of $z$ and $L$, such that for any $\tau>0$, we have that
\beqa
 \liminf_{z\to \infty} \ \liminf_{L\to \infty}  P\left[ \sup_{t \in (0,\tau]} \left [ \delta_{L}(t,z) - \frac{C_{1}}{\sqrt{t}} \right ] \leq  0 \right]=1. \nonumber
\eeqa
\end{proposition}
\begin{proof}
  Proof is provided in Appendix~\ref{sec:proof:t1}.
\end{proof}
Proposition~\ref{result:1} states that for every finite time-horizon $(0,\tau]$, with probability 
approaching one as first the network size $L$ approaches infinity and then $z$ approaches infinity, the distance  $ \delta_{L}(t,z)$ between  $\theta_{L}(t,z)$ and the maximum fraction of active links $0.5$ converges to $0$ and drops as $O(\frac{1}{\sqrt{t}})$ for $ t \in (0,\tau]$.

The above convergence has two important implications. First, under the classical CSMA policy,
the distance to maximum schedules asymptotically drops as $O(\frac{1}{\sqrt{t}})$, only depending on time $t$. Second, as the 
$O(\frac{1}{\sqrt{t}})$ bound does not depend on the network-size $L$ or attempt-rate $z$, the convergence is not negatively
affected by a large $L$ or large $z$.
This is in a stark contrast to the results obtained for for the mixing time of CSMA policies, i.e., the rate at which CSMA policies reach their steady-state, which increases with attempt-rate $z$ and
 can be exponential in the network size $L$~\cite{borgs_torpid_1999}.

To illustrate the convergence behaviour, we have simulated
a $n\times n $ lattice, $n\in\{20,30,50,100\}$, under the classical CSMA policy with $z=100$.
Fig.~\ref{fig:rho_t:400:10000} shows $\theta_{L}(t,z) $, averaged over $20$ simulation runs, for each
lattice. As predicted by Proposition~\ref{result:1},
 convergence behaviour becomes independent of the network size for large lattices. In fact, 
$\theta_{L}(t,z)$ for the largest lattice can be very closely fitted by a curve of the form 
$0.1(1+0.4t)^{-0.5}$, which drops to zero as $1/\sqrt{t}$, as stated by the proposition.

\begin{figure}[tp]
\centering
\includegraphics[width=.35\textwidth]{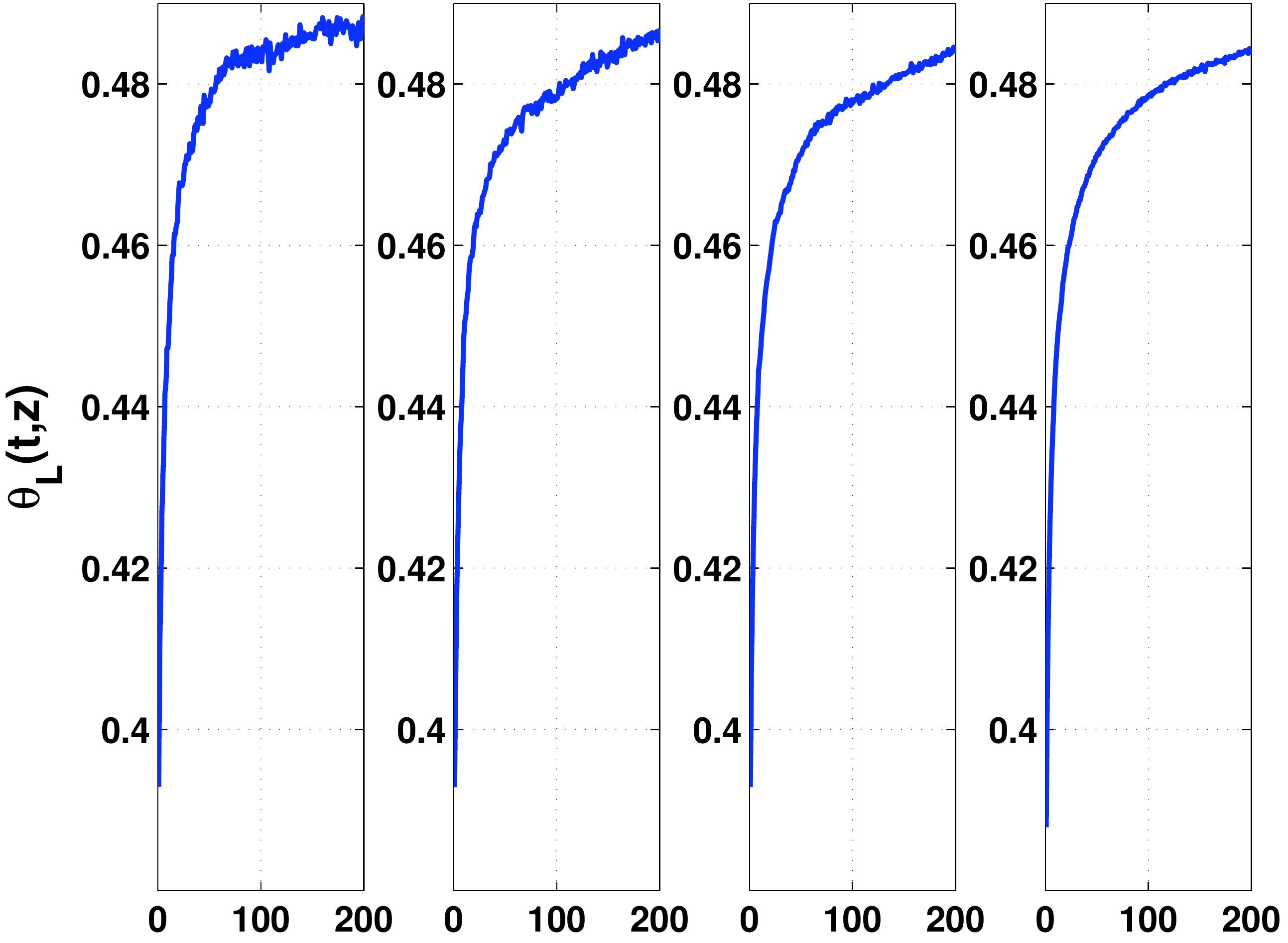}
\caption{Convergence of $\theta_L(t,z)$ for lattices of size $20 \times 20$, $30 \times 30$, $50 \times 50$, and $100 \times 100$, 
from left to right, respectively ($z=100$).}
\label{fig:rho_t:400:10000}
\end{figure}

\subsection{Delay-Throughput Trade-off under U-CSMA Policy}\label{sec:result:2}

Proposition~\ref{result:1} states that under the classical CSMA policy, the distance to maximum
 schedules converges to zero at a rate independent of the network
size in the limit of large network sizes and attempt rates. Our second result stated in Proposition~\ref{result:2} characterizes the delay-throughput trade-off under U-CSMA policy for the 
torus interference graph with uniform attempt-rate $z$ (see Section~\ref{sec:mod:lattice:uaa}). Intuitively, Proposition~\ref{result:2} states that in large networks, the delay-throughput
 trade-off under U-CSMA policy does not depend on the network size $L$.

In order to formally state the throughput-delay trade-off for any given link in the network, irrespective of its position, 
we consider the torus interference graph $\setT_{L}$ 
(see Section~\ref{sec:mod:lattice:uaa})
 instead of the lattice interference graph $G_{L}$. For the lattice interference graph and similar topologies, it is well
known that due to boundary effects, the throughput achieved by links in the network is not uniform over all links in the network 
when a uniform attempt rate $z$ is used \cite{durvy:08}. The torus interference graph is symmetric with respect to link positions, and as a result
boundary effects do not exist. While we develop the analysis for the torus interference graph, the general insight gained through the analysis carries over to more general settings, as discussed in 
Section~\ref{sec:approach}

To state Proposition~\ref{result:2}, we introduce several definitions. We first note that by
 Proposition~\ref{result:1}, for the torus interference graph $\setT_{L}$ and a given $\tau>0$, we can
 define a non-negative function $\eps_{p}(L,z,\tau)$ such that we have
 \begin{align}
   P\left[ \sup_{t \in (0,\tau]} \left [ \delta_{L}(t,z) - \frac{C_{1}}{\sqrt{t}} \right ] \leq  0 \right]\geq 1-\eps_{p}(L,z,\tau),
 \end{align}
and 
\vspace{-.1in}
\begin{align}
  \limsup_{z\to \infty} \limsup_{L\to \infty} \eps_{p}(L,z,\tau)=0.
\end{align}

For a given $\eps'>0$, the above limit allows us to define $z(\eps',\tau)$ and $L(z,\eps',\tau)$ such that for 
$z>z(\eps',\tau)$ and $L>L(z,\eps',\tau)$, we have 
\vspace{-.1in}
\begin{align}
  \eps_{p}(L,z,\tau)< \frac{1}{2} \eps'.
\end{align}

Furthermore, for a given uniform packet arrival-rate $\lambda$, $0 < \lambda < 0.5$, and a given 
uniform attempt-rate $z$ (see Section~\ref{sec:mod:lattice:uaa}), we define $Q_{l}(t,z,\lambda)$ 
as the queue size of link $l$ at time $t$.

\begin{figure}[tp]
\centering
\includegraphics[width=.35\textwidth]{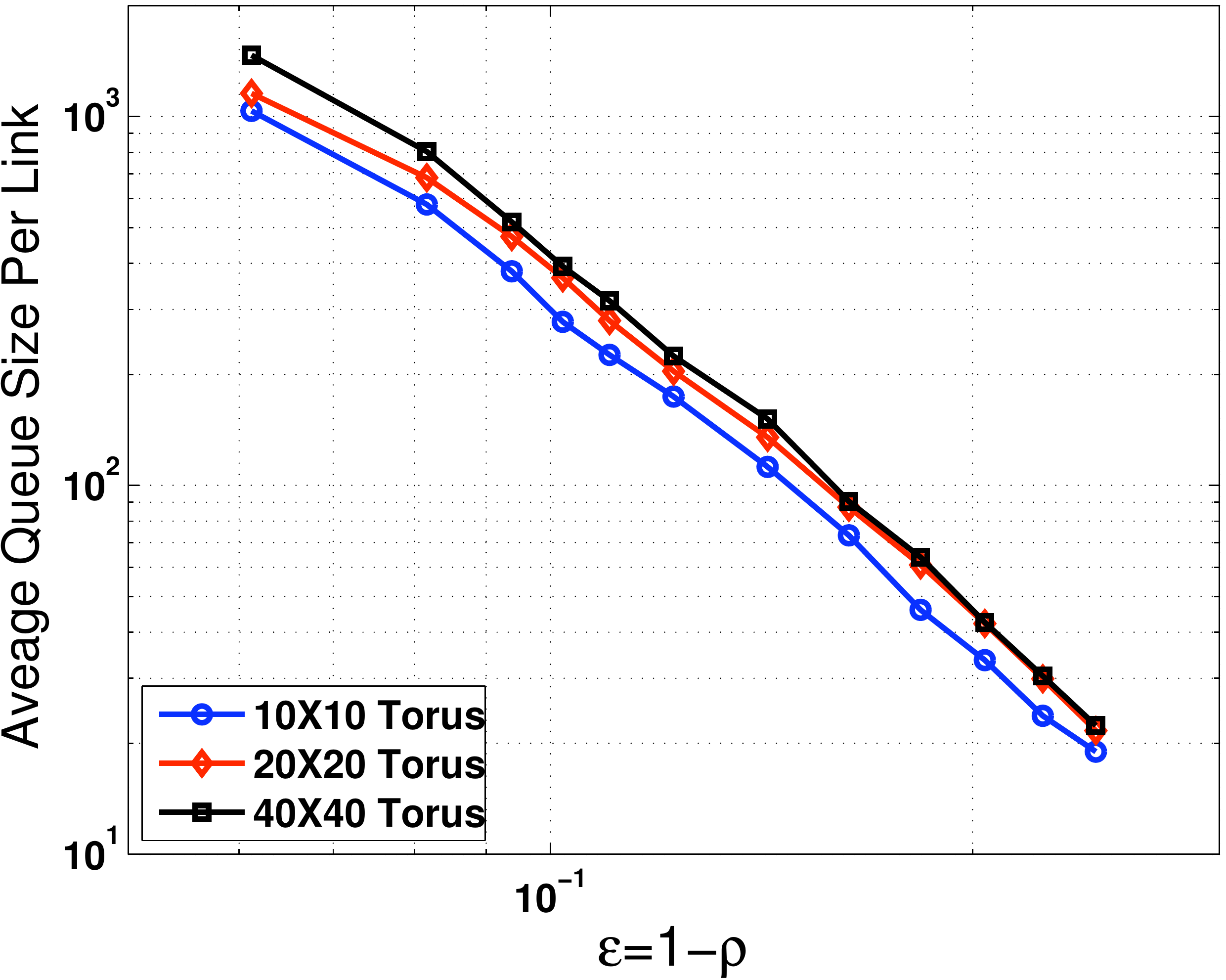}
\caption{$\log$-$\log$ plot of average queue-size
 as a function of distance $\eps$ to the maximum load $\rho=1$, under U-CSMA policy in torus interference graph.}
\label{fig:torus_q_log} 
\vspace{-.2in}
\end{figure}

Using the above definitions, Proposition~\ref{result:2} is given as follows. 


\begin{proposition}\label{result:2}
Consider the torus interference graph $\setT_{L}$, 
and suppose Assumptions~\ref{assum:1}-\ref{assum:2} hold for $\setT_{L}$.
Let the uniform packet arrival rate to each link be $\lambda$, corresponding to load
$\rho(\lambda)$. 
Let the unlocking period $T(\lambda)$ used by the U-CSMA policy be
$$T(\lambda) = \frac{(16C_{1})^{2}}{\eps^{2}}=\boldsymbol{\Theta} \left(\frac{1}{\eps^{2}} \right)$$
where 
\vspace{-.1in}
$$\eps=\eps(\lambda)=1-\rho(\lambda),$$
 and $C_{1}$ is a constant given in Proposition~\ref{result:1}.
Then, the there exists a positive constant $C_2$ such that for $z>z(\eps,T(\lambda) )$ and $L>L(z,\eps, T(\lambda)) $, the time average of the queue size for any link $l$ in
$ \setT_{L}$ satisfies the following under U-CSMA policy with the unlocking period $T(\lambda)$:
\begin{align}
\limsup_{t\to \infty } \ev\left[\frac{1}{t} \int_{0}^{t} Q_{l}(t,z,\lambda) \ dt\right] 
<  \frac{C_{2}k_{\frac{\eps}{16}}}{\eps^{3}}=\boldsymbol{\Theta} \left(\frac{k_{\frac{\eps}{16}}}{\eps^{3}} \right)
\nonumber 
\end{align}
where $k_{\eps}$ is defined by the arrival process in Section~\ref{sec:mod:lattice:uaa}.
\end{proposition}
\begin{proof}
  Proof is provided in Appendix~\ref{sec:proof:t2}.
\end{proof}

Proposition~\ref{result:2} states that in order to get $\eps$ close to
 the maximum load of $\rho=1$, the expected time average of
 any queue-size in the network becomes only 
 $O\big(k_{\frac{\eps}{16}} / \eps^{3}\big)$, independent of network-size $L$ for large $L$. This is achieved by choosing the unlocking period $T$ to be
on the order of $\frac{1}{\eps^{2}}$. By Little's Theorem, we have that the average delay is also $O\big(k_{\frac{\eps}{16}}/ \eps^{3}\big)$. 
We note that
$k_{\eps}$ represents the rate by which the arrival process converges to it expected value (in the sense of \eqref{eq:def:arrival:exp}). Hence, we expect this rate 
to appear in the average queue size and the average delay of any link. In the case where every unit of time, packets arrive to each queue according
 to an i.i.d process, we have $k_{\eps}=1$. For such a case, we have that 
the average delay for any given link is
\begin{align}
  O\Big(\Big[\frac{1}{\eps} \Big]^{3}\Big). \label{bound:d:ucsma:ind}
\end{align}

 Quite surprisingly, the above delay-bound and the resulting throughput-delay trade-off are valid for arbitrarily large 
torus networks as long as $z>z(\eps, T(\lambda))$. Moreover, since $C_{2}$ in the proposition 
is a constant, the delay-bound does not depend on the network-size $L$, and hence, we have an order-optimal average delay. This makes
the delay-throughput trade-off under U-CSMA policy independent of the network-size $L$ for large 
$L$.
As a result, U-CSMA policy, which benefits from an unlocking mechanism,
can indeed provide high throughput with low delay for arbitrarily large 
torus networks.

To investigate the accuracy of the delay-bound in \eqref{bound:d:ucsma:ind}, 
we have replotted the queue-size as a function of $\eps$ under the simulation setup of Section~\ref{sec:approach}. The figure
shows that the queue-size increases with (average) slop 3.02 in $\log$-$\log$ scale, which, 
as expected, is
close to the exponent 3 given in \eqref{bound:d:ucsma:ind}\footnote{As mentioned in Section~\ref{sec:overview}, we have also observed an exponent
close to 3 when queue-size is plotted against $1-\rho_{u}$ for the case where U-CSMA is combined with 
congestion control in random geometric topologies, implying that a variant of Proposition~\ref{result:2} should likely be true for this case. }.

\begin{figure*}[htp]
\centering
\subfigure[Snapshot at time $t=5$, i.e., after five packet-transmission times.]{
\includegraphics[width=.32\textwidth]{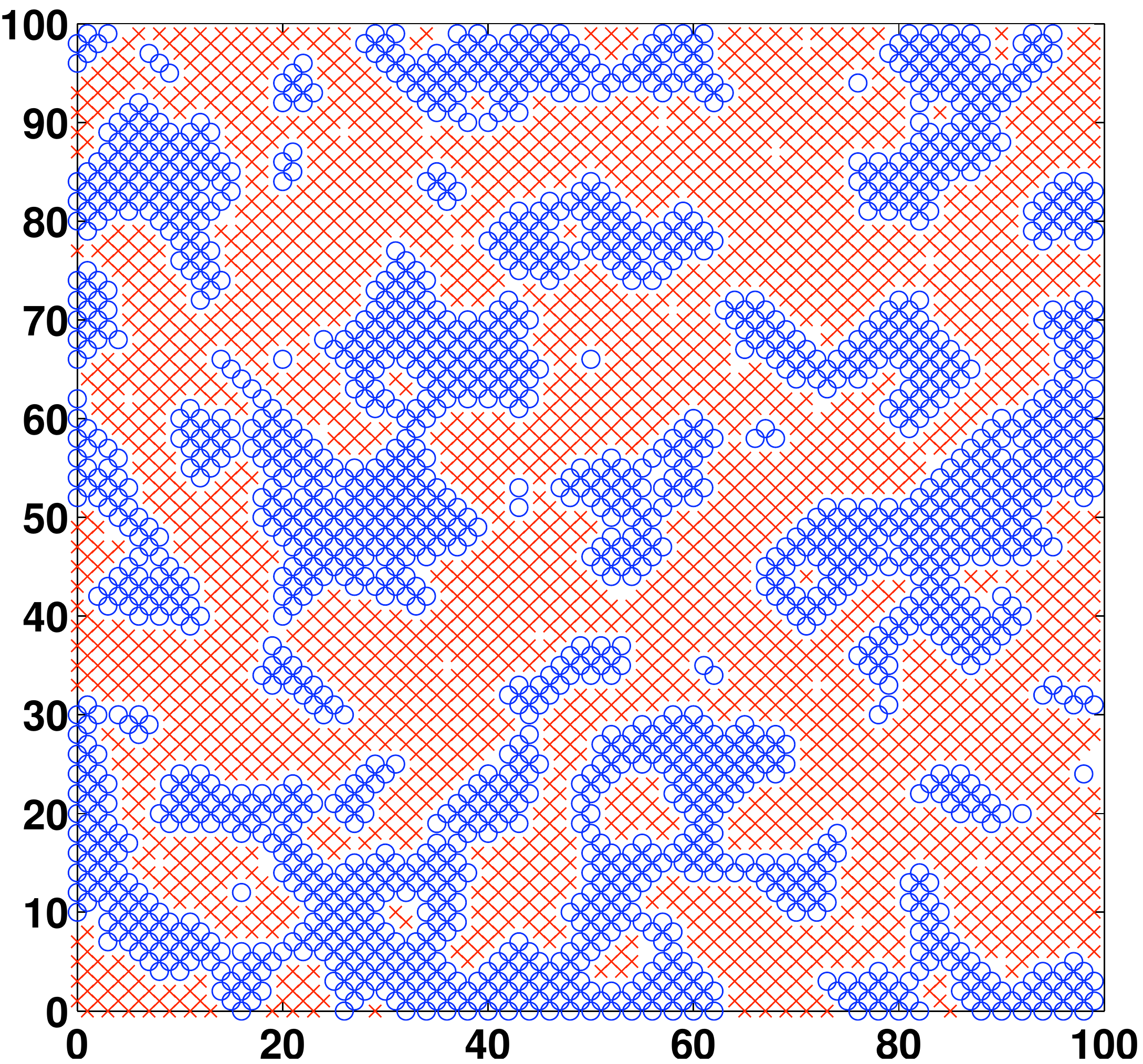}
\label{fig:snapshot1}
}
\hfill
\hspace{-.4 in}
\subfigure[Snapshot at time $t=50$, i.e., after fifty packet-transmission times.]{
\includegraphics[width=.32\textwidth]{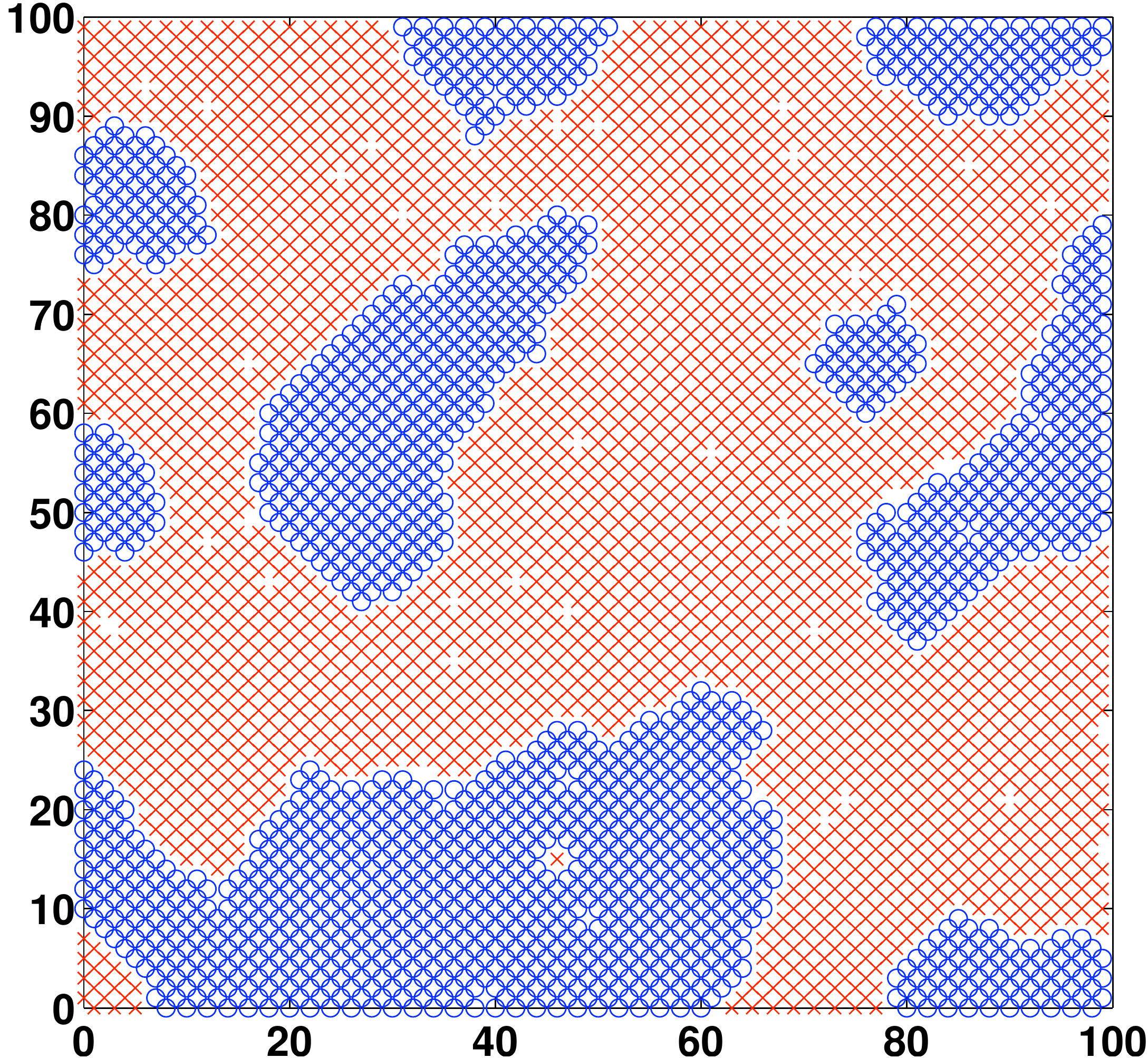}
\label{fig:snapshot2}
}
\hfill \hspace{-.4 in}
\subfigure[Snapshot at time $t=200$, i.e. after two hundreds packet-transmission times.]{
\includegraphics[width=.32\textwidth]{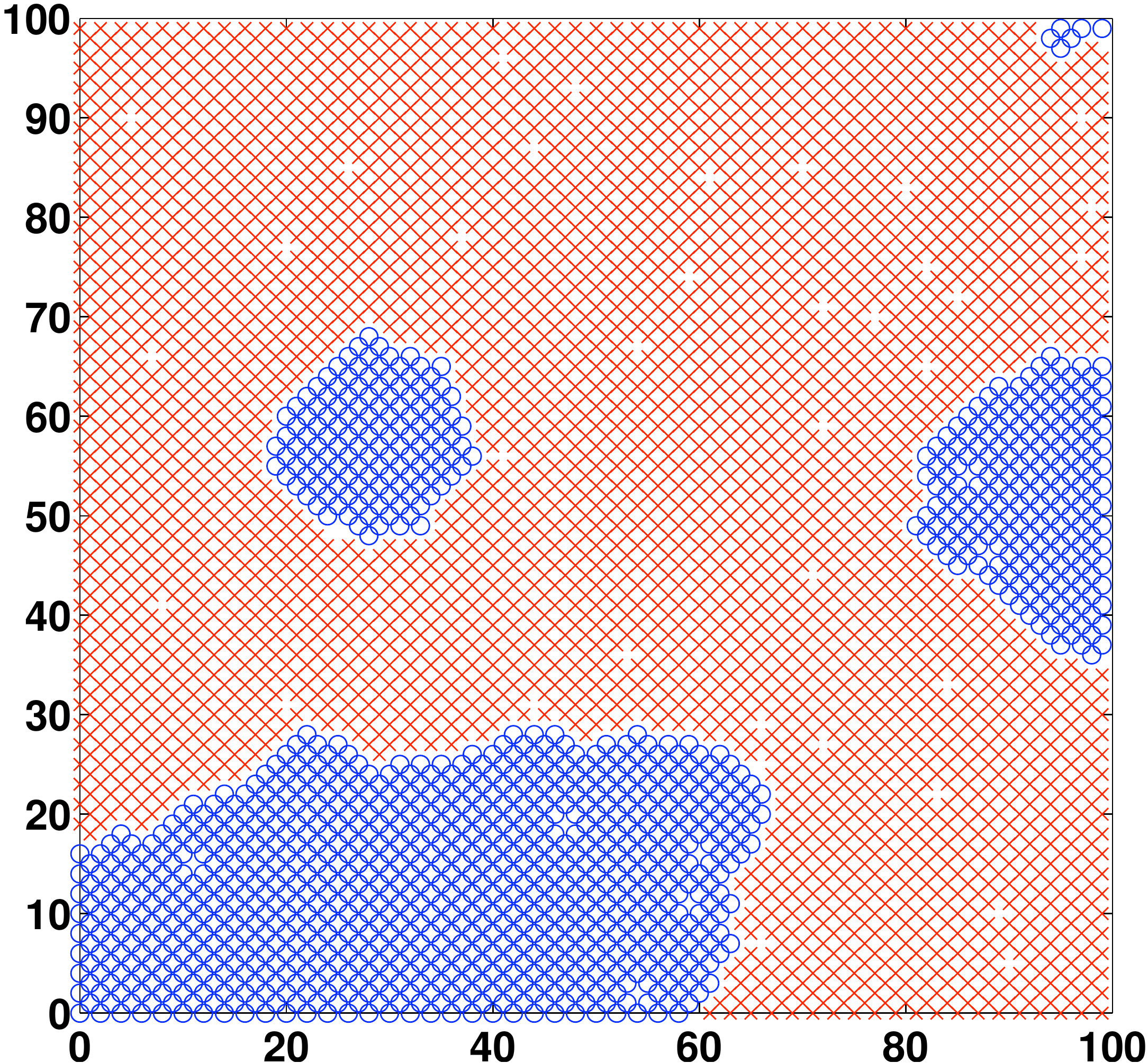}
\label{fig:snapshot3}
}
\caption{Snapshots of the schedules under the classical CSMA policy for the $100\times 100$ 
lattice. A red cross represents an odd active link and a blue circle represents an even 
active link.}
\end{figure*}

\subsection{Discussion}
Note that in Proposition~\ref{result:1} and~\ref{result:2}, we first let $L$ approach infinity and then let $z$ approach infinity.
We believe that the same result holds if one changes the order of limits. For instance, in Section~\ref{sec:flow:control}, we
consider a fix network where by using different values of the unlocking period, we effectively increase the attempt rate. The obtained results, as
illustrated in Fig.~\ref{fig:rand_q_log}, closely match of those if we could change the order of limits. We have left a formal proof
of this property for the future research.

\section{Dynamics of Schedules Under Classical CSMA Policy}\label{sec:dynamics:of:csma}
Having provided a formal statement of our main results in
 Section~\ref{sec:main}, we now turn our attention to the dynamics of schedules under
the classical CSMA policy whose characterization is the first step for
 the derivation of Proposition~\ref{result:1}~and~\ref{result:2}.

At the heart of the proofs for Proposition~\ref{result:1}~and~\ref{result:2}, lies
 the analysis of how the density $\theta_{L}(t,z)$ of active links evolves over time under
 the classical CSMA policy with uniform attempt rate $z$ where all links are idle at 
time $t=0$. To better understand the evolution of $\theta_{L}(t,z)$, consider
 Fig.~\ref{fig:rho_t:400:10000} in which $\theta_{L}(t,z)$ is plotted for a $100\times 100$
lattice. Recalling that each unit of time equals to one packet transmission time, we make the following observations:
\begin{enumerate}
\item At time $t=0$, all links are idle; thus, $\theta_{L}(t,z)=0$.
\item At $t=5 $, i.e., after five packet-transmission times, since the attempt rate $z$ is high, the density $\theta_{L}(t,z)$ increases quickly to $0.39$.
\item At time $t=50$, the density $\theta_{L}(t,z)$ increases to $0.47$.
\item At time $t=200$, the density $\theta_{L}(t,z)$ is slowly reaching to the
 limit of approximately $0.49$.
\end{enumerate}

The evolution of $\theta_{L}(t,z)$ as explained above is a function of the dynamics of CSMA schedules. To see how these dynamics
affect the evolution of $\theta_{L}(t,z)$, in Fig.~\ref{fig:snapshot1}-Fig.~\ref{fig:snapshot3}, we have shown three snapshots of the $100\times 100$ lattice
 interference graph under the classical CSMA policy at times $t\in\{5,50,200\}$.
In these snapshots, we have shown only the active links. We have shown the even active links by blue circles and 
the odd active links by red crosses (see Section~\ref{sec:mod:lattice:uaa} for the definition of odd and even links). Fig.~\ref{fig:snapshot1}-Fig.~\ref{fig:snapshot3} illustrate the following typical characteristics of
the dynamics of schedules under the classical CSMA policy: 
\begin{enumerate}
\item[(a)] After an initial transient behavior,  e.g., at time $t=5$ after five packets transmission times, \emph{clusters}\footnote{For a formal definition of
clusters, their boundaries and areas see Appendix~\ref{sec:definitions}. } of active links have emerged where in each cluster
 all active links are odd and belong to $\setL^{(o)}$, or all active links are even and belong to $\setL^{(e)}$. 
\item[(b)] Shortly after the network starts, e.g., at time $t=5$, clusters of add active links or clusters of even active links
seem to be uniformly distributed over the interference graph.
\item[(c)] Over time, e.g., when we reach time $t=50$, clusters of odd active links or otherwise clusters of even active
links start to dominate the schedule. In Fig.~\ref{fig:snapshot2}-\ref{fig:snapshot3}, clusters of odd active links are dominating the schedule.
\item[(d)] Over time, e.g., at time $t=200$, the dominant clusters dominate further and grow in size.
As a result, the corresponding schedule becomes similar to 
and approaches a maximum schedule in which only one type
of links, odd or otherwise even, are active. 
\end{enumerate}

Based on the above observations, we can see that Proposition~\ref{result:1} states how
 fast the schedule under classical CSMA policy consisting of dominant clusters approaches
a maximum schedule with only odd active links, or only even active links, 
in which half of the links are active. Therefore, to prove Proposition~\ref{result:1}, as the first step, we need to 
analyze the dynamics of CSMA clusters and their evolution over time. Our analysis of the dynamics of CSMA clusters is based on two assumptions on
the properties of CSMA clusters. We formally introduce the two assumptions in the next section. These assumptions are made in order to make the analysis tractable and we comment on this in more details in Section~\ref{sec:eden}. These assumptions are the only assumptions that we use to formally prove Proposition~\ref{result:1} and~\ref{result:2}
 in Appendix~\ref{sec:proof:t1} and~\ref{sec:proof:t2}, respectively.

\section{ Regularity and Randomness Assumptions}\label{sec:assumptionss}

Our analysis leading to Proposition~\ref{result:1} and~\ref{result:2} 
is based on two assumptions on the properties of the clusters
that emerge under the classical CSMA policy on the lattice or torus interference graph 
 with a uniform attempt rate $z$. The first assumption is a \emph{regularity} assumption which states that the geometry of the clusters can not be arbitrary, but satisfy some minimal regularity assumption. The second assumptions is a \emph{randomness assumption}, which states that
 while clusters should satisfy some minimal regularity assumption, 
they cannot be \emph{too regular} and need to satisfy some minimal randomness assumption. In the following, we first introduce the definitions
and notations required to formally state the assumptions.

\subsection{Definitions and Notations}\label{sec:def:not}

 For the purpose of illustration, we assume that
each link $l$ in $G_{L}$ or $\setT_{L}$ represented by coordinates $(i, j)$ can be interpreted
and mapped to the point $(i, j)$ in $\mathbb{R}^{2}$. With such an extension,
we have mapped the vertex set $\setL$ of $G_L$ or $\setT_{L}$ to a subset of points
in $\mathbb{R}^{2}$.

We use the following definitions and notations, valid for both lattice and torus interference graphs. These
definitions are formally presented in Appendix~\ref{sec:definitions}. 
For a given cluster $\setC$, e.g., the cluster of even active links inside the shaded area in Fig.~\ref{fig:cluster_def}, we use $\partial \setC$ to denote its boundary, and use $\ell(\setC)$ to refer to the length of the
 boundary $\partial(\setC)$. We also use $A(\setC)$ to refer to the scalar
 value of the area that cluster $\setC$ covers. Fig.~\ref{fig:cluster_def} illustrates these definitions. Here, the area that a cluster covers and the length of the boundary of a cluster have their usual
 meaning for geometric objects under the euclidean geometry in $\mathbb{R}^2$. We use $\setC_L(t,z)$ to
 denote the set of all clusters that exist at time $t$ in the network of size $L$ under the classical CSMA policy with attempt rate $z$, i.e., 
$$\setC_{L}(t,z)=\{\setC_{i}, 1\leq i\leq \# \setC_{L}(t,z)\},$$
where we define $\# \setC_{L}(t,z)$ to be the number of clusters at time $t$.

At any given time, either the number of even active links is the same as the number of odd active links, or the number of one type (odd or even)
of active links is less than the number of the other type (even or odd, respectively) of active links. In the first case, we define the 
\emph{non-dominating type} to be either
of the odd or even type. In the second case, we define the non-dominant type to be the type of active links whose contribution,
in terms of the number of active links, to the CSMA schedule is
less the contribution of the other type. For instance, in Fig.~\ref{fig:snapshot3}, even is the non-dominating
type. Using this definition, we define $\setC^{(nd)}_{L}(t,z)$ to be
 the set of clusters of non-dominating type of active links at time $t$, and $\# \setC^{(nd)}_{L}(t,z) $ to be the number of such 
clusters. We define a non-dominating cluster to be a cluster that belongs to $\setC^{(nd)}_{L}(t,z) $. 

We also define the set $\setC_{L}^{(nd)}(t,z,\ell)$ to be the set of all non-dominating clusters whose boundary-length is equal to $\ell$, i.e.,
\begin{align}
  \setC_{L}^{(nd)}(t,z,\ell)=\{\setC_{i}: \ \setC_{i}\in\setC_{L}^{(nd)}(t,z), \ell(\setC_{i})=\ell \}.
\end{align}
We define $\# \setC_{L}^{(nd)}(t,z,\ell) $ to be the number of non-dominating clusters with boundary length $\ell$ at time $t$.

We next state the two assumptions, i.e, the regularity and the randomness assumptions.

\begin{figure}
\centering
\includegraphics[scale=0.35]{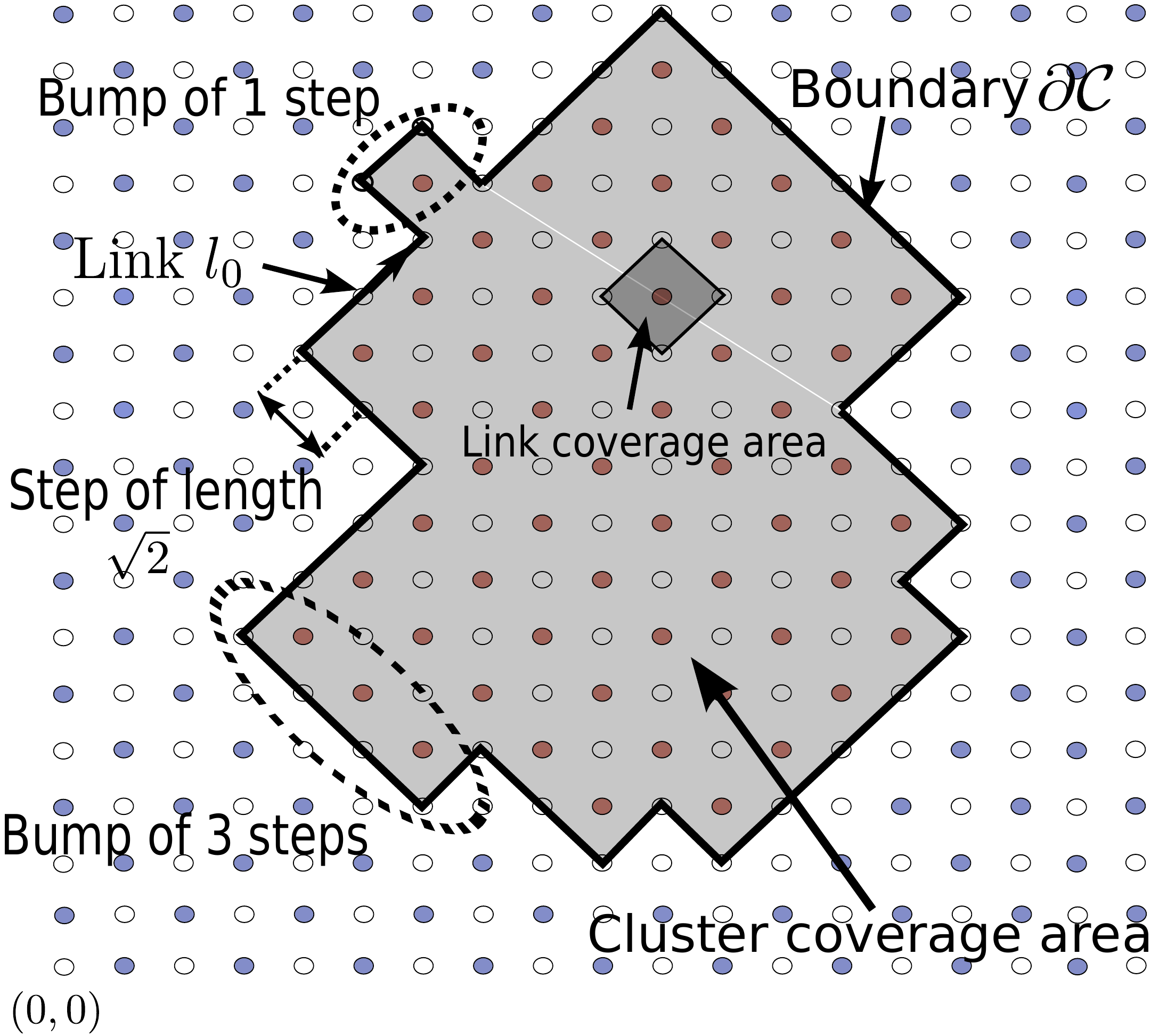}
\caption{Demonstration of active links, link coverage area, cluster coverage area
 and its boundary, and bumps in the lattice interference graph $G_{L}$. Active links
are colored.}
\label{fig:cluster_def}
\end{figure}

\subsection{Regularity Assumption}

The first assumption is a regularity assumption on the geometry of a cluster. 
The intuition behind Assumption~\ref{assum:1} is that clusters do not prefer 
a particular direction when they grow or shrink, i.e., clusters tend to grow or shrink at similar rates in all directions. As a result, it must not
be true that clusters stay \emph{thin} so that they have grown only in one direction and essentially look like a
 line. In other words, clusters must be \emph{fat} so that the have grown or shrunk at similar comparable 
rates in all directions.

Formally, we can define a cluster to be fat if the ratio of its area to the square of its
 boundary length satisfies the following
\begin{align}
   \frac{A(\setC)}{\ell(\setC)^2} \geq c_{a}, \label{ineq:ratio:Al2}
\end{align}
for some constant $c_{a}>0$. For instance, if clusters are rectangular and for which, the
 length of the both sides of
 the rectangles grow by the same factor, as we increase the network size $L$, we then have that the above ratio stays lower-bounded at a constant value. On the other
hand, if clusters grow in only one direction and look like a line, as the network size $L$
 increases, then the above ratio approaches zero.

Instead of stating that every cluster is fat so that \eqref{ineq:ratio:Al2} holds, Assumption~\ref{assum:1} states that \emph{on the average} non-dominating 
clusters are fat so that the ratio
$$ \frac{ \sum_{\setC \in \setC_{L}^{(nd)}(t,z)} A(\setC) } { \sum_{\setC \in C_{L}^{(nd)}(t,z)} \ell(\setC)^2},$$ 
is properly lower-bounded, where the ratio is equal to the ratio of average area covered by a cluster to the average of the square of its boundary length, over all non-dominating clusters 
at time $t$. Assumption~\ref{assum:1} is stated in the following:
 \begin{assumption}\label{assum:1}
For the lattice and torus interference graphs, there exists a positive constant $c_a$ such that
 clusters under the classical CSMA policy
 with uniform attempt rate $z$ 
satisfy the following:
  \begin{align}
    \liminf_{L \to \infty} \inf_{z\geq 1, t\geq t_{0}=1} 
\frac{ \sum_{\setC \in C_{L}^{(nd)}(t,z)} A(\setC) } { \sum_{\setC \in C_{L}^{(nd)}(t,z)} \ell(\setC)^2}> c_a.
  \end{align}
\end{assumption}


The reason why we consider non-dominating clusters is that dominating clusters typically are large clusters that often are not fat.
To see this, consider Fig.~\ref{fig:snapshot2} in which even (blue) clusters are the non-dominating clusters. Intuitively, in this figure,
even clusters can be considered as ``islands'' in a ``sea'' of a large odd (red) cluster.
While the islands, i.e., even clusters in Fig.~\ref{fig:snapshot2} look fat, the ``sea'', i.e., the large odd cluster in the figure
 essentially contains all cluster boundaries and relative to even clusters is thin. 

We also note that in Assumption~\ref{assum:1}, we first take an $\inf$ over $z$ and $t$. 
This $\inf$ makes the ratio lower-bounded
independent of $t$ and $z$.
 We choose $t_{0}$ to be bigger than zero since at time zero there are no clusters. As far as the analysis is
concerned, the choice for a positive constant $t_{0}$ is arbitrary. Finally, we have chosen $z\geq 1$ to ensure that clusters 
exist in the limit of large lattices or toruses.


\subsection{Randomness Assumption}\label{sec:randomness}
Assumption~\ref{assum:1} requires that clusters (at least non-dominating clusters) are 
regular and fat based on the intuition that clusters 
grow or shrink at a similar rates at all directions. 
Assumption~\ref{assum:2} on the other hand requires that clusters are not \emph{too regular} but satisfy some minimal randomness assumption with
respect to the geometry of their boundaries. 
To formally state Assumption~\ref{assum:2}, we first provide the intuition behind this assumption. This intuition does not serve as 
a proof for the assumption. 

Consider the non-dominating cluster $\setC$ of even active links (inside the shaded area) in Fig.~\ref{fig:cluster_def}. 
Recall that we use $\partial \setC$ 
to denote the boundary of cluster $\setC$, as shown by solid lines in Fig~\ref{fig:cluster_def}. We see that on the boundary of the cluster, there are 
a number of inactive links. Consider one arbitrary such inactive link $l_{0}$ on the boundary $\partial \setC$, e.g., link $l_{0}$ in 
Fig.~\ref{fig:cluster_def}. Consider moving along the boundary starting
from link $l_{0}$ in either of the two possible directions, e.g., the one shown in the figure. Let $l_{i}$ be the $i$th link visited along the boundary. Suppose after visiting 
$n$ distinct links, we return to link $l_{0}$. In such a case, we  will close the loop and have $l_{n+1}=l_{0}$. 

 We note that moving along the boundary is possible 
in \emph{steps} of length $\sqrt{2}$, as shown in Fig.~\ref{fig:cluster_def}. We define the $i$th step to be the step from link $l_{i-1}$ to link $l_{i}$, $i\geq1$. For each step, we
define its direction to be \emph{direction} of the movement from link link $l_{i-1}$ to link $l_{i}$. For the cluster $\setC$, let 
$n_r(\setC)$ be the number steps taken to return to the starting link $l_{0}$. Since each step is $\sqrt{2}$ in length, and that the boundary
length of cluster $\setC $ is defined as $\ell( \setC)$, we have
\begin{align}
  n_r(\setC)= \frac{\ell(\setC)}{\sqrt{2}}. \label{ineq:n_r}
\end{align}

Now, suppose the only information 
available about a cluster $\setC$ is its boundary length, i.e., we have that  
$$\ell(\setC)=\ell.$$
While moving along $\partial \setC$, to close the loop, we inevitably need to make direction changes. Formally, we define 
a direction-change event at the $i$th step to be the event where the direction in the $i$th step is different than the previous $(i-1)$th step. 
If $i=1$, we compare the direction of the first step with the direction of the last step returning to $l_{0}$. Let $p_{\ell}$ be the probability that 
there is a direction change clockwise at a given step on the boundary. Due to symmetry, $p_{\ell}$ is also the probability of 
counter-clockwise direction change. Since we close the loop, we have that there must be at least four direction changes. 
Hence, for the expected number of direction changes, we have
\begin{align}
  n_r(\setC) 2 p_{\ell}\geq 4.
\end{align}
Hence, using \eqref{ineq:n_r}, we must have that 
\begin{align}
  p_{\ell}\geq  \frac{c}{\ell} \label{ineq:pl:ell}
\end{align}
where 
\begin{align}
  c=2\sqrt{2}.
\end{align}

The probability $p_{\ell}$ is the probability that there is direction change clockwise (or counter-clockwise) at a given step on
 the boundary of a cluster with length $\ell$. 
However, since we close the loop, it is clear these direction changes are correlated on the boundary of the given cluster. Moreover, these
direction changes can be also correlated on two different clusters. However, as the boundary length $l$ grows, for a fixed $i$, 
we expect direction changes on the $i$th  step to become independent of direction changes at the $(i-1)$th, $(i-2)$th, ..., and the first steps.
 Assumption~\ref{assum:2} uses this intuition to assume and state a property for the boundary of clusters for all $\ell$.

To state the property, consider moving along the boundary $\partial \setC$ of the cluster $\setC$. While moving along the boundary, we 
may encounter \emph{bumps} of $n$ steps. Fig.~\ref{fig:cluster_def} shows bumps of one and three steps. 
Formally, a bump of $n$ steps at a given step occurs when starting at the given step, as the  
first step, 1) there is a direction change in the second step, 2) after the direction change, there is no direction change in the next $n-1$ steps, and
3) at the $(n+2)$th step there is a direction change so that the new direction is opposite to the direction at the first step. 
Hence, when a bump of $n$ steps occurs, for the first time, a direction reversal occurs at the $(n+2)$th step.

Now, consider all non-dominating
clusters of boundary-length $\ell$, and suppose an \emph{independence assumption} holds so that a direction change
along the boundary of a non-dominating cluster with boundary length $\ell$ occurs with probability $p_{\ell}$ independently of any other direction on
 the boundary of the same cluster or other clusters. Using this independence assumption, the probability of having
 a bump of $n$ steps at a given step is
$$ p_{l}(1-p_{l})^{n-1}p_{l}.$$

Let $N^{(b)}_{\setC}(n)$ be the number of bumps of $n$ steps on the boundary of cluster $\setC$. 
Since there are $n_r(\setC)$ steps on the boundary, using \eqref{ineq:n_r} and \eqref{ineq:pl:ell}, for 
the expected number of bumps of $n$ steps on the boundary of cluster $\setC$, we have
\begin{align}
\ev[ N^{(b)}_{\setC}(n)] & = n_r(\setC) p_{l}(1-p_{l})^{n-1}p_{l} \nonumber \\ 
& \geq 2 (1-p_{l})^{n-1} \frac{c}{\ell} 
\nonumber \\ & \geq \frac{c_n'}{\ell}
\label{ineq:Nb:cn}
\end{align}
where 
\begin{align}
  c_{n}'= 2c(1-p_{l})^{n-1} \geq 0.
\end{align}
In particular, for $n=1$, we have 
\begin{align}
  c_{1}'=2c >0.
\end{align}
The constant $c_{1}'$ is independent of $z$, $L$, or time $t$.

As the network size grows, for any fixed $\ell$, we expect the number of non-dominating clusters with boundary-length $\ell$ to increase. As such, if 
a law of large numbers also holds, we expect the number of bumps of $n$ steps averaged over all non-dominating 
clusters of boundary-length $\ell$ to be at least
 $\frac{c_{n}'}{\ell}$, according to \eqref{ineq:Nb:cn}. In other words, with
probability approaching one as $L$ increases,
we expect to have that 
\begin{align}
    \liminf_{L \to \infty}\frac{\sum_{\setC \in \setC_{L}^{(nd)}(t,z,\ell) } N_{\setC}^{(b)}(n) }{\# \setC_{L}^{(nd)}(t,z,\ell)} \geq  \frac{c_{n}'}{\ell}.
\end{align}

Assumption~\ref{assum:2} states that the above inequality holds with probability one: 

\begin{assumption}\label{assum:2}
For lattice and torus interference graphs, for each $n\geq1$, there exists a non-negative constant $c_{n}$ with $c_{1}>0$ such that clusters under the classical CSMA policy with uniform attempt rate $z$ satisfy the
following:
\begin{align}
    \liminf_{L \to \infty} \inf_{z\geq 1,t\geq t_{0}=1}\frac{\sum_{\setC \in \setC_{L}^{(nd)} (t,z,\ell) }
 N_{\setC}^{(b)}(n) }{\# \setC_{L}^{(nd)}(t,z,\ell)} \geq  \frac{c_{n}}{\ell}.
\end{align}
 \end{assumption}

Assumption~\ref{assum:2} implies that cluster boundaries exhibit a minimal amount of randomness and can not be \emph{too regular} or \emph{too smooth}. 
For example, by Assumption~\ref{assum:2}, it would be unlikely to have all non-dominating clusters as perfect rectangles, instead the assumption
 requires that a minimal fraction of such clusters of boundary length $\ell$ to have bumps of for instance length $n=1$ along the their boundaries so that
each such cluster contributes $\frac{c_{1}}{\ell}$ bumps of unit length on the average. As explained earlier, this behaviour would be expected if 
direction changes occurred independently over cluster boundaries. This states that the randomness assumption can be viewed as a
 consequence of, and hence weaker than, 
an independence assumption for the
direction changes along the cluster boundaries. For similar reasons as explained
 for Assumption~\ref{assum:1}, Assumption~\ref{assum:2} takes
 an $\inf$ over $z$ and $t$.

\subsection{Discussion}\label{sec:eden}
A few comments on Assumption~\ref{assum:1} and~\ref{assum:2} are in order. A natural question that arises in this context
is that whether we can prove the conditions of Assumption~\ref{assum:1} and~\ref{assum:2} for 
clusters under the classical CSMA policy. To answer this question, we note 
that there has been considerable effort in trying to derive regularity properties as given in Assumption~\ref{assum:1} for spatial random processes. 
But only for processes that are much simpler than the CSMA process considered here, such results have been obtained. The model for which 
results have been partially obtained, and that is the closest to the random process that we consider in this paper, is the
 Eden model~\cite{peters_radius_1979}\cite{Botet:85}. 

The Eden model is a discrete-time model where initially a cluster is given by a single node, and 
at each time step exactly one node on the lattice is added to the cluster boundary, where the added node is chosen uniformly and independently 
(from all previous steps) from the set of all nodes that are next to a node on the boundary of the current cluster. For this model, it has
 been shown that \cite{peters_radius_1979} clusters indeed satisfy the regularity condition given by Assumption~\ref{assum:1}, and we have that
$$ \lim_{ A(\setC) \to \infty} \log \Big[\frac{A(\setC)}{\ell(\setC)^{2}}\Big] =C,$$
for some constant $C$.

The analysis of the Eden model heavily relies on the assumptions that 
\begin{enumerate}
\item[(a)] at each step exactly one node is added to the cluster boundary, and 
\item[(b)] nodes are added to the boundary of a cluster uniformly and independently of each other and all previous steps.
\end{enumerate}
These assumptions clearly do not hold for clusters under the classical CSMA policy. 
Lack of these assumptions makes
the analysis for CSMA clusters difficult. However, it seems true that for large clusters, the processes by clusters grow or shrink
 tend become (almost) independent at boundary points that are sufficiently far from each other. In such a case, clusters tend to grow 
or shrink in similar comparable rates at all directions, which makes it unlikely to have clusters that are too thin. This is the intuition behind 
Assumption~\ref{assum:1}.

As mentioned earlier, making Assumption~\ref{assum:1}~and~\ref{assum:2} does not make 
the analysis of the dynamics of CSMA clusters trivial. 
This analysis requires tools and techniques from mean-field theory \cite{Leboudec:08}, ODE theory \cite{markley:book,perko:book}
, and large deviation results \cite{borovkov}. 
Appendix~\ref{sec:proof:t1} provides the analysis for the CSMA clusters. 


Simulation results in Section~\ref{sec:overview} and~\ref{sec:main} verify that Assumption~\ref{assum:1}~and~\ref{assum:2} lead to the correct and 
precise characterization of the classical CSMA behaviour and delay-throughput trade-off of
U-CSMA policy. As such, we believe that making these assumptions is well validated, and
the hope is that the formulation of these assumptions will also serve as a 
 possible starting point for additional studies, possibly allowing to further relax or formally prove these assumptions.


\section{Random Geometric Topologies Under Congestion Control Combined with U-CSMA Policy}\label{sec:flow:control}
In this section, we provide the details on how to use U-CSMA policy jointly 
with a congestion control algorithm for general topologies. Congestion control is 
necessary since in practice, the capacity region $ \Gamma$ (see Section~\ref{sec:mod:ideal:csma}) 
is often not known, and
packet arrival rates could initially be outside the capacity region $\Gamma$. We provide  
a detailed look at the simulation results provided earlier in Section~\ref{sec:overview}, and show
that we can use U-CSMA policy in arbitrarily large random geometric topologies to 
assign arrival rates close to the optimal utility with a low packet delay that exhibits order-optimality 
in the sense that it stays bounded as 
the network-size increases.

For simplicity, we assume that links always have data to send and consider flows of data
instead of discrete-size data packets. Using congestion control, we like to ensure that,
1) the admitted flows are indeed supportable, and 2) the set of admitted flows
is chosen in a fully distributed manner such that a network-wide utility function is maximized.
Suppose $U_{l}(\cdot)$ is the (concave, monotonically increasing, and differentiable)
 utility function for link
$l$ as a function of its admitted long-term flow rate $r_{l}$. Suppose the objective is to find
$\{r_{l}\}$ to have the optimal utility $U_{opt}$ \cite{neely:modiano:infocom:05,lin:shroff:ton06,Jiang:09}:
\beqa
U_{opt}=\max_{\{r_{l}\} } \ \ \sum_{l} U(r_{l}), \qquad \{r_{l}\} \in \Gamma.  \label{eq:utility}
\eeqa

Let \emph{utility ratio} $\rho_{u}$ be the ratio of the achieved utility $U_{net}$ to the 
optimal utility $U_{opt}$:
\begin{align}
  \rho_{u}=\frac{U_{net}}{U_{opt}}. \nonumber
\end{align}
Let $\eps_{u}$ be the distance of $\rho_{u}$ to the optimal ratio of 1:
$$\eps_{u}=1-\rho_{u}.$$
In order to be $O(\nu)$, $\nu>0$, away from $U_{opt}$, and have
$\eps_{u}=O(\nu)$,
one approach \cite{neely:modiano:infocom:05,lin:shroff:ton06} is that each link $l$
sets its own admitted flow $\xi_{l}(t)$ at time $t$ to be
\beqa
\xi_{l}(t)=\arg\max_{0\leq\xi \leq  \xi_{max}} \Big[ \nu^{-1}  U_{l}(\xi)- Q_{l}(t) \xi \Big],
\label{eq:flowcontrol:l}
\eeqa
where $Q_{l}(t)$ is the queue of admitted flow to link $l$ at time $t$, and $\xi_{max}$ is a sufficiently large constant. As for scheduling, at any time $t$,
 MWM policy can be used that chooses
 a valid schedule $\bfI(t) \in \setI$ to solve
\beqa
\max_{\bfI(t)\in\setI} \ \ \sum_{l} Q_{l}(t)I_{l}(t), \label{eq:mwm}
\eeqa
where $\bfI(t)=(I_{1}(t))_{l\in\setL}$ with $I_{l}(t)=1$ meaning
link $l$ is active at time
$t$, and $I_{l}(t)=0$, otherwise. The set $\setI$ is the set of all valid schedules. 

Since MWM policy is hard to be implemented (see
Section~\ref{sec:related}), we are interested to use U-CSMA policy for scheduling.
 To see how we use U-CSMA policy, first consider a classical CSMA policy that sets
the attempt-rate of link $l\in\setL$ as
\beqa
z_{l}=e^{w_{l}}, \label{eq:zwl}
\eeqa
where $w_{l}$ is a weight associated with link $l$.
Let 
$$W^{*}=\max_{\bfI \in \setI} \sum_{l} w_{l} I_{l}.$$




Suppose Proposition~\ref{result:1} extends to random geometric interference graphs such that 
under the classical CSMA policy with attempt rates given in \eqref{eq:zwl} and with all links
 inactive at time $t=0$, the schedule $\bfI(t)$ used at time $t$ 
satisfies the following:
\begin{align}
  \sum_{l} w_{l}I_{l}(t) \geq \Big[1-O\Big(\frac{1}{\sqrt{t}}\Big)\Big] W^{*}, \label{ineq:prop:ext}
\end{align}
with probability approaching one in the limit of large networks. 
We note that Proposition~\ref{result:1} can be considered as a special case of the above for the 
torus interference graph with $w_{l}=1$, $l\in\setL$. The above extension essentially states
that we can use CSMA policies to \emph{approximate} MWM policy in random geometric interference graphs, consistent with the existence of PTAS for MWM in geometric graphs \cite{mazumdar:shroff:06}.

The above extension motivates us to design U-CSMA policy as follows. It resets the scheduling
pattern with requiring all links to become inactive every $T$ units of time, as described in Section~\ref{sec:approach}, and sets 
the attempt rate of any link $l$ at time $t$ be $z_{l}(t)$ where
\beqa
z_{l}(t)=e^{w_{l}(t)}=e^{\frac{Q_{l}(t)}{k T}},\label{eq:zl:csma}
\eeqa 
for a fixed integer $k$. By the above choice, we can ensure that the weights
 do not change substantially\footnote{Attempt rates that are slowly varying functions of
links queue-sizes have been used in \cite{rajag:09,ghaderi:10} to achieve throughput-optimality for CSMA 
policies.} over an interval
of length $T$. At the same time, using the unlocking mechanism with period $T$, we ensure that we 
never lock into schedules for more than $T$ time-units.

\begin{figure}[tp]
\centering
\includegraphics[width=.35\textwidth]{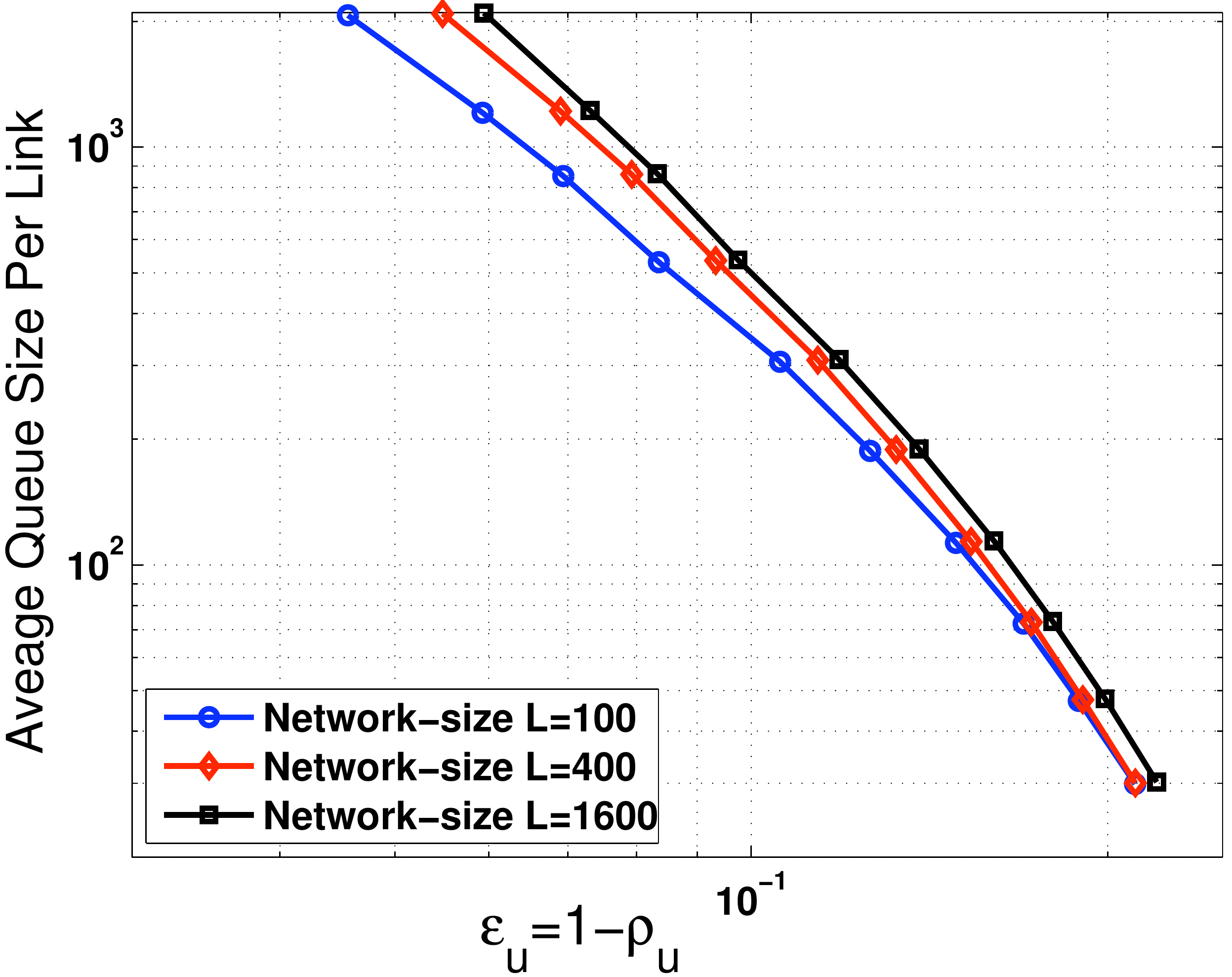}
\caption{$\log$-$\log$ plot of average queue-size under a congestion control algorithm combined
with U-CSMA policy in random geometric interference graphs, as a function 
of the distance $\eps_{u}$ to the optimal utility ratio $\rho_{u}=1$.}
\label{fig:rand_q_log} 
\end{figure}

Since we are working with $\frac{ Q_{l}(t)}{2T}$ instead of $Q_{l}(t)$, we also modify
the congestion control of \eqref{eq:flowcontrol:l} so that every link $l$ chooses its own admitted flow $\xi_{l}(t)$ to 
be 
\beqa
\xi_{l}(t)=\arg\max_{0\leq\xi \leq  \xi_{max}} \Big[ \nu^{-1}  U_{l}(\xi)-\frac{ Q_{l}(t)}{2T} \xi \Big].
\label{eq:flow:fin}
\eeqa

Analysis in \cite{neely:modiano:infocom:05,lin:shroff:ton06} shows that using the complex MWM policy
to solve \eqref{eq:mwm}
along with the distributed congestion control of \eqref{eq:flowcontrol:l}, we will have $\eps_{u}=O(\nu)$
 with average packet delay of $O(\nu^{-1})$. Using a similar analysis and assuming that 
Proposition~\ref{result:1} can be extended as described earlier, we can show that
using U-CSMA with attempt rates given in \eqref{eq:zl:csma} and distributed congestion control of
\eqref{eq:flow:fin}, we will
have $\eps_{u}=O(1/\sqrt{T})+O(\nu)$ with average delay of $O(T/ \nu)$ for large random geometric 
networks.
Choosing $\nu=1/\sqrt{2T}$, we then have that $\eps_{u}=O(\nu)$ and the average delay as
$O(\nu^{-3})$. Hence, to be $\eps_{u}$ from utility optimality, for large networks, the average
delay becomes 
\begin{align}
   O\Big(\Big[ \frac{1}{\eps_{u}} \Big]^{3}\Big), \label{bound:d:ucsma:cgc}
\end{align}
similar to the delay bound derived from Proposition~\ref{result:2} in \eqref{bound:d:ucsma:ind}.

To investigate the performance of U-CSMA policy, with attempt rates given in \eqref{eq:zl:csma},
 used jointly with the congestion of \eqref{eq:flow:fin}, we have conducted
 simulation for random geometric networks
of size $L\in\{100, 400, 1600\}$ with interference range $r$ such that on the average each
link interferes with six other links, as described in Section~\ref{sec:overview}. We have set
 $$ U_{l}(r)=\log(1+r), \qquad l\in\setL. $$
We have used different values of unlocking period $T$, and hence, different values of
$\nu$, in order to obtain different values of $\eps_{u}$. Note that
the exact value of $U_{opt}$ is difficult to compute. However, using the fact that $\log(\cdot)$ is concave, we can show that $U_{opt}$ is upperbounded by $L\log(1+r_{L})$, where $r_{L}$ is the maximum
fraction of links that can be activated in a network with $L$ links. For the setup
considered here, $r_{L}$ approaches $\frac{1}{3}$ for large $L$. We have used the
 upperbound for $U_{opt}$ to obtain
 a conservative estimate for $\eps_{u}$. 

Fig.~\ref{fig:rand_q_log} replots Fig.~\ref{fig:Q}(c) and shows
the average queue-size as a function of $\eps_{u}$ in $\log$-$\log$ scale for small $\eps_{u}$.
We observe the following. First, as $\eps_{u}$ decreases the average queue-size increases
with a slope close to 3 in $\log$-$\log$ scale, as expected by \eqref{bound:d:ucsma:cgc}, similar to 
the exponent 3 obtained in Proposition~\ref{result:2} for delay-throughput under U-CSMA for torus interference graph. 
This suggests that 
an extension of Proposition~\ref{result:1}
should likely hold. 

Second, we observe that the plots for different network-sizes behave similarly.
In particular, for large $L$, i.e., $L=400$ and $L=1600$, the average queue-sizes are very 
close. This confirms that U-CSMA exhibits the same delay order-optimality and the same 
desirable delay-throughput behaviour observed in the torus interference graph
 (see Fig.~\ref{fig:torus_q_log}). In 
particular, the simulation results show that we can indeed use U-CSMA jointly
 with congestion control in large
random geometric networks to operate close to the optimal utility with low packet delay.


\section{Conclusion}\label{sec:con}
In this paper, we have proposed U-CSMA policy as a new CSMA policy. In contrast to the scheduling 
policies in the literature, U-CSMA policy not only is simple and distributed but also provides
high throughput with low delay. Our analysis for torus topologies with uniform packet arrivals shows that 
the delay under U-CSMA is order-optimal, and hence, it stays bounded as the network size
increases. Simulations show that the same desirable delay behaviour also holds for the practical
case where U-CSMA is combined with congestion control in large random geometric networks to
maximize a network-wide utility. Our study in this
paper uses a novel approach to characterize the performance of random access policies and provides
a new prospect into the scheduling of large-scale multihop wireless networks.

\bibliographystyle{IEEEtran} 
\bibliography{IEEEabrv,UCSMA_TechRep}

%

\appendices

\section{Formal Definition of Clusters, Cluster Area, and Length of Cluster Boundary}\label{sec:definitions}


 In this section, we provide formal definitions of a cluster $\setC$, its area $A(\setC)$, and 
its boundary length $\ell(\setC)$ for the lattice interference graph $G_{L}$. Corresponding definitions for the 
torus interference graph $\setT_{L}$ are defined in a similar manner.

Given a lattice interference graph $G_L$, we denote by $\partial G_{L}$ the boundary of $G_L$, i.e.
$\partial G_{L}$ is the set of all links for which at least one coordinate is equal to 0 or $n$.

\subsection{Cluster $\setC$}
To define clusters, consider the lattice $G_{L}$, and assume that links in $G_{L}$ use the classical CSMA policy with uniform attempt rate $z$ for
transmission. Recall that in Section~\ref{sec:netw:mod}, a link $l \in \setL $ is defined to be active at time $t$ if it is transmitting at that time.
Then, at any time $t$, the set of active links can be partitioned to a finite number of \emph{clusters}. We define each
cluster $\setC_{i}$, $1\leq i\leq \# \setC_{L}(t,z)$, to have the following properties:
\begin{enumerate}
\item [a)] $C_{i}$ is a subset of active links in $G_{L}$.
\item [b)] Connectivity: For any two links $l$ and $l'$, where $l\neq l'$ and $l,l'\in \setC_{i}$, there exists a path of $n$ links
$\{l_{1},l_{2}, \cdots, l_{n}\}$ in $\setC_{i}$ for some $n\geq 2$ where $l_{1}=l$ and $l_{n}=l'$, such that
 $l_{k}=(i_{k},j_{k})\in \setC_{i}$, $1\leq k \leq n$,
and $|i_{k}-i'_{k+1}|=|j_{k}-j'_{k+1}|=1$. 
\item [c)] Maximality: $\setC_{i}$ is maximal is the sense that no further links can be added to $\setC_{i}$ without violating one of the above properties.
\end{enumerate}
 The above properties define a cluster as a maximal connected subset of active links, where each link
in the cluster can reach any other in the cluster in a sequence of links, or path, in the cluster. Considering
the mapping from $\setL$ to $\mathbb{R}^{2}$ as explained in Section~\ref{sec:def:not}, along the path, the
 euclidean distance of one link
to the next is $\sqrt{2}$. In Fig.~\ref{fig:cluster_def},
colored links inside the inner polygon represent one cluster.

By the above definition, each cluster contains only odd active links or only even active links. We define
 an odd (resp. even) cluster to be a cluster consisting of odd
(resp. even) links.




\subsection{Cluster Area $A(\setC)$}

For each cluster $\setC$, we use $A(\setC)$ to denote the scalar value of its coverage area $\setA_{\setC}$.

To define $\setA_{\setC}$,
first consider links that are inside
the lattice, i.e., all links $l\notin \partial G_{L}$.
Any such link has four interfering links on the lattice. Considering the mapping from the links in $\setL$ to
the points in $\mathbb{R}^{2}$, for any link $l$ inside the lattice, we
define its \emph{coverage-area} $\setA_{l}$ to be the \emph{square-area} formed by its four closest links. In Fig.~\ref{fig:cluster_def},
we have shown the coverage area of one active link. For any link that is
 not inside the lattice, i.e, $l\in \partial G_{L}$, we
define $\setA_{l}$ to be the intersection of the area $[0,n]\times [0,n]$ in $\mathbb{R}^{2}$ and the square
 that would exist if link $l$ were also inside the lattice. For each link $l$, we define $A_{l}$ to be the scalar value of the 
its coverage area $\setA_{l}$.

For each cluster $\setC$, we can define its \emph{coverage-area} as 
$$\setA_\setC= \cup_{l \in \setC} \setA_{l},$$
 i.e., the union of the coverage area of all links that belong to $\setC$. The coverage-area $\setA_\setC$ contains some
points $(i,j)$, where links of $G_{L}$ may be located, and also some points in $\mathbb{R}^{2}$ where links are not located.
For instance, the area inside the inner polygon in Fig.~\ref{fig:cluster_def} is the coverage-area of one cluster.


\subsection{Length of Cluster Boundary  $\ell(\setC)$}

For each cluster $\setC$, we use $\ell(\setC)$ to denote
the \emph{length of the boundary} of cluster $\setC$, i.e, the \emph{length} of its boundary $\partial \setC$.

For each cluster $\setC$, we define its \emph{boundary} $\partial \setC$ to be the boundary of its coverage-area $\setA_{\setC}$. 
Hence, $\partial \setC$ is the set of all points in $\mathbb{R}^{2}$ that any neighbourhood of which contains
 points both in $\setA_{\setC}$ and points
not in $\setA_{\setC}$ (see Fig.~\ref{fig:cluster_def}).

\section{Proof of Theorem~1}\label{sec:proof:t1}

In this Appendix, we provide the proof of Theorem~1 for the lattice interference graph $G_L$. The proof for the 
torus interference graph $\setT_{L}$ follows similarly. 

Since by \eqref{eq:def:delta:l:n}, 
$$\delta_{L}(t,z) \leq 0.5 < 1,$$
the probability limit in the theorem trivially holds for 
$\tau\leq 2$ by choosing $C_{1}\geq \sqrt{2}$. Therefore, in the rest, we consider only the case where
\begin{align}
  \tau>2, \label{ineq:assum:tau}
\end{align}
and at the end of the proof, we choose $C_1\geq \sqrt{2}$.

In the following, we first introduce several definitions that will be used throughout the proof.
We next in Appendix~\ref{sec:first:ques} classify events that occur while the classical 
CSMA policy operates. Using this classification, in Appendix~\ref{subsec:diff:evol}, we 
derive a set of stochastic \emph{difference} equations that characterize how
the density (fraction) of active links $\theta_{L}(t,z)$ changes over time. In 
Appendix~\ref{subsec:deter:eqs}, we 
define the deterministic ODEs associated with the obtained
stochastic difference equations. Finally in Appendix~\ref{subsec:final:step}, we use several lemmas that
 are provided in Appendix~\ref{appendix:lemma} to state how the properties of the defined ODEs  
relate to the properties of the obtained difference equations, and use these properties to complete the proof of the theorem.

The following are the definitions that will be used throughout the analysis.
\begin{definition}
 The density of events (or links) that satisfy a given property is 
the total number of such events (or links) divided by the total number of links $L$.
\end{definition}
\begin{definition}
We often define density of links that satisfy a given property. To simplify the presentation, we use
the defined density to also denote the set of links that satisfy the property; hence, by writing $l\in \theta_{L}(t,z)$ we mean that link $l$ belongs
to the set of active links whose density is $\theta_{L}(t,z)$.
\end{definition}
\begin{definition}\label{notation:mineq}
  For two matrices (or vectors) $\bfA$ and $\bfB$ with the same dimensions, we write
$$\bfA\leq \bfB $$
if and only if matrix $\bfA$ is component-wise less than or equal to matrix $\bfB$.
\end{definition}

\subsection{Event Classification} \label{sec:first:ques}

In this section, we classify the events that occur while the classical CSMA policy operates, as described in Section~\ref{sec:mod:ideal:csma}. This classification provides a 
basis for the analysis in Appendix~\ref{subsec:diff:evol}. 
We define four types of events: rare events, ordinary events of type-I and type-II, and critical events. 
We start by defining rare events.

\subsubsection{Rare Events}\label{sec:rare:new}

We first define the \emph{$r_{n}$-neighbourhood} of each link $l$ 
where $r_{n}>10$ is a constant. Considering the mapping of the vertex set $\setL$ of $G_L$ to 
the points in $\mathbb{R}^{2}$, we define the $r_{n}$-neighbourhood of any link $l$ 
to be the set of links whose distance from link $l$ is less than $r_{n}$.

A \emph{rare} event occurs at time $t$ if 1) a link $l$ senses the channel as idle in the interval $[t_{1},t]$, $t_{1}<t$, and 2) 
at time $t$, an active link $l'$ in the \emph{$\frac{r_{n}}{2}$-neighbourhood}
of link $l$ stops transmitting. We define links $l$ and $l'$ to be \emph{involved} with the rare event at time $t$. 
If the links that interfere with link $l'$ find the channel as idle when the defined rare event occurs at time $t$, we also 
define link $l$ and these interfering links to be involved with the rare event. Based on these definitions, 
involved with a rare event at time $t$, there are two or more inactive links within $r_{n}$-neighbourhood 
of each other that sense the channel as idle at time $t$. We use $\theta_{L,r}(t,z)$ to denote the density (fraction) of links that are
inactive at time $t$, and that
in whose $r_{n}$-neighbourhood, there is another inactive link that senses the channel as idle at time $t$ such that both inactive links have remained idle until time $t$ after a rare event
with which both links are involved.

\subsubsection{Ordinary Events}\label{sec:typ:event}

Depending on the position of links that stop transmitting, we consider two types of 
\emph{ordinary} events. To simplify the presentation, we study these two types through the 
following examples.

\underline{\emph{2.a) Case1:}} 
Consider the active link $a$ in Fig.~\ref{fig:small_net} that is
next to a corner on its cluster boundary. We call active links such as $a$ that are next to a corner on a cluster boundary as 
\emph{corner} links. Suppose link $a$ stops transmitting, and suppose no other link 
stops transmitting before link $a$ resumes transmitting or link $b$ starts transmitting. In such a
case, links $a$ and $b$ both sense the channel as idle and compete for transmission with attempt-rate $z$. 
With rate $2z$, one of
these links starts 
transmitting before the other. The transmitting link can be either $a$ or
 $b$ with probability $0.5$.
If link $b$ \emph{wins}, the corner in the boundary moves to north-west, and there is  
a change in the boundary shape. Otherwise, the boundary stays the same. 
A similar discussion holds for links $c$ and $d$.

\underline{\emph{2.b) Case2:}}
Consider the active link $e$ that is next to a corner in Fig.~\ref{fig:small_net}. We define corners
such as the one next to link $e$ in Fig.~\ref{fig:small_net} as \emph{double corners} that are
 corners on the boundary of a cluster extending in two directions at least two steps, each of length $\sqrt{2}$.
We define active links such as link $e$ as the \emph{double-corner} links. 
Suppose link $e$ stops transmitting, and suppose no other active link stops transmitting before link 
$e$ resumes transmission or link $f$ starts transmitting. Similar to case~1, with rate $2z$ 
one of links $e$ or $f$ starts transmitting. If link $f$ starts transmitting, which happens
with probability $0.5$, there is a change in the cluster boundary, and two new corners
are created.

\underline{\emph{2.c) Case3:}}
Consider active link $g$ in Fig.~\ref{fig:small_net}. Suppose link $g$ stops transmitting, and suppose
no other link within the neighbourhood of link $g$ stops transmitting. In such a case, none of interfering links of link $g$ can sense the channel as idle. 
As a result, link $g$ is the only link that senses the channel, and with rate $z$ tries to resume its transmission.

\begin{figure}[tp]
\centering
\includegraphics[width=.47\textwidth]{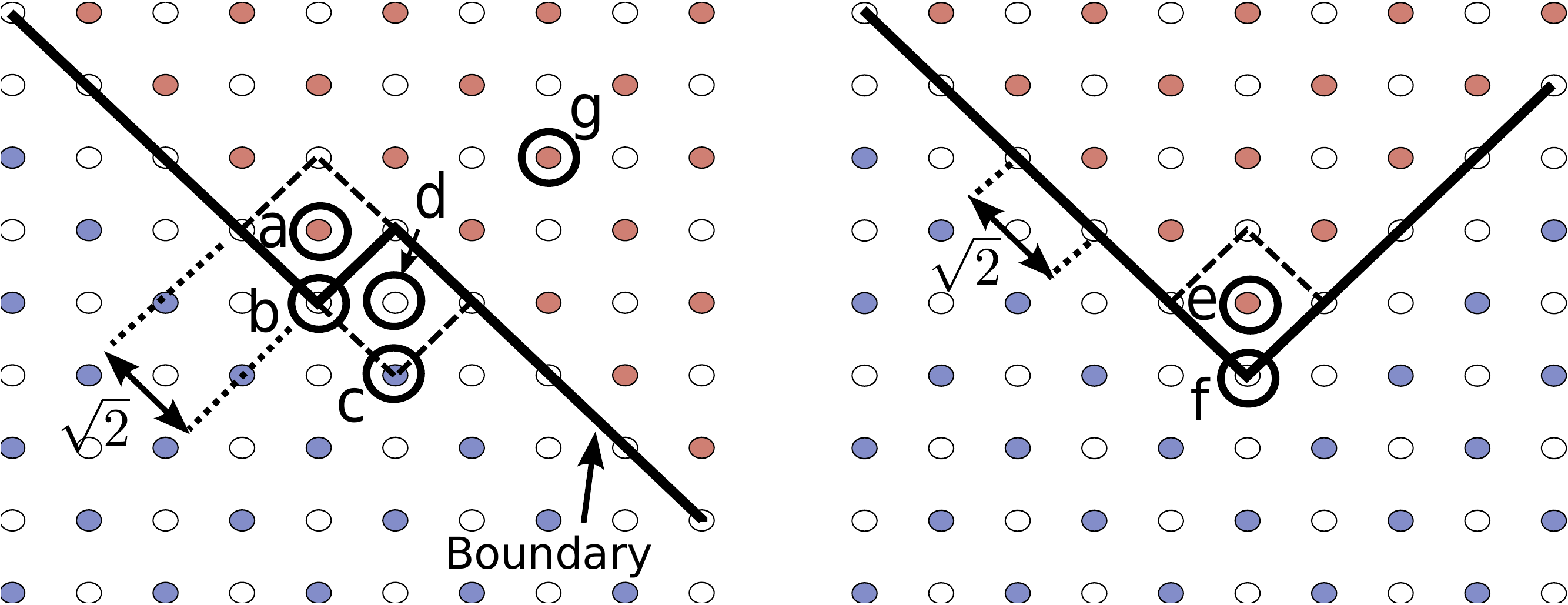}
\caption{Cases leading to ordinary events. Colored circles represent active links.}
\label{fig:small_net}
\end{figure}

Considering the third case above, we define an ordinary event of 
type-I as the event in which without causing a rare event, a non-corner active link, such as link $g$, stops transmitting.
Considering the first two cases above, we define an
ordinary event of type-II to be an event where without causing a rare event, a corner or a double-corner link, such as link $a$ or $e$, respectively,
 stops transmitting.

We use $\tilde{\theta}_{L,1}(t,z)$ to denote the density (fraction) of links that 1) at sometime $t_{1}$, $t_{1}\leq t$, 
have stopped transmitting leading to an ordinary event of type-I, and 2) have sensed the channel as idle in 
the closed interval $[t_{1},t]$. Similarly, we use $\tilde{\theta}_{L,2}(t,z)$ to denote the density (fraction) of
links that any of which, e.g., link $a$ in Fig.~\ref{fig:small_net}, 1) at sometime $t_{1}$, $t_{1}\leq t$ 
has stopped transmitting leading to an ordinary event of type-II, and that 2) the link and the second link associated with the 
ordinary event, e.g., link $b$ in Fig~\ref{fig:small_net}, have sensed the channel as idle in 
the closed time interval $[t_{1},t]$.

\subsubsection{Critical Events}\label{sec:chang:event}

We first define \emph{critical} events. We then explain why these events are critical.

Consider Fig.~\ref{fig:lattice_net:1}, in which, we have shown boundaries of two neighbouring clusters. Consider
the line-segment $f-g$, from link $f$ to link $g$ on the boundary of the right cluster. We
 see that 1) this line-segment has unit step length (length of $\sqrt{2}$), and 2) while moving
along the boundary, the direction of movement before line-segment $f-g$ is opposite to the direction 
after line-segment $f-g$. Considering this example, we define 
a \emph{critical event} to be an event in which 1) on the boundary of a cluster there is a line-segment of unit step length, and 2)
the direction before the line-segment is opposite to the direction after the line-segment, while the boundary
of the cluster is traversed. We define an active link such as $e$ as a \emph{critical link} if around which, the boundary satisfies the above
two properties.

To see why critical events are important, consider the critical link $e$ in Fig.~\ref{fig:lattice_net:1}. Consider
the following sequence of events, which we define as the \emph{$c$-sequence}: 
\begin{itemize}
\item [(1)] Active link $e$ stops transmitting,
\item [(2)] inactive link $f$ (or $g$) starts transmitting, and
\item [(3)] inactive link $g$ (or, respectively, $f$) starts transmitting after link $f$ (or, respectively, $g$) does.
\end{itemize}
In this sequence of events,
 one active link $e$ stops transmitting  
and two new links $f$ and $g$ start transmitting. Hence, the \emph{net effect} is increasing the number
of active links by one\footnote{We note that having 
a reverse sequence of these events, i.e., both links $f$ and $g$ stop transmitting and link $e$ 
start transmitting requires that a rare event to occur.}. 




Inspired by the $c$-sequence defined above, we
 define three new densities, all of which as r.v.'s depending on time $t$. We define $R_{L}(t,z)$ as the density 
of critical events at time $t$, i.e., the total number of these events divided by 
$L$. We use $R_L(t,z)$ to also denote the density (fraction)
 of links that are critical at time $t$, such as $e$ in Fig.~\ref{fig:lattice_net:1}.

\begin{figure}[tp]
\centering
\includegraphics[scale=0.3]{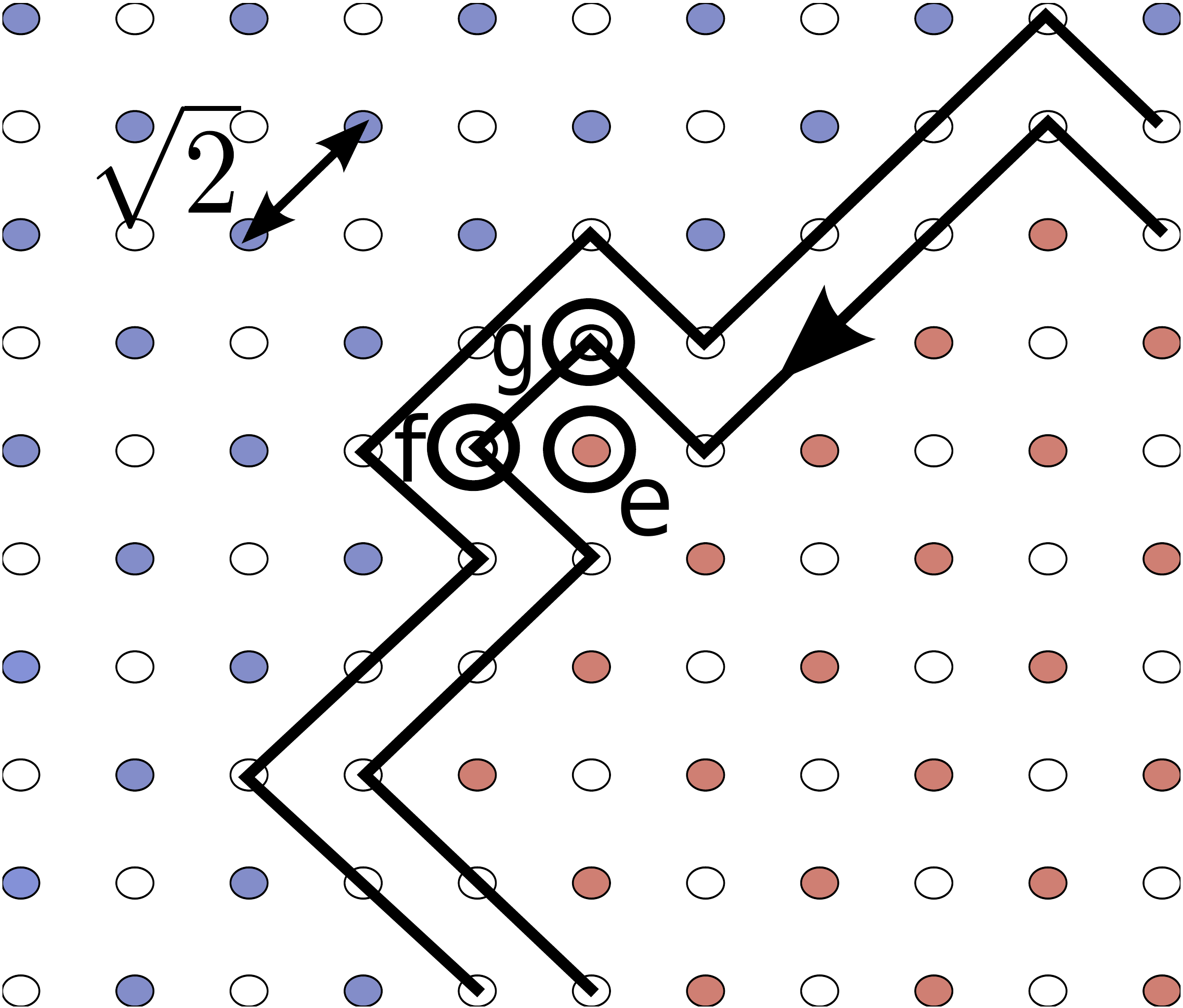}
\caption{Illustration of the critical event on the boundary of a cluster at line-segment $f-g$. Colored circles
represent active links.}
\label{fig:lattice_net:1}
\end{figure}

To define $\theta^{(fg)}_{L}(t,z)$, consider the following. Associated with each critical link, e.g., link $e$ in Fig.~\ref{fig:lattice_net:1}, there are 
two inactive links on the boundary of the cluster to which the critical link belongs, e.g., links $f$ and $g$ in Fig.~\ref{fig:lattice_net:1}.
These inactive links can sense the channel as idle
if the associated critical link stops transmitting. We define $\theta^{(fg)}_{L}(t,z)$ as the density of events, i.e., their total number at time $t$ 
divided by $L$, where both of such inactive links in addition to the (previously) critical link sense the channel as idle. 
Note that these events include the event at time $t$
where the following has occurred: 1) link $e$ stops transmitting at time $t_{1}<t$, 2) link $f$ starts transmitting at time $t_{2}$,
 $t_{1}<t_{2}<t$, and then 3) link $f$ stops
 transmitting at time $t$ so that both links $f$ 
and $g$ sense the channel as idle at time $t$. We note that with each event associated with density $\theta^{(fg)}_{L}(t,z)$, 
three links sense the channel as idle.

 Finally, we define
$\theta^{(f)}_{L}(t,z)$. As mentioned earlier, associated with each critical link, e.g., link $e$ in Fig.~\ref{fig:lattice_net:1}, there are 
two inactive links on the boundary of the cluster to which the critical links belongs, e.g., links $f$ and $g$ in Fig.~\ref{fig:lattice_net:1}. Consider the events
where one of these links has become active and the other link is inactive. We define $\theta^{(f)}_{L}(t,z)$ as the density, i.e., the
 total number divided by $L$, of such events. An example of such events occurs when
the first and second events defined in the $c$-sequence have occurred. In such a case, we have that link
$f$ is active and link $g$ is inactive. Since the events associated with density $\theta^{(f)}_{L}(t,z)$ involve
one active and one inactive link, we reuse $\theta^{(f)}_{L}(t,z)$ to denote the density of such active links, and also reuse it to
denote the density of such inactive links.

\subsection{Difference Equations for the System Evolution}\label{subsec:diff:evol}

In this section, we derive a set of difference equations that will be used to characterize the evolution of $\theta_{L}(t,z)$ over time.
To simplify the presentation, in the rest of the proof, we drop dependency of the defined link and event densities on $z$.
 In Appendix~\ref{sec:first:ques}, we defined densities
$\theta_{L,r}(t)$, $\tilde{\theta}_{L,1}(t)$, $\tilde{\theta}_{L,2}(t)$, $R_{L}(t) $, $\theta^{(fg)}_{L}(t)$, and $\theta^{(f)}_{L}(t)$. 
We use these densities to obtain the difference equations.

 To obtain the difference equations, we focus on the change in the defined densities
from time $t$ to time $t+\varep$ where 
\beqa
\varepsilon = \frac{1}{ \lceil L^{1-\zeta} \rceil}
\label{eq:def:vareps}
\eeqa
where $\zeta$ is a constant and $0<\zeta<1$.
For any fixed $z>0$, we have that
\beqa
\lim_{L\to \infty} z^{2}\varepsilon=\lim_{L\to \infty} z\varepsilon= \lim_{L\to \infty} z L^{-\zeta}= 0. 
\label{eq:lim:z:varepsilon}
\eeqa
We obtain the difference equations for 
$$t \geq t_{0}$$
where 
\begin{align}
  t_{0}=1. \label{eq:def:t0}
\end{align}

Deriving the difference equations requires us to define r.v.'s and obtain their expected values 
given the history $\setH(t)$ of the defined densities up to and including time $t$. Formally, we define $\setH(t)$ as
the history of all defined densities, i.e., $\theta_{L}(t)$, $\theta_{L,r}(t)$, $\tilde{\theta}_{L,1}(t)$, 
$\tilde{\theta}_{L,2}(t)$, $R_{L}(t) $, $\theta^{(fg)}_{L}(t)$, and $\theta^{(f)}_{L}(t)$ from time zero up to and including time $t$. 

In the following, we first focus on the change in $\theta_{L}(t)$ from time $t$ to time $t+\varep$. We then consider 
the change in other densities from time $t$ to time $t+\varep$. The analysis in this section
is based on the following inequality
\begin{align}
  0< \theta_{l} <\theta_{L}(t)<\theta_{u}<0.5, \label{ineq:ineq:assum}
\end{align}
where $\theta_{l}$ and $\theta_{u}$ are positive constants independent of $z$ and $L$.
We can assume the above inequality since the theorem is stated in the limit of first letting $L$ approaching infinity and then letting
$z$ approach infinity, and that by
 Lemma~\ref{lemma:global:bound} for $z>1$ and all $t\in [t_{0},\tau]$, the above inequality
 holds with probability approaching one as $L$
approaches infinity.

\subsubsection{Difference Equation for  $\theta_{L}(t)$}\label{ss:sec:diff:rho:L}
In this section, we obtain the difference equation for the change in $\theta_{L}(t)$ from time $t$ to
time $t+\varep$ given the history $\setH(t)$. Note that by definition, knowing $\setH(t)$, we know the densities
$\theta_{L}(t)$, $\theta_{L,r}(t)$, $\tilde{\theta}_{L,1}(t)$, 
$\tilde{\theta}_{L,2}(t)$, $R_{L}(t) $, $\theta^{(fg)}_{L}(t)$, and $\theta^{(f)}_{L}(t)$. 

Recall that $\theta_{L}(t)$ is the density of active links at time $t$. To obtain the change in $\theta_{L}(t)$, we need to consider all events that can affect the number of 
active links from time $t$ to time $t+\varep$. The events that contribute to the change in $\theta_{L}(t)$
are the following events:
\begin{itemize}
\item A link that is active at time $t$ is not active at time $t+\varep$.
\item An inactive link at time $t$ associated with density $\theta^{(f)}_{L}(t)$ is
active at time $t+\varep$. 
\item Either of three inactive links at time $t$ associated with density $\theta^{(fg)}_{L}(t)$ is active 
at time $t+\varep$.
\item An inactive link at time $t$ associated with density $\tilde{\theta}_{L,1}(t)$ or $\tilde{\theta}_{L,2}(t)$ is active at time $t+\varep$.
\item An inactive link $l$ that at time $t$ finds the channel busy is active time $t+\varep$. This events requires \emph{multiple transitions}, i.e.,
multiple links changing their states from time $t$ to time $t+\varep$,
in the neighbourhood of the given link $l$.
\item An inactive link that at time $t$ is within $r_{n}$-neighbourhood of inactive links associated with density $\theta_{L,r}(t)$
is active at time $t+\varep$.
\end{itemize}

In the following, we consider the contribution of each of the above events on the change in 
$\theta_{L}(t) $ from time $t$ to time $t+\varep$. We start by studying the contribution of active links on the change in $\theta_{L}(t)$. 

\emph{1.a) }Consider an active link $l$, and let $\mathbf{1}_{\theta_{L},l}(t+\varepsilon)$ be the indicator function that at time 
$t+\varepsilon$, link $l$ is not transmitting. Let $X_{1}(t)$ denote the rate of change in $\theta_{L}(t)$ due to
 changes in states of active links, i.e., let
\beqa
X_{1}(t)=-\frac{1}{L\varepsilon} \sum_{l\in \theta_{L}(t)} \mathbf{1}_{\theta_{L},l}(t+\varepsilon).
\label{eq:def:X:hat}
\eeqa
The negative sign is used to indicate that active links that stop transmitting decrease the number of 
active links. We next find $\ev[X_{1}(t) |\setH(t)]$.

We consider two cases leading to $\mathbf{1}_{\theta_{L},l}(t+\varepsilon)=1$. 
In the first case, an active link $l\in \theta_{L}(t)$ stops transmitting at some time $t'$, $t<t'<t+\varepsilon$,
and stays idle by time $t+\varepsilon$. 
Recall that packet transmission times and 
back-off timers are governed by memoryless exponential r.v.'s. As a result, since links stop transmitting with unit rate, independent of $\setH(t)$, the probability that an active link $l\in \theta_{L}(t)$ stops transmitting at some time $t'$, $t<t'<t+\varepsilon$
is 
\begin{align}
  \varepsilon-O(\varepsilon^{2}). \label{eq:prob:stop:xh}
\end{align}
Since the attempt rate is $z$, once
a links stops transmitting at time $t'$, with probability\footnote{If a link that has stopped transmitting cannot transmit again, e.g., since a neighbour has started 
transmitting,  
then the $O(z\varepsilon)$ term in \eqref{eq:prob:stay:idle:xh} can be zero. The probability $1-O(z\varepsilon)$ accounts for
these cases, too.}
\begin{align}
  1-O(z\varepsilon)\label{eq:prob:stay:idle:xh}
\end{align}
 it will
stay idle until time $t+\varepsilon$. Therefore, independent of $\setH(t)$, with probability   
\begin{align}
  \big[\varepsilon-O(\varepsilon^{2})\big]\big[1-O(z\varepsilon)\big]=\varepsilon-O(z\varepsilon^2) \label{eq:l:stop:no:start}
\end{align}
an active link $l\in \theta_{L}(t)$ stops transmitting at some time $t'$, $t<t'<t+\varepsilon$ and stays idle
until time $t+\varepsilon$. 

In the second case, we have $\mathbf{1}_{\theta_{L},l}(t+\varepsilon)=1$ as a result of link $l$ changing its state 
more than once. For example, this happens when after time $t$ link $l$ stops, starts, and stops transmitting all before time $t+\varep$ 
and stays idle up to time $t+\varep$. Independent of $\setH(t)$, the event that link $l$ changes its state more than once, from time $t$ to time $t+\varep$, occurs with
probability not larger than
$$ \big[\varepsilon-O(\varepsilon^{2})\big]\big[z\varepsilon-O\big((z\varepsilon)^{2}\big)\big] =O(z\varepsilon^{2}).$$ This follows since such an event requires a link to
stop and start transmitting at least once, each with rate one or $z$, respectively. Hence, 
the second case occurs with probability 
\begin{align}
  O(z\varepsilon^{2}).\label{eq:prob:c2:ind:xh}
\end{align}

 Considering the above two cases, it follows that
 \begin{align}
P\big[\mathbf{1}_{\theta_{L},l}(t+\varepsilon)=1 \ |\setH(t)\big]=\varepsilon+e_{1,l}(t) \label{eq:p:ind:r:l:varep}   
 \end{align}
where 
\begin{align}
  |e_{1,l}(t)|=O(z\varepsilon^{2}). \label{eq:abs:el1}
\end{align}
Having \eqref{eq:p:ind:r:l:varep}, we can write
\begin{align}
 \ev\big[X_{1}(t) \ |\setH(t)\big]&=-\frac{1}{L \varepsilon}\Big( \sum_{l\in \theta_{L}(t)}(\varepsilon +e_{1,l}(t) )  \Big)
\nonumber \\ &=-\theta_{L}(t) (1 +e_{1}(t) ),
\label{eq:hat:rho:hat}
\end{align}
where 
\begin{align}
  |e_{1}(t)|=O(z   \varepsilon). \label{eq:hat:rho:hat:e1:n}
\end{align}

We next find the variance of $X_{1}(t)$ given $\setH(t)$. Let 
\begin{align}
B_{1}(t)\triangleq X_{1}(t)-\ev\big[X_{1}(t)\ |\setH(t)\big].\label{eq:def:var:X:hat}  
\end{align}
By \eqref{eq:p:ind:r:l:varep} and \eqref{eq:abs:el1}, we have
$$P\big[ \mathbf{1}_{\theta_{L},l}(t+\varepsilon)=1 \ |\setH(t) \big] \leq \varepsilon +O(z\varep^2).$$ 
Moreover, as shown in Lemma~\ref{lemma:cond:prob}, since packet transmission times and back-off timers are independent across all links, we have that
for two different active links $l,l'\in \theta_{L}(t)$
\begin{align}
  P\Big[\mathbf{1}_{\theta_{L},l}(t+\varepsilon)=1 \  \big| \mathbf{1}_{\theta_{L},l'}(t+\varepsilon)=1, \setH(t)\Big]\leq \varepsilon +
O(z\varep^2)
.
\end{align}
Using the above inequality, \eqref{eq:p:ind:r:l:varep}, and \eqref{eq:abs:el1}, we obtain
\begin{align}
&  \ev\Big[\mathbf{1}_{\theta_{L},l}(t+\varepsilon)\mathbf{1}_{\theta_{L},l'}(t+\varepsilon) \big| \setH(t)\Big]
\nonumber \\ &  -\ev\Big[\mathbf{1}_{\theta_{L},l}(t+\varepsilon) \ \big|  \setH(t)\Big]\ev\Big[\mathbf{1}_{\theta_{L},l'}(t+\varepsilon)\ \big|  \setH(t) \Big]
 = O( z\varepsilon^{3} ).
\label{ineq:corel:pll'}
\end{align}
Using the definitions in \eqref{eq:def:X:hat} and \eqref{eq:def:var:X:hat}, we have
\begin{align}
&   \ev\Big[B_{1}(t)^{2}\ \big| \setH(t)\Big]
 =\frac{1}{(L \varepsilon)^{2}} \Bigg(
\nonumber \\ &
\sum_{l\in\theta_{L}(t) }
 \ev\Big[ \Big(\mathbf{1}_{\theta_{L},l}(t+\varepsilon)-\ev\big[\mathbf{1}_{\theta_{L},l}(t+\varepsilon)  \ | \setH(t) \big] \Big)^2  \ \big| \setH(t)\Big]
\nonumber \\ &
+ \sum_{l,l'\in\theta_{L}(t), l\neq l' }
 \ev\bigg[ \Big(\mathbf{1}_{\theta_{L},l}(t+\varepsilon)-\ev\big[\mathbf{1}_{\theta_{L},l}(t+\varepsilon) \ | \setH(t) \big] \Big)
\nonumber \\ & \qquad 
\Big(\mathbf{1}_{\theta_{L},l'}(t+\varepsilon)-\ev\big[\mathbf{1}_{\theta_{L},l'}(t+\varepsilon)  \ | \setH(t) \big] \Big) \ \Big| \setH(t) \bigg]
\Bigg)
\end{align}

Having obtained the above equality, we can use \eqref{eq:p:ind:r:l:varep} and \eqref{ineq:corel:pll'} to upper-bound
the first and the second summation in the above, respectively, leading to
\begin{align}
   \ev\big[B_{1}(t)^{2} \ | \setH(t)\big]     
  &\leq \frac{1}{(L \varepsilon)^{2}} \sum_{l\in\theta_{L}(t)} (\varepsilon+O(z\varepsilon^2))
   \nonumber \\  & + \frac{1}{(L \varepsilon)^{2}} \sum_{l,l'\in\theta_{L}(t), l\neq l' } O(z\varepsilon^3)
\nonumber \\ & \leq \frac{1}{(L \varepsilon)^{2}} L\theta_{L}(t) (\varepsilon+O(z\varepsilon^2))
\nonumber \\ & + \frac{1}{(L \varepsilon)^{2}} [L\theta_{L}(t)]^2 O(z\varepsilon^3).
\end{align}
Since $0\leq \theta_{L}(t)\leq1$, we can use \eqref{eq:lim:z:varepsilon} in the above to show that
\begin{align}
   \ev\big[B_{1}(t)^{2} \ | \setH(t) \big]  & \leq
L^{-\zeta} \big[1+O(z\varepsilon)\big]+O(z\varepsilon)
\nonumber \\ &=L^{-\zeta}+ O(z\varepsilon)
\label{eq:hat:rho:hat:error}.
\end{align}

Thus far, we have characterized the contribution of active links by obtaining \eqref{eq:hat:rho:hat} and \eqref{eq:hat:rho:hat:error}.

\emph{1.b) }We next consider the contribution of inactive links associated with density $\theta^{(f)}_{L}(t)$ in 
the change of $\theta_{L}(t)$ from time $t$ to time $t+\varep$. Any such inactive link if it is active at time $t+\varep$,
it increases by one the number of active links. 

Recall that $\theta^{(f)}(t)$ represents the density of a subset of inactive links (see Appendix~\ref{sec:chang:event}).  Consider the inactive link $l\in \theta^{(f)}_{L}(t)$. For such link $l$, let $\mathbf{1}_{\theta^{(f)}_{L},l}(t+\varepsilon)$ 
be the indicator function that link $l$ is active at time $t+\varep$. Using a similar analysis used to obtain \eqref{eq:p:ind:r:l:varep} and \eqref{eq:abs:el1}, and recalling that the attempt rate is $z$, we can show that
\beqa
 P\big[\mathbf{1}_{\theta^{(f)}_{L},l}(t+\varepsilon)=1  \ | \setH(t) \big]= z \varepsilon+e_{2,l}(t),
\label{eq:P:rho:f:l:s}
\eeqa
 where 
\beqa
|e_{2,l}(t)|=O(z^{2}\varepsilon^{2}).
\label{eq:P:rho:f:l:s:e3l}
\eeqa

 Define
 $X_{2}(t)$ similar to $X_{1}(t)$ as the rate of change in $\theta_{L}(t)$ associated with inactive links
in $ \theta^{(f)}_{L}(t)$, i.e, let 
\begin{align}
 X_{2}(t)=\frac{1}{L\varepsilon} \sum_{l\in \theta^{(f)}_{L}(t)} \mathbf{1}_{\theta^{(f)}_{L},l}(t+\varepsilon). \nonumber 
\end{align}
In addition, given $\setH(t)$, let 
$$B_{2}(t)=X_{2}(t)-  \ev\big[X_{2}(t) \ | \setH(t) \big].$$

Using \eqref{eq:P:rho:f:l:s} and \eqref{eq:P:rho:f:l:s:e3l}, we obtain
\begin{align}
 \ev\big[X_{2}(t)\ | \setH(t) \big]= \theta^{(f)}_{L}(t) (z +e_{2}(t) ), 
\label{eq:hat:rho:f}
\end{align}
where
\begin{align}
  |e_{2}(t)|=O(z^{2}\varepsilon). \label{eq:hat:rho:f:e:3}
\end{align}
Using a lemma similar to Lemma~\ref{lemma:cond:prob}, and taking similar steps leading to \eqref{eq:hat:rho:hat:error}, we also obtain
\begin{align}
   \ev\Big[ B_{2}(t)^{2} \ \big| \setH(t)\Big]\leq z L^{-\zeta}+ O(z^3\varep).
\label{eq:hat:rho:f:error}
\end{align}


\emph{1.c) }We next consider the contribution of inactive links associated with density $\theta^{(fg)}_{L}(t)$ in 
the change of $\theta_{L}(t)$ from time $t$ to time $t+\varep$. By the definition given in Section~\ref{sec:chang:event}, 
 $\theta^{(fg)}_{L}(t)$ represents the density of events where three links are inactive and sense the channel as idle at time $t$. As a result,
these events each with rate $3z$ increase the number of active links. This is similar to the previous case where
the inactive links associated with $\theta^{(f)}_{L}(t)$ increase the number of active links. Therefore, the
 contribution of inactive links associated with  $\theta^{(fg)}_{L}(t)$ can be 
characterized in a similar way used to characterize the contribution of inactive links associated with $\theta^{(f)}_{L}(t)$.

In particular, defining $X_{3}(t)$, similar to $X_{2}(t)$, to be the rate of change in $\theta_{L}(t)$ due to inactive links
associated with density $ \theta^{(fg)}_{L}(t)$, and letting
$$B_{3}(t)=X_{3}(t)-  \ev\big[X_{3}(t) \ | \setH(t) \big],$$
similar to \eqref{eq:hat:rho:f} and \eqref{eq:hat:rho:f:e:3}, we obtain that 
\begin{align}
 \ev\big[X_{3}(t) \ | \setH(t) \big]= \theta^{(fg)}_{L}(t) (3z +e_{3}(t) ), 
\label{eq:hat:rho:fg}
\end{align}
where
\begin{align}
  |e_{3}(t)|=O(z^{2}\varepsilon). \label{eq:hat:rho:fg:e:3}
\end{align}
Similar to \eqref{eq:hat:rho:f:error}, we also obtain 
\begin{align}
   \ev\Big[ B_{3}(t)^{2} \ \big| \setH(t) \Big]\leq z L^{-\zeta}+ O(z^3\varep).
\label{eq:hat:rho:fg:error}
\end{align}


\emph{1.d) }We next consider the contribution of inactive links associated with density $\tilde{\theta}_{L,1}(t)$ or $\tilde{\theta}_{L,2}(t)$ on the 
change in $\theta_{L}(t)$
from time $t$ to time $ t+\varep$. Let
\begin{align}
  \tilde{\theta}_{L}(t)=\tilde{\theta}_{L,1}(t)+\tilde{\theta}_{L,2}(t).
\end{align}
Based on the definitions given in Appendix~\ref{sec:typ:event}, the density $\tilde{\theta}_{L,1}(t)$ accounts for 
ordinary events of type-I leading to only one inactive link that tries to access the channel with rate $z$. 
In contrast, $\tilde{\theta}_{L,2}(t)$ 
accounts for ordinary events of type-II leading to two inactive links that with total rate of $2z$ are trying to access the channel.
 However, in both cases, only
one link can be made active. As a result, the only difference between the contribution of inactive links 
associated with  $\tilde{\theta}_{L,1}(t)$ or $\tilde{\theta}_{L,2}(t)$ 
 is the total rate by which 
the channel is accessed. In the rest, we assume in either case with a lower rate of $z$ the channel is accessed. This leads to a lower bound on the fraction of active links, which is what is stated in the statement of Theorem~1. 

Any inactive link associated with $\tilde{\theta}_{L}(t)$ if active at time $t+\varep$ increases the number of active links by one.
Define
$X_{4}(t)$ similar to $X_{2}(t)$ to be the rate of change in $\theta_{L}(t)$ due to inactive links associated with $\tilde{\theta}_{L}(t)$. 
In addition, let
$$B_{4}(t)=X_{4}(t)-  \ev\big[X_{4}(t) \ | \setH(t) \big].$$
Taking similar steps used to obtain \eqref{eq:hat:rho:f}, \eqref{eq:hat:rho:f:e:3}, and \eqref{eq:hat:rho:f:error}, we can show that
\begin{align}
 \ev\big[X_{4}(t) \ | \setH(t) \big]= \tilde{\theta}_{L}(t) (z +e_{4}(t) ), \label{eq:hat:rho:tilde}
\end{align}
where 
\begin{align}
|e_{4}(t)|=O(z^{2}\varepsilon),   \label{eq:hat:rho:tilde:e:4}
\end{align}
and 
\begin{align}
   \ev\Big[  B_{4}(t)^{2} \ \big| \setH(t) \Big]\leq zL^{-\zeta} + O(z^3 \varep).
\label{eq:hat:rho:tilde:error}
\end{align}


\emph{1.e) }We next consider the contribution of inactive links at time $t$ that sense the channel as busy 
at time $t$ on the change in $\theta_L(t)$, from time $t$ to time $t+\varep$. Let $\theta^{(b)}(t)$ be the fraction of these links at time $t$. 
Consider the link $l\in \theta^{(b)}(t) $. 
Link $l$ can affect the number of active links only when 
the following occurs. First, all active interfering neighbours of link $l$ stop transmitting before time $t+\varep$, and second, link $l$ starts transmitting
before time $t+\varep$.
For such an event, at least one link should stop transmitting, which based on~(\ref{eq:p:ind:r:l:varep}) and (\ref{eq:abs:el1}), independent of
$\setH(t)$, occurs with probability at most
\begin{align}
  \varep+O(z\varep^2)
\end{align}
Moreover, afterwards, link $l$ should start transmitting, which according to \eqref{eq:P:rho:f:l:s} and
 \eqref{eq:P:rho:f:l:s:e3l}, occurs with probability at most
\begin{align}
  z\varep+O(z^2\varep^2)
\end{align}
As a result, independent of $\setH(t)$, the probability that link $l\in \theta^{(b)}(t)$ starts transmitting before or at time $t+\varep$ is at most
\begin{align}
  z\varep^2+O(z^2 \varep^3).
\end{align}

Define $X_{5}(t)$ similar to $X_{2}(t) $ to be the rate of change in $\theta_{L}(t)$ due to inactive links that find the channel
busy at time $t$. In addition, let
$$B_{5}(t)=X_{5}(t)-  \ev\big[X_{5}(t) \ | \setH(t)\big].$$ 
Treating the density $\theta^{(b)}(t)$ in a similar way as we treated density $\theta^{(f)}_L(t)$, and we obtain
\begin{align}
 \ev\big[X_{5}(t) \ | \setH(t)\big] \leq   \theta^{(b)}(t) (z\varep   +e_{5}(t) ), \label{eq:X:rho:m}
\end{align}
where 
\begin{align}
|e_{5}(t)|=O(z^{2}\varepsilon^2 ).   \label{eq:X:rho:m:e:5}
\end{align}
Since $\theta^{(b)}(t)\leq 1$, we have
\begin{align}
  \ev\big[X_{5}(t)\ | \setH(t)\big] = O(z\varep).\label{eq:X:rho:m:1}
\end{align}
Moreover, we can show that
\begin{align}
   \ev\Big[  B_{5}(t)^{2} \ \big| \setH(t)\Big]= O(z^2 \varep).
\label{eq:X:rho:m:error}
\end{align}


\emph{1.f) }The final contribution that we consider is due to the events that are not previously considered. 
These events are related to rare events, and are the events where an inactive link that at time $t$ is within $r_{n}$-neighbourhood of an inactive link associated with density $\theta_{L,r}(t)$ is active at time $t+\varep$. Define $X_{r}(t)$ similar to $X_{1}(t)$ or $X_{2}(t) $ to be the rate of change in $\theta_{L}(t)$ due to these events. In addition, let
$$B_{6}(t)=X_{r}(t)-  \ev\big[X_{r}(t)\ | \setH(t)\big] .$$ 

Considering the definition for $r_{n}$-neighbourhood given in Appendix~\ref{sec:rare:new}, for any $r_{n}$, there exists a $\hat{c}_{n}$ where
 $\hat{c}_{n}$ represents an
upperbound on the number 
of links that are within $r_{n}$-neighbourhood of any given link $l$. As a result, the total number of links that may lead to the final contribution is bounded
by 
\begin{align}
  \hat{c}_{n}L\theta_{L,r}(t).
\end{align}
Defining $\theta_{L,r}'(t)$ to be the density of links that may lead to the final contribution, we have 
\begin{align}
  \theta_{L,r}'(t)\leq \hat{c}_{n}\theta_{L,r}(t).
\end{align}

Treating the density $\theta_{r,L}'(t)$ the same way we treated density $\theta^{(f)}_L(t)$, we obtain
\begin{align}
 \ev\big[X_{r}(t)\ | \setH(t)\big] \leq \hat{c}_{n}\theta_{L,r}(t) (z +e_{6}(t) ), \label{eq:X:rho:rare}
\end{align}
where 
\begin{align}
|e_{6}(t)|=O(z^{2}\varepsilon),   \label{eq:X:rho:rare:e:5}
\end{align}
and 
\begin{align}
   \ev\Big[  B_{6}(t)^{2} \ \big| \setH(t)\Big]\leq zL^{-\zeta} + O(z^3 \varep).
\label{eq:X:rho:rare:error}
\end{align}


\emph{1.g) }Having considered the contribution of all events, we next derive the difference 
equation for $\theta_{L}(t)$. We first give a definition. For any random process $h(t)$ (scalar or vector), we define the random process $\Delta(h(t))$ as
\beqa
\Delta(h(t))= \frac{1}{\varepsilon}\Big[ h(t+\varepsilon)-h(t)\Big].
\label{eq:def:Delta:h}
\eeqa
This process measures the rate of change in the random process $h(t)$ from time $t$ to time $t+\varep$.

Using the definition of $\Delta(\cdot)$ and the results in~(\ref{eq:hat:rho:hat})~(\ref{eq:hat:rho:hat:e1:n})~\eqref{eq:hat:rho:hat:error},
 (\ref{eq:hat:rho:f})-(\ref{eq:hat:rho:f:error}), (\ref{eq:hat:rho:fg})-(\ref{eq:hat:rho:fg:error}), (\ref{eq:hat:rho:tilde})-(\ref{eq:hat:rho:tilde:error}), \eqref{eq:X:rho:m:1}-\eqref{eq:X:rho:m:error},
(\ref{eq:X:rho:rare})-(\ref{eq:X:rho:rare:error}), given $\setH(t)$, we obtain
\begin{align}
\Delta(\theta_{L}(t))&= \bigg[-\theta_{L}+3z \theta_{L}^{(fg)}(t) +z\theta_{L}^{(f)}(t) +z\tilde{\theta}_{L}(t)
\nonumber \\ &\qquad \qquad \qquad \qquad  \qquad 
 + e(t) +B(t)\bigg], \label{eq:diffeq:1}
\end{align}
where 
\begin{align}
|e(t)|=O(z^{2}\varepsilon)+ \ev\big[X_{r}(t) \ | \setH(t) \big],   \label{eq:ehat:bound}
\end{align}
and 
$$B(t)=\sum_{i=1}^{6}B_{i}(t),$$
for which 
\beqa
 \ev\Big[B(t)^{2}\ \big|\setH(t)\Big]\leq O(z L^{-\zeta})+O(z^{3}\varep).
\label{eq:B:hat:ex}
\eeqa
In the next section, we derive the difference equations for other defined densities
$\theta_{L}^{(fg)}(t)$, $\theta_{L}^{(f)}(t)$, and $\tilde{\theta}_{L}(t)$.
 

\subsubsection{Difference Equations for Other  Densities}\label{sec:all:densities}

In the previous section, we obtained the difference equation for $\theta_{L}(t)$. We can take 
similar steps to obtain difference equations for other densities $\theta^{(fg)}(t)$, $\theta^{(f)}(t)$, and $\tilde{\theta}_{L}(t)$.
The difference equations for these densities are provided in Lemmas~\ref{lemma:diff:rho:fg}-\ref{lemma:diff:rho:tilde}.

The difference equations developed in the previous section and the ones in Lemmas~\ref{lemma:diff:rho:fg}-\ref{lemma:diff:rho:tilde} 
are based on terms such as 
$X_{r}(t)$, $X^{(f)}_{r}(t)$, $X_{r}^{(fg)}(t)$, $\tilde{X}_{r}(t)$, $\tilde{X}_{h}(t)$, and $R_{L}(t)$. In this section, we study
 these terms to obtain a set 
of difference equations that are based on only the densities $\theta_{L}(t)$, $\theta^{(fg)}(t)$, $\theta^{(f)}(t)$, and $\tilde{\theta}(t)$. 

Let $\Omega$ be the sample probability space and $\omega \in \Omega$. For a given $z$, we define the event $\setE(L,z,\tau)$ as 
\begin{align}
  \setE(L,z,\tau) = & \{\omega:\sup_{t\in [t_{0},\tau]} \theta_{L,r}(t) < c_{r}z^{-2} \} \cap
\nonumber \\ & \{\omega: \sup_{t\in [t_{0},\tau]} \theta_{L,h}(t) < c_{\theta}z^{-1} \} \cap
\nonumber \\ & \{\omega: \sup_{t\in [t_{0},\tau]} \theta_{L}^{(f)}(t) < c_{\theta}z^{-1} \} \cap
\nonumber \\ & \{\omega: \sup_{t\in [t_{0},\tau]} \theta_{L}^{(fg)}(t)< c_{\theta}z^{-1} \} \cap
\nonumber \\ & \{\omega: \sup_{t\in [t_{0},\tau]} \tilde{\theta}_{L}(t) < c_{\theta}z^{-1} \} \cap
\nonumber \\ & \{\omega: \sup_{t\in [t_{0},\tau]} \theta_{L}(t) < 0.5- c_{\theta,2 }e^{-\tau} \}\cap
\nonumber \\ & \{\omega: \inf_{t\in [t_{0},\tau]} \theta_{L}(t) >  c_{\theta,1 }e^{-\tau} \}\cap
\nonumber \\ & \{\omega:  c_{\theta,1}e^{-t_0} < \theta_{L}(t_{0}) <  0.5- c_{\theta,2 }e^{-t_0} \}
\label{eq:def:E:Ltau}
\end{align}
for some constants $c_{r}>0$, $c_{\theta}>0$, $c_{\theta,1}>0$, and $c_{\theta,2}>0$, all independent of $L$ and $z$, and for 
$\theta_{L,h}(t)$ defined as the density of links that are inactive and sense the channel as idle at time $t$. In the rest, to simplify the presentation, we let
$$\setE(L,z) =\setE(L,z,\tau),$$
and drop the dependency of $\setE(L,z,\tau)$ on $\tau$. Where appropriate, we do the same for other functions that depend on 
$\tau$.

Using Lemma~\ref{lemma:global:bound}, Lemma~\ref{lemma:rare}, and Lemma~\ref{lemma:rholh}, we have that for any given finite and fixed $z>1$ and
 $\tau>t_{0}$,
\begin{align}
\lim_{L\to \infty} P\big[ \setE(L,z) \big]=1.   \label{eq:lim:prob:setE}
\end{align}
Hence, we can define $\eps_{r}(L,z)>0$ such that  
\begin{align}
  P\big[\setE(L,z)\big]\geq 1-\eps_{r}(L,z)
\end{align}
where for any finite and fixed $z>1$ and $\tau>t_{0}$
\begin{align}
  \lim_{L\to \infty } \eps_{r}(L,z)=0
\label{eq:eps:Ltau:lim}
\end{align}
 In particular, there exists $L_{0}(z)=L_{0}(z,\tau)$ such that for all $L>L_{0}(z)$
 \begin{align}
   \eps_{r}(L,z)\leq 0.5. \label{ineq:epsr:Llower}
 \end{align}

We next use the definition of $ \setE(L,z) $ to state $R_{L}(t)$ as a function of $\theta_{L}(t)$.  
Considering \eqref{eq:def:delta:l:n} and \eqref{ineq:ineq:assum}, we choose constants $z_{R}$ and $L_{R}$ such that for $z>z_{R}$ and $L>L_{R}$ if $ \setE(L,z)$ occurs, the conditions of 
Lemma~\ref{lemma:R} hold so that w.p.1. we have 
\begin{align}
  R_{L}(t)\geq c_{R} \delta_{L}(t)^{3}
\end{align}
where $c_{R}$ is a positive constant independent of $z$, $L$, $\tau$, and $t$. Since
a larger $R_{L}(t)$ indicates more critical events, and hence, a larger rate 
by which $\theta_{L}(t)$ increases, in the rest we assume w.p.1.
\begin{align}
  R_{L}(t)= c_{R} \delta_{L}(t)^{3}.\label{eq:RL:delta}
\end{align}

Equation \eqref{eq:RL:delta} allows us to write $R_{L}(t)$ as a function of $\delta_{L}(t)$ and
thus as a function of $\theta_{L}(t)$.

We next consider the terms $X_{r}(t)$, $X_{r}^{(fg)}(t)$, $X_{r}^{(f)}(t)$, and $\tilde{X}_{r}(t)$ and $\tilde{X}_{h}(t)$
given in \eqref{eq:X:rho:rare}, Lemma~\ref{lemma:diff:rho:fg}, Lemma~\ref{lemma:diff:rho:f}, Lemma~\ref{lemma:diff:rho:tilde}, respectively. 
We first consider the difference equation in \eqref{eq:diffeq:1} containing the terms $e(t)$ and $B(t)$, and use
 the properties of event $\setE(L,z)$ to obtain an upperbound
for $|e(t)|$ given in \eqref{eq:ehat:bound} that contains the expected value of $X_{r}(t)$. To do so, we consider
the difference equation in \eqref{eq:diffeq:1} conditioned on the event $\setE(L,z)$ as well as history $\setH(t)$. 
The equation in \eqref{eq:diffeq:1} contains some terms such as $-\theta_{L}(t)$ and $3z \theta_{L}^{(fg)}$
that given $\setH(t)$ are known, and hence, knowing that event $\setE(L,z)$ occurs does not affect these terms. However, 
both of $e(t)$ and $B(t)$ are affected by knowing that event $\setE(L,z)$ occurs. 

To upperbound  the term $|e(t)|$ given that $\setE(L,z)$ occurs, we first note
 that by the definition of $\setE(L,z)$ given in \eqref{eq:def:E:Ltau}, we have
\begin{align}
 \theta_{L,r}(t) \leq c_{r} z^{-2}. \label{ineq:rholr:z2:b}
\end{align}
If this inequality was the only information available, in addition to $\setH(t)$, then 
by \eqref{eq:X:rho:rare}, we would have that 
\begin{align}
  \ev\Big[ X_{r}(t) \ \big| \setH(t),\eqref{ineq:rholr:z2:b} \Big] \leq \hat{c}_{n}\theta_{L,r}(t) (z +e_{6}(t) )= O(z^{-1}).
\end{align}
However, the event $\setE(L,z)$ includes other inequalities. To address this, we can use the above bound, Lemma~\ref{lemma:conditional}, and assume
 $L> L_{0}(z)$ so that \eqref{ineq:epsr:Llower} holds, to show that
\begin{align}
  \ev\Big[ X_{r}(t) \ \big| \setH(t),\setE(L,t)\Big] \leq \frac{O(z^{-1})}{1-\eps_{r}(L,z)} =O(z^{-1}).
\end{align}
Thus, given that $\setE(L,z)$ occurs and using \eqref{eq:ehat:bound}, we have
\begin{align}
  |e(t)| \leq O(z^2 \varep)+ \ev\Big[ X_{r}(t) \ \big| \setH(t),\setE(L,z)\Big] =O(z^{-1}).
\end{align}

For the r.v. $B(t)$ in \eqref{eq:diffeq:1}, using Lemma~\ref{lemma:conditional}, and
 assuming $L> L_{0}(z)$ so that \eqref{ineq:epsr:Llower} holds, we also have that 
\begin{align}
  \ev\Big[B(t)^{2} \ \big|\setH(t),\setE(L,z)\Big] &\leq  \frac{\ev\Big[B(t)^{2} \ \big|\setH(t)\Big]}{1-\eps_{r}(L,z)}
\nonumber \\ & \leq 2 \big[O(zL^{-\zeta})+O(z^{3}\varep)\big]
.
\end{align}
Hence,
\begin{align}
  \ev\Big[B(t)^{2} \ \big|\setH(t),\setE(L,z)\Big]=O(zL^{-\zeta})+O(z^{3}\varep).
\end{align}
 
We can repeat the same arguments for the difference equations for $\theta_{L}^{(fg)}(t)$, $\theta_{L}^{(f)}(t)$, and 
$\tilde{\theta}_{L}(t)$ containing the terms $X^{(fg)}_{r}(t)$, $X^{(f)}_{r}(t)$, and $\tilde{X}_{r}(t)$ and $\tilde{X}_{h}(t)$, respectively,
 given in Lemmas~\ref{lemma:diff:rho:fg}-\ref{lemma:diff:rho:tilde}, respectively. Doing so and using \eqref{eq:RL:delta}, and applying the definition of $\setE(L,z)$ to upperbound $\theta_{L}^{(f)}(t)$ in Lemma~\ref{lemma:diff:rho:fg}-\ref{lemma:diff:rho:tilde}, and also to upperbound 
$\theta_{L,h}(t)$ in
Lemma~\ref{lemma:diff:rho:tilde}, given $\setH(t)$ and $\setE(L,z)$,
we obtain the following vector difference equation for $z>z_{R}$ and $L>\max(L_{0}(z),L_{R})$:
\beqa
  \Delta[\boldsymbol{\theta}_{L}(t)]=\bfA \boldsymbol{\theta}_{L}(t)+f(\boldsymbol{\theta}_{L}(t))+\mathbf{e}(t)+\mathbf{B}(t) 
\label{eq:difference:eqs}
\eeqa
where
\beqa
&& \boldsymbol{\theta}_{L}(t)=\!\!\!
\begin{bmatrix}
\theta_{L}(t) \\
\theta^{(fg)}_{L}(t) \\
\theta^{(f)}_{L}(t) \\
\tilde{\theta}_{L}(t)
\end{bmatrix}, 
 f( \boldsymbol{\theta}_{L}(t))=\!\!\!\begin{bmatrix}
0 \\
 c_{R} \delta_{L}(t)^{3} \\
0 \\
-c_{R} \delta_{L}(t)^{3}
\end{bmatrix} ,
\nonumber \\ 
&& 
\mathbf{e}(t)=\!\!\! \begin{bmatrix}
  e(t) \\
 e^{(fg)}(t) \\
e^{(f)}(t) \\
\tilde{e}(t)
\end{bmatrix} , 
\mathbf{B}(t)=
\begin{bmatrix}
  B(t) \\
  B^{(fg)}(t)\\
  B^{(f)}(t) \\
    \tilde{B}(t)  \\ 
\end{bmatrix}
, \label{eq:difference:vecs}
\eeqa
and
\begin{eqnarray}
  \bfA=\begin{bmatrix}
-1 & 3z & z & z \\
0 & -3z & 0 & 0 \\
0 & 2z & -z & 0 \\
1 & 0 & 0& -z 
 \end{bmatrix}.
\label{eq:difference:def:A}
\end{eqnarray}
Given $\setH(t)$ and $\setE(L,z)$, we have that w.p.1
\begin{align}
&|e(t)|=O(z^{-1}), 
\
 |e^{(fg)}(t)|=O(z^{-1}), 
\nonumber \\ 
&
|e^{(f)}(t)|=O(z^{-1}),
\ 
|\tilde{e}(t)|=O(z^{-1}  ).
\label{eq:es:bound}
\end{align}
Moreover, we have
\beqa
&&\ev\Big[ B(t)^{2}\ \big| \setH(t), \setE(L,z)\Big]\leq O( zL^{-\zeta}) + O(z^3 \varep), 
\nonumber \\
&&\ev\Big[ B^{(fg)}(t)^{2}\ \big| \setH(t), \setE(L,z)\Big]\leq O( zL^{-\zeta}) + O(z^3 \varep), 
\nonumber \\ & & \ev\Big[B^{(f)}(t)^{2} \ \big|\setH(t), \setE(L,z)\Big]\leq  [O( zL^{-\zeta}) + O(z^3 \varep),
\nonumber \\ & & \ev\Big[\tilde{B}(t)^{2}\ \big|\setH(t), \setE(L,z)\Big]\leq  O( zL^{-\zeta}) + O(z^3 \varep)
\label{eq:Bs:hat:ex}.
\eeqa

This concludes our analysis to obtain the difference equations for the evolution of the defined densities. In the next, 
we define a deterministic ODE as the counterpart of the above difference equation. 

\subsection{Deterministic Differential Equations}\label{subsec:deter:eqs}
In the previous section, we obtained the difference equation \eqref{eq:difference:eqs} for evolution of the vector density
$\bftheta(t)$ from time $t$ to time $t+\varep$. Using this difference equation, we can obtain 
a deterministic ODE by letting 
$$\|\mathbf{e}(t)\| \equiv 0, \ \|\mathbf{B}(t)\| \equiv 0 ,$$
and taking the limit of $\varepsilon$ approach zero while keeping $z$ fixed.
Doing so, and replacing $\theta_{L}(t)$, $\theta^{(fg)}_{L}(t)$, $\theta^{(f)}_{L}(t)$, 
and $\tilde{\theta}_{L}(t)$,
 with $x_{1}(t)$, $x_{2}(t)$, $x_{3}(t)$, and $x_{4}(t)$, respectively, and replacing $\delta_{L}(t)$ with
\beqa 
\delta_{\mathbf{x}}(t)=0.5-x_{1}(t), \label{def:hat:delta:bfx}
\eeqa
 we obtain
\begin{eqnarray}
  \frac{d}{dt}\mathbf{x}=\bfA\mathbf{x}+f(\mathbf{x}), \ 
f(\mathbf{x})=\begin{bmatrix}
0 \\
c_{R} \delta_{\mathbf{x}}^{3}\\
0\\
-c_{R}\delta_{\mathbf{x}}^{3}
\end{bmatrix}, \label{eq:diff:cont}
\end{eqnarray}
with
\begin{align}
\mathbf{x}_{t_{0}}=\mathbf{x}(t_{0})=\boldsymbol{\theta}_{L}(t_{0})  
\label{eq:init:xt0:rhot0}
\end{align}
 as the initial condition for $\bfx(t)$ at time $t_{0}$.


The solution $\bfx(t)$ to the above ODE will be used to characterize the evolution of $\bftheta(t)$ over time. 
However, we need to study the properties of $\bfx(t)$ itself. To do so, we define 
\begin{align}
  \bfy(t) =
\begin{bmatrix}
  y_{1}(t) \\
  y_{2}(t) \\
  y_{3}(t) \\
  y_{4}(t)
\end{bmatrix}
\end{align}
 to be the solution to the
following ODE
\begin{eqnarray}
  \frac{d}{dt}\mathbf{y}=\tilde{\bfA}\mathbf{y}+\tilde{f}(\mathbf{y}) \label{eq:diff:y}
\end{eqnarray}
where
\begin{eqnarray}
  \tilde{\bfA}=\begin{bmatrix}
0 & 0 & 0 & 0 \\
0 & -3z & 0 & 0 \\
0 & 2z & -z & 0 \\
1 & 0 & 0& -z 
 \end{bmatrix}, 
\
  \tilde{f}(\mathbf{y})=\begin{bmatrix}
\frac{2}{3}c_{R}\delta_{\bfy}^{3}\\
 c_{R}\delta_{\bfy}^{3}\\
0\\
-c_{R}\delta_{\bfy}^{3}
\end{bmatrix}\label{eq:def:tildef}
\end{eqnarray}
with $\delta_{\bfy}(t)$ defined as 
\begin{align}
  \delta_{\bfy}(t)=0.5-y_{1}(t).\label{eq:def:delta:hat:y}
\end{align}
The above ODE has a simpler structure compared to the ODE in \eqref{eq:diff:cont} since in \eqref{eq:diff:y}, given an initial condition, 
$y_{1}(t)$
can be determined independent of $y_{2}(t)$, $y_{3}(t)$, and $y_{4}(t)$. Having found the solution for $y_{1}(t)$, one can 
also find the solutions for $y_{2}(t)$, $y_{3}(t)$, and $y_{4}(t)$.

In order to state how $\bfy(t)$ relates to $\bfx(t)$, we first introduce two definitions.
We define $\setD_{w}$ as
\begin{eqnarray}
  \mathcal{D}_{w} &= \{ \mathbf{x}: 0 \leq x_{i} \leq 0.5, i=1, \cdots, 4. \} .
\label{eq:assum:well}
\end{eqnarray}
Given $\setD_{w}$, we define a solution $\mathbf{x}(t)$ to the ODE in \eqref{eq:diff:cont} to be \emph{well-defined} if for
 an initial condition $\mathbf{x}_{t_{0}}$ with 
$\mathbf{x}_{t_{0}} \in \mathcal{D}_{w}$, we also have that
$$\mathbf{x}(t)\in \mathcal{D}_{w}, \ t \in [t_{0},\tau].$$
In a similar manner, we define $\bfy(t)$ to be well-defined if given an initial value for $\bfy(t)$ in $\setD_{w}$, 
we have that $\bfy(t)$ stays in $\setD_{w}$.

Lemma~\ref{lemma:t1} states how $\bfy(t)$ and $\bfx(t)$ are related, which will be used in the final step of the analysis in 
Appendix~\ref{subsec:final:step}.


\subsection{Final Step}\label{subsec:final:step}

In this section, we provide the final step of the proof for Theorem~1. 
The goal is to characterize how 
fast $\delta_{L}(t)$ diminishes as a function of time $t$ in the time-interval of interest $(0,\tau]$. Here is the sketch of the final step 
of the proof.
We first characterize the initial conditions that will be used to define $\bfx(t)$ and $\bfy(t)$ as the solutions to
\eqref{eq:diff:cont} and \eqref{eq:diff:y}, respectively. We 
also define $\retheta(t=n\varep)$ as an approximation to $\bfx(t=n\varep)$, where $n$ is 
a non-negative integer. We 
next show that $\bftheta(t)$ is \emph{close} to $\bfy(t)$ using the following inequality 
\begin{align}
  \| \boldsymbol{\theta}_{L}(t) -\mathbf{y}(t) \| &\leq   \| \boldsymbol{\theta}_{L}(t) -\bftheta(n_{t}\varep) \|
\nonumber \\ &\  +
  \|\boldsymbol{\theta}_{L}(n_{t}\varep) - \retheta(n_{t}\varep)  \|
\nonumber \\ & \ +
\|\retheta(n_{t}\varep)-\mathbf{x}(n_{t}\varepsilon)\|
\nonumber \\ & \ +\| \mathbf{x}(n_{t}\varepsilon) -\mathbf{y}(n_{t}\varepsilon) \|
\nonumber \\ & \ +\| \mathbf{y}(n_{t}\varepsilon) -\mathbf{y}(t) \| 
\label{ineq:master}
\end{align}
where for $t\geq 0$,
$$n_{t}= \Big \lfloor \frac{t}{\varep} \Big\rfloor. $$
 We properly upperbound each term on the RHS of \eqref{ineq:master}, and then use
properties of $\bfy(t)$ to complete the proof of the theorem.

To start, we choose $L$ and $z$ such that 
\begin{align}
z>\max(1,z_{R}) \label{ineq:zRf}  
\end{align}
and 
\begin{align}
L>\max(L_{0}(z),L_{R}) \label{ineq:LRf}  
\end{align}
where $L_{0}(z)$, $L_{R}$, and $z_{R}$ are all defined in the previous section.

Recall that by \eqref{eq:lim:prob:setE}, for $z>1$, the event $\setE(L,z)$ defined in \eqref{eq:def:E:Ltau}
 occurs with probability approaching one as $L$
approaches infinity. The statement in Theorem~1 is in the limit of first letting $L$ approach infinity and then letting
$z$ approach infinity. Hence, in the rest, we assume that event $\setE(L,z)$ occurs. Assuming that event $\setE(L,z)$
occurs, by \eqref{eq:def:E:Ltau}, we have that 
\begin{align}
 0<c_{\theta,1}e^{-t_{0}} < \theta_{L}(t_{0})<0.5-c_{\theta,2}e^{-t_{0}}, \label{ineq:fin:init:rho:t0}
\end{align}
and that for all $t\in[t_{0},\tau]$
\begin{align}
 c_{\theta,1} e^{-\tau}< \theta_{L}(t)<0.5-c_{\theta,2}e^{-\tau}, \label{ineq:fin:init:rho}
\end{align}
and also
\begin{align}
  \sup_{t\in [t_{0},\tau]} \max\Big[ \theta_{L}^{(fg)}(t), \theta_{L}^{(f)}(t), \tilde{\theta}_{L}(t)\Big] <c_{\theta}z^{-1}
\label{ineq:fin:init:rhos}
\end{align}
where $t_{0}=1$ by \eqref{eq:def:t0}.

The inequalities in \eqref{ineq:fin:init:rho} imply that the inequality \eqref{ineq:ineq:assum} holds which along with 
\eqref{ineq:zRf} and \eqref{ineq:LRf} ensures that the difference equation of 
\eqref{eq:difference:eqs} holds for all $t\in [t_{0},\tau]$.
Moreover, the inequalities \eqref{ineq:fin:init:rho:t0} and \eqref{ineq:fin:init:rhos} at time $t_{0}$ and equality 
\eqref{eq:init:xt0:rhot0} provide constraints on the initial condition of the ODE in 
\eqref{eq:diff:cont} that define $\bfx(t)$. We will use these constraints to define the initial condition for $\bfy(t)$ and to
properly upperbound each term on the RHS of \eqref{ineq:master}.

Having characterized the constraints on the initial condition for $\bfx(t)$, by Lemma~\ref{lemma:t1}, there exists a $t_{1}$ independent of $z$, $t_{0}\leq t_{1}<2$ such that
\begin{align}
\sup_{t\in [t_{1},\tau]}  \| \bfx(t)-\bfy(t)\|=O(z^{-1}) ,\quad  \text{as $z \to \infty$},
\label{eq:fin:dis:x:y}
\end{align}
where $\bfy(t)$ is the solution to \eqref{eq:diff:y} with an initial condition at time $t_{1}$ such that
$$\bfy(t_{1})=\bfx(t_{1}).$$
By Lemma~\ref{lemma:t1}, we also have that regardless of $t_{1}$, for all $t\in[t_{1},\tau]$
\begin{align}
  \left|\frac{d}{dt}\bfy(t)\right|=O(1)\mathbf{1}_{4\times 1},\ \text{ as $z\to \infty$}, \label{eq:fin:der:y}
\end{align}
and
\begin{align}
  y_{1}(t)> 0.5 - \frac{C_{y}}{\sqrt{t}} \label{ineq:fin:y1:b}
\end{align}
where $C_{y}>0$ is a constant independent of $t_{1}$, $\tau$, and $z$.

 By \eqref{eq:fin:der:y}, and that 
$$|t-n_{t}\varepsilon|\leq \varepsilon,$$
we have also have that 
\begin{align}
 \sup_{t\in [t_{1}+\varep,\tau]}\| \mathbf{y}(n_{t}\varepsilon)-\mathbf{y}(t) \| = O(1) \varep, \ \text{as $z\to \infty$}. 
\label{eq:fin:dis:yd:yc}
\end{align}

Define the index $n$ to be such that 
$$ \ n\in \{t_{0}\varep^{-1},t_{0}\varep^{-1}+1,t_{0}\varep^{-1}+2,\cdots, \tau \varep^{-1}\}.$$ 
 Define 
$$\retheta(t=n\varepsilon)$$
 to 
be the deterministic solution to the difference equation of \eqref{eq:difference:eqs} with
$$\|\mathbf{e}(t)\|\equiv 0, \ \|\mathbf{B}(t)\|\equiv 0.$$
Considering the ODE in \eqref{eq:diff:cont}, we have that 
$\boldsymbol{\theta}_{\mathbf{e}=\mathbf{0},\mathbf{B}=\mathbf{0}}(t=n\varepsilon)$ is a discrete-time 
approximation of $\mathbf{x}(t)$ at times $t=n\varep$.
By Lemma~\ref{lemma:t1}, for $z>z_{0}'$, where $z_{0}'$ is a constant, and for $t\in[t_{0},\tau]$, we have that
$\mathbf{x}(t)\in\mathcal{D}_{w}$ where $\setD_{w}$ is defined in \eqref{eq:assum:well}. Therefore, the solution $\mathbf{x}(t)$ is bounded over
 $t\in[t_{0},\tau]$. For a fixed $z>z_{0}'$, this boundedness and that $f(\bfx)$ has continuous first partial derivatives
provide sufficient conditions to have (e.g., see
 Theorem~1.16 in \cite{markley:book})
\begin{eqnarray}
  \sup_{\frac{t_{0}}{\varep} \leq n \leq \frac{\tau}{\varep} }\|\retheta(n\varepsilon)-\mathbf{x}(n\varepsilon)\|=O(\varepsilon), \
\text{as $L \to \infty$}.\label{eq:dis:rho:e:b:0}
\end{eqnarray}

Since we started by assuming that the event $\setE(L,z)$ occurs, by Lemma~\ref{lemma:rho:stoch}, we also have that for any $\eps_{\theta}>0$
\begin{align}
& \liminf_{z\to \infty} \liminf_{L\to \infty}
  \nonumber \\ & \qquad  P\bigg[ \sup_{\frac{t_0}{\varepsilon} \leq n \leq \frac{\tau}{\varepsilon} }
\big \|  \boldsymbol{\theta}_{L}(n\varep) -\retheta(n\varep)  \big\| \leq \eps_{\theta}\bigg] = 1 .
\label{eq:fin:dis:rho:reB:d}
\end{align}

Finally by Lemma~\ref{lemma:cont:disc}, we have that for $z>1$
\begin{align}
 \lim_{L\to \infty} \sup_{t\in[t_{0},\tau] } \|\boldsymbol{\theta}_{L}(t) -\bftheta(n_{t}\varep)\|=0 \ (\text{in prob.})
\label{eq:fin:dis:rhod:rhoc}
\end{align}

Using \eqref{ineq:master}, \eqref{eq:fin:dis:x:y},
\eqref{eq:fin:dis:yd:yc}, \eqref{eq:dis:rho:e:b:0},
\eqref{eq:fin:dis:rho:reB:d}, and \eqref{eq:fin:dis:rhod:rhoc}, and noting that by \eqref{ineq:assum:tau} and Lemma~\ref{lemma:t1}
$$t_{1}<2<\tau, $$
 we then have that for any $\eta_{1}>0$
\begin{align}
\liminf_{z\to \infty} \ \liminf_{L\to \infty}
    P\bigg[ \sup_{t\in[2,\tau] }
\| \boldsymbol{\theta}_{L}(t) -\mathbf{y}(t) \| \leq \eta_1 \bigg] = 1 .
\end{align}
The above limit implies that
\begin{align}
 \liminf_{z\to \infty} \ \liminf_{L\to \infty}
  P\bigg[ \sup_{t\in[2,\tau] }
  \big|\theta_{L}(t)-y_{1}(t)\big| \leq \eta_1\bigg] = 1 .
\end{align}

By \eqref{eq:def:delta:l:n}, and recovering the dependencies on $z$, we have 
 $$\delta_{L}(t,z)=\delta_{L}(t)=0.5-\theta_{L}(t).$$
By the preceding limit, we then have that
\begin{align}
& \liminf_{z\to \infty} \ \liminf_{L\to \infty}
\nonumber \\ & \qquad 
  P\bigg[ \sup_{t\in[2,\tau] }
  \Big|\delta_{L}(t,z)-[0.5-y_{1}(t)]\Big| \leq \eta_1\bigg] = 1 ,
\end{align}
which along with \eqref{ineq:fin:y1:b} implies that
\begin{align}
 \liminf_{z\to \infty} \ \liminf_{L\to \infty}
 \ P\bigg[ \sup_{ t\in [2,\tau]} 
 \Big[ \delta_{L}(t,z)-\frac{C_{y}}{\sqrt{t}}\Big] \leq \eta_{1}\bigg] =1 .
\end{align}
Since the choice for $\eta_{1}>0$ is arbitrary, we can let
$$\eta_{1}=\frac{C_{y}}{\sqrt{\tau}},$$
which leads to
\begin{align}
 \liminf_{z\to \infty} \ \liminf_{L\to \infty}
\   P\bigg[ \sup_{t\in[2,\tau]} \Big [
  \delta_{L}(t,z)-\frac{2C_{y}}{\sqrt{t}} \Big]\leq 0\bigg] = 1 .
\end{align}
By defining 
\begin{align}
C_{1}=\max\big(\sqrt{2},2C_{y}\big), \label{eq:def:C1:fi}  
\end{align}
we can extend the above statement for all $t\in(0,\tau]$ since 
always $ \delta_{L}(t,z)\leq 0.5 <1 $. This completes the proof of the theorem.


\section{Proof of Theorem~2}\label{sec:proof:t2}
Here, we provide a drift analysis to prove Theorem~2. By the assumption in the theorem, 
we consider U-CSMA policy with the unlocking period $T$, as described in Section~\ref{sec:approach} in which
transmission patterns are unlocked at times 
\begin{align}
T_{i}=iT,  \ i\in\{0,1,2,\cdots\}. \label{eq:unlock:T}  
\end{align}

To simplify the presentation, in the rest, we drop the dependency of $Q_{l}(t,z,L)$ on $z$ and $L$.
For any link $l$ and $t_{2}>t_{1}\geq0$, its queue size $Q_{l}(t)$ evolves according to the following:
\beqa
Q_{l}(t_{2})=Q_{l}(t_{1})
-\int_{t_{1}}^{t_{2}}\!\!\!\!D_{l}(t)\mathbf{1}_{Q_{l}(t)>0} \ dt 
+A_{l}(t_{1}, t_{2}), \label{eq:q:update}
\eeqa
where 
$$D_{l}(t)\in\{0,1\},$$
 and $A_{l}(t_{1}, t_{2})$ is the number of packets that arrive to link $l$ in the time interval $(t_{1},t_{2}]$ as defined 
in Section~\ref{sec:mod:lattice:uaa}.
By the assumption in the theorem and the definition of $\eps(\lambda)$ in \eqref{eq:def:eps:load}, 
the average packet arrival rate to queue $Q_{l}(t)$ is
 \beqa
\lambda_{l}=\lambda= (1-\eps) \mu_{max}(L)
\label{eq:exp:arrival}
\eeqa
 where $0<\eps(\lambda)<1$. 
By \eqref{ineq:max:arrival:l}, we also have that w.p.1. 
\beqa
A_{l}(t,t+1)\leq A_{max}. \label{eq:a:max}
\eeqa
Without loss of generality, for constant $A_{max}$, we let
$$A_{max}=1,$$
which can be used to show that for $t\geq 0$,
\beqa
|Q_{l}(t+1)-Q_{l}(t)|\leq  \max(A_{max},1)=1. \label{ineq:Q:diff}
\eeqa

As stated in the theorem, we choose the unlocking period $T$ as follows
\begin{align}
  T=T(\lambda)= \bigg\lceil \frac{(16C_{1})^{2}}{\eps^{2}} \bigg \rceil.
\label{eq:T:value}
\end{align}
By \eqref{eq:def:C1:fi}, we have $C_{1}>1$, and hence, for any $\eps$, $0< \eps<1$, we have
\begin{align}
  T>1.
\end{align}

Let 
\begin{align}
  k=k_{\frac{\eps}{16}} \label{eq:def:k:eps3}
\end{align} 
where $k_{\eps}$ is defined by the convergence property of the arrival process for link $l$ in \eqref{eq:def:arrival:exp}. Since $T> 1$ and
$$ T_{(i+1)k}-T_{ik} =kT > k,$$
by the definition of $k$, we have that
\begin{align}
\left |  \ev\Big [ \frac{1}{kT} A_{l}(T_{ik}, T_{(i+1)k}) \ \big |\setH_{s}(T_{ik})\Big ]-\lambda \right | < \frac{\eps}{16}
\label{ineq:Al:k}
\end{align}
where $\setH_{s}(t)$ is the system history up to and including time $T_{ik}$.

We also need to determine $z$ and $L$. To do so, as stated in the theorem statement, we assume that 
$$z>z(\eps, T) \text{ and } L>L(z,\eps, T ),$$ 
where $z(\eps, T)$ and $L(z,\eps, T)$ are defined in Section~\ref{sec:result:2}, so that 
\begin{align}
P \Big[\inf_{t\in[0,T]} \Big(\theta_{L}(t) - 0.5 +\frac{C_{1}}{\sqrt{t}} \Big) \geq 0\Big ] \geq 1-\eps_{p}(L,z,T) \label{ineq:prob:epsP}
\end{align}
and
\begin{align}
  \eps_{p}(L,z,T)< \frac{1}{2}\eps.\label{ineq:epsP}
\end{align}

Having determined the necessary parameters, we next proceed with the following drift analysis. Using \eqref{eq:q:update} and \eqref{ineq:Q:diff}, 
we have that 
\begin{align}
\Delta_{i}(T)&=\ev\Big[Q_{l}^{2}(T_{(i+1)k})-Q_{l}^{2}(T_{ik})\ \big|Q_{l}(T_{ik})\Big]
\nonumber \\ & \leq
2Q_{l}(T_{ik})\ev\bigg[ A_{l}(T_{ik},T_{(i+1)k}) -
\nonumber \\ & \qquad \int_{T_{ik}}^{T_{(i+1)k}} D_{l}(t) \mathbf{1}_{Q_{l}(t)>0}\ dt \ \Big|Q_{l}(T_{ik})\bigg ]+  (kT)^{2}.\label{eq:drift:l}
\end{align}
If $Q_{l}(T_{ik})>kT$, by \eqref{ineq:Q:diff}, it then follows that w.p.1.,
\begin{align}
  \mathbf{1}_{Q_{l(t)}>0}=1, \ \ t \in[T_{ik}, T_{(i+1)k}] .\label{eq:ind:Qb}
\end{align}
For $Q_{l}(T_{ik})\leq kT$, by \eqref{ineq:Q:diff}, w.p.1, $Q_{l}(T_{(i+1)k})\leq 2kT$, and hence,
\begin{align}
  \Delta_{i}(T) \leq 4 (kT)^{2}, \label{ineq:Di:Ql}
\end{align}
In addition, for $Q_{l}(T_{ik})\leq kT$, w.p.1.,
\begin{align}
 Q_{l}(T_{ik})  \int_{T_{ik}}^{T_{(i+1)k}} D_{l}(t) \mathbf{1}_{Q_{l}(t)>0}\ dt \leq (kT)^{2}. \label{ineq:int:Dl:b:Ql}
\end{align}
Dividing $\Delta_{i}(T)$ to two expectations, one conditioned on $Q_{l}(T_{ik})\leq kT$ and the other
on $Q_{l}(T_{ik})> kT$, and using the results in \eqref{eq:drift:l}-\eqref{ineq:int:Dl:b:Ql}, we can show that 
\begin{align}
\Delta_{i}(T)&=\ev[Q_{l}^{2}(T_{(i+1)k})-Q_{l}^{2}(T_{ik})\ |Q_{l}(T_{ik})]
\nonumber \\ & \leq
2Q_{l}(T_{ik})\ev\bigg[ A_{l}(T_{ik},T_{(i+1)k}) -
\nonumber \\ & \qquad \int_{T_{ik}}^{T_{(i+1)k}} D_{l}(t) \ dt \ \big|Q_{l}(T_{ik})\bigg ]+ 5 (kT)^{2}.\label{eq:drift:l:f}
\end{align}

We next derive an upperbound for the RHS of \eqref{eq:drift:l:f}. By \eqref{ineq:Al:k}, we have that
\begin{align}
   \ev\Big [  A_{l}(T_{ik}, T_{(i+1)k}) \ \big |Q_{l}(T_{ik})\Big ] < kT\big[\lambda + \frac{\eps}{16} \big].
\label{ineq:drift:arr}
\end{align}
To obtain a lowerbound for 
\begin{align}
\ev\Big[  \int_{T_{ik}}^{T_{(i+1)k}} D_{l}(t) \ dt \ \big|Q_{l}(T_{ik})\Big ],
\end{align}
we use the conditions considered for $z$ and $L$ which state that \eqref{ineq:prob:epsP} and \eqref{ineq:epsP} 
hold. Since the unlocking mechanism restarts the network at time $T_{i}$, $i\geq 0$, the CSMA transmission events after 
this time are independent of the CSMA transmission events before this time. Hence, \eqref{ineq:prob:epsP} and 
\eqref{ineq:epsP} hold for any time interval $(T_{i},T_{i+1}]$, $i\geq 0$, independent of $\setH_{s}(T_{i})$, 
which can be used along with $C_{1}>1$ to show that
\begin{align}
&\ev\bigg[  \frac{1}{T} \int_{T_{i}}^{T_{i+1}}\theta_{L}(t)dt  \ \Big| \setH_{s}(T_{i}) \bigg ] \geq  
\ev\bigg [\frac{1}{T} \int_{T_{i}\ + \ 1}^{T_{i+1}}\theta_{L}(t) dt \ \Big| \setH_{s}(T_{i}) \bigg ]
\nonumber \\ 
& \qquad \qquad  \geq \big(1-\eps_{p}(L,z,T)\big) \Big[0.5 \frac{T-1}{T}- 2C_{1}\frac{\sqrt{T} -1}{T} \Big]
\nonumber \\ 
&  \qquad \qquad  > 0.5 -0.5 \eps_{p}(L,z,T)-\big[1-\eps_{p}(L,z)\big]\frac{2C_{1}}{\sqrt{T}}
\nonumber \\ 
&  \qquad \qquad >  0.5-\frac{\eps}{4}-\frac{2C_{1}}{\sqrt{T}}.
\end{align}
By the choice for $T$ as given in \eqref{eq:T:value}, we then have that
\begin{align}
\ev\bigg[  \frac{1}{T} \int_{T_{i}}^{T_{i+1}}\theta_{L}(t)dt  \ \Big| \setH_{s}(T_{i}) \bigg ]  > 0.5-\frac{\eps}{4}-\frac{\eps}{8}=0.5 -\frac{3\eps}{8}.
\label{ineq:ave:rho:pT}
\end{align}

Moreover,
since the 2D torus is symmetric with respect to link positions, we have that the average 
transmission, over one unlocking period, for any link $l$ is the same as the average transmission for any other link $l'$. This allows us
to write
\begin{align}
&\ev \left[  \frac{1}{T} \int_{T_{i}}^{T_{i+1}}\theta_{L}(t) dt \ \Big| \setH_{s}(T_i) \right]
\nonumber \\ & \qquad =\ev \left [  \frac{1}{T} \int_{T_{i}}^{T_{i+1}} \frac{1}{L} \sum_{l\in\setL} D_{l}(t) dt \ \Big| \setH_{s}(T_i) \right]
\nonumber \\ & \qquad
= \frac{1}{L}  \sum_{l\in\setL} \ev \left [  \frac{1}{T} \int_{T_{i}}^{T_{i+1}} D_{l}(t) dt\ \Big| \setH_{s}(T_i) \right]
\nonumber \\ & \qquad 
=\ev \left [   \frac{1}{T} \int_{T_{i} }^{T_{i+1}} D_{l}(t) dt \ \Big| \setH_{s}(T_i) \right]
\end{align}
for any $l\in\setL$.
By the above equality and \eqref{ineq:ave:rho:pT}, we obtain
\begin{align}
  \ev \left [   \frac{1}{T} \int_{T_{i} }^{T_{i+1}} D_{l}(t) dt \ \big| \setH_{s}(T_i) \right]>0.5 -\frac{3}{8}\eps.\label{ineq:drift:dep}
\end{align}

Using \eqref{eq:exp:arrival}, \eqref{eq:drift:l:f}, \eqref{ineq:drift:arr}, \eqref{ineq:drift:dep}, 
and that $\mu_{max}(L)\leq 0.5$, we have that
\begin{align}
\Delta_{i}(T) & \leq 2 Q(T_{ik})\Big[kT (\lambda+\frac{1}{16}\eps) -kT (0.5-\frac{3}{8}\eps)\Big]+ 5(kT)^{2}
\nonumber \\
&  = -\eps'kT Q(T_{ik}) + 5(kT)^{2}
\end{align}
where 
\begin{align}
  \eps'=\frac{1}{8}\eps<1. \label{eq:def:eps':ll}
\end{align}
Summing over $i=\{0, \cdots, J-1\}$, and then taking the expected value, we obtain
\begin{align}
&\ev\Big[\sum_{i=0}^{J-1}\Delta_{i}(T)\Big]=\ev\Big[Q_{l}^{2}(T_{Jk})-Q_{l}^{2}(T_{0}=0)\Big]
\nonumber \\ & \qquad \qquad \leq -\eps' \ev\Big[\sum_{i=0}^{J-1} kT Q_{l}(T_{ik})\Big]+ J 5 (kT)^{2} ,\label{ineq:exp:drift}  
\end{align}

 Rearranging terms in \eqref{ineq:exp:drift}, dividing by
$JkT\eps'$, letting $J$ approach infinity, and assuming queues are initially bounded, we obtain
\begin{align}
& \limsup_{J \to \infty}\frac{1}{JkT}\ev\Big[\sum_{i=0}^{J-1} kT Q_{l}(T_{ik})\Big]
\nonumber \\ 
& \qquad \qquad 
\leq \limsup_{J \to \infty} \frac{1}{JkT\eps'}\Big(\ev\big[Q_{l}^{2}(0)\big]+ J 5 (kT)^{2} \Big)  =  
\frac{5kT}{\eps'}. \label{ineq:sum:q}
\end{align}

From \eqref{ineq:Q:diff} and that $k\geq 1$ and $T>1$, we obtain
\beqa
\int_{T_{ik}}^{T_{(i+1)k}} Q_{l}(t) \ dt &\leq&  \int_{T_{ik}}^{T_{(i+1)k}}  \Big[Q_{l}(T_{ik})+(\lceil t \rceil -T_{ik})\Big] \ dt 
\nonumber \\ &\leq&
 kT Q_{l}(T_{ik})+2(kT)^{2}. \label{ineq:int:Q:dis:l}
\eeqa
For $t> 0$, define
$$i_{t}=\Big\lceil \frac{t}{kT}  \Big\rceil.$$
By the non-negativity of queue sizes and that $t\leq T_{i_{t}k}$, we have that
\begin{align}
\frac{1}{t} \int_{0}^{t} Q_{l}(t) \ dt \leq \frac{i_{t}kT}{t}
 \frac{1}{i_{t} kT}\int_{0}^{T_{i_{t} k} } Q_{l}(t)  \ dt,
\end{align}
which along with \eqref{ineq:int:Q:dis:l} leads to
\begin{align}
  \frac{1}{t} \int_{0}^{t} Q_{l}(t) \ dt \leq  \frac{i_{t}kT}{t}
 \frac{1}{i_{t} kT}\bigg[\Big(\sum_{m=0}^{i_{t}-1} kT Q_{l}(T_{mk})\Big) +  i_{t}2(kT)^{2} \bigg]
\label{ineq:int:pre:sup}
\end{align}
Using \eqref{ineq:sum:q}, taking the expected value of both sides of \eqref{ineq:int:pre:sup} and then the $\limsup$ as $t$ approaches infinity, and using the limit 
$$\lim_{t\to \infty}\frac{i_{t}kT}{t} =\lim_{t\to \infty}\frac{\lceil \frac{t}{kT} \rceil kT}{t} =1, $$ 
we obtain
\begin{align}
\limsup_{t\to \infty } \ev\left[\frac{1}{t} \int_{0}^{t} Q_{l}(t) \ dt\right] 
\leq   5\frac{kT}{\eps'}+2kT.
\end{align}
Using \eqref{eq:T:value} and \eqref{eq:def:eps':ll}, we then have that 
\begin{align}
\limsup_{t\to \infty } \ev\left[\frac{1}{t} \int_{0}^{t} Q_{l}(t) \ dt\right] 
< \frac{56kT}{\eps} = \frac{C_{2}k}{\eps^{3}}
\end{align}
where the constant $C_{2}$ is given by
$$C_{2}=56(16C_{1})^{2}.$$
Recalling that $k$ is defined in \eqref{eq:def:k:eps3}, we have completed the proof of Theorem~\ref{result:2}.


\section{ Lemmas}\label{appendix:lemma}

\begin{lemma}\label{lemma:global:bound}
  Consider $\tau>t_{0}$ and suppose $z>1$. There exist positive constants $c_{\theta,1}$ and $c_{\theta,2}$, independent of $\tau$, $z$, and $L$, such that 
\begin{align}
&\lim_{L\to \infty } P\bigg[ \Big(c_{\theta,1} e^{-\tau} < \inf_{t\in[t_{0},\tau]} \theta_{L}(t)\Big)
\nonumber \\ & \qquad \qquad 
 \cap  \Big(\sup_{t\in[t_{0},\tau]} \theta_{L}(t) < 0.5- c_{\theta,2} e^{-\tau}\Big) \bigg] =1.
\end{align}
Moreover, for $t_{0}>0$
\begin{align}
\lim_{L\to \infty } P\Big[ c_{\theta,1}e^{-t_{0}}  < \theta_{L}(t_{0}) < 0.5- c_{\theta,2}e^{-t_{0}} \Big] = 1.
\end{align}
\end{lemma}

\begin{proof}
Recall that 
$$t_{0}=1.$$
Consider link $l$, which has a maximum of four interfering neighbours in the set $\setN_{l}$. Recall that at time 
$t=0$, all links are idle. Let time $w_{l}$ be the \emph{first time} after time $t=0$ that link $l$ 
or one of its interfering links in $\setN_{l}$ starts transmitting.

We claim that at any time $t<w_{l}$, at least one of the 
links in the set $\{l\}\cup \setN_{l} $ senses the channel as idle, independent of states
of other links in $G_L$ from time $0$ to time $\tau$. To reach a 
contradiction, suppose $t<w_{l}$, and
it is true that all links in $\{l\}\cup \setN_{l} $ sense the channel as busy at time $t$. Since the
first transmission occurs at time $w_{l}$, this means that link $l$ is not transmitting
at time $t$ and finds one of its neighbours transmitting at that time. A 
neighbour transmitting at time $t$ means that by definition $w_{l}\leq t$, contradicting the assumption 
that $t<w_{l}$.

Define the set $\setL^{(e)}_{s}$ as the set of links in the $(n+1)\times (n+1)$ lattice $G_{L}$ with coordinates of the form 
$$(4i,4j), \ 0\leq 4i \leq n, \ 0 \leq 4j \leq n.$$
 The set $\setL^{(e)}_{s}$ is a subset of even links $\setL^{(e)}$. By the definition of $\setL^{(e)}_{s}$, we have that for any two links $l,l' \in \setL^{(e)}_{s}$, the sets 
$\{l\}\cup \setN_{l} $ and $\{l'\}\cup \setN_{l'} $ are disjoint.

By definition, at time $w_{l}$, for the first time link $l$ or one of its interfering links starts transmitting. Let $t_{l}$ the amount
of time that such a first transmission lasts. Since before time $w_{l}$, at least one link in the set $\{l\}\cup \setN_{l} $ 
is trying to access the channel with rate 
$z$, and since packet transmission times are independent of each other with unit rate, we have that 
\begin{align}
  & P\Big[w_{l} < \frac{t_{0}}{2}, t_{l}> \tau  \ \big|  \{w_{l'},t_{l'}, l'\neq l, l'\in\setL^{(e)}_{s}  \}\Big]
\nonumber\\  & \qquad \qquad \qquad \qquad \qquad \qquad \qquad \geq (1-e^{-0.5 z})e^{-\tau}.
\label{ineq:lb:p:wl}
\end{align}

Let $I_{\tau,l}$ be the indicator function that the for the first time after time zero link $l$ or one of its interfering 
links starts transmitting
at time $w_{l}<0.5t_{0}=0.5$,  
and that such a transmission continues until and including time $\tau$. For $z>1$, by \eqref{ineq:lb:p:wl}
we have that 
\begin{align}
   P\Big[I_{\tau,l}=1  \ \big|  \{I_{\tau,l'}, l'\neq l, l'\in\setL^{(e)}_{s}  \}\Big]\geq p_{\tau}
\end{align}
where
$$p_{\tau}=(1-e^{-0.5})e^{-\tau}>0.$$

Therefore, we have a sequence $\{I_{\tau,l}, l\in \setL^{(e)}_{s}\}$ of indicator functions that independently of each other 
will be one with probability at least $p_{\tau}$. Similar to the W.L.L.N, it then follows that 
\beqa
\lim_{L\to \infty } P\left[\frac{\sum_{l\in \setL_{s}^{(e)} }  I_{\tau,l}}{L}\geq 0.5 p_{\tau}\right] = 1.
\label{eq:limit:p:I:t0:1}
\eeqa

By definition, $\theta_{L}(t)$ is the fraction of all active links at time $t$, and 
for $t\in [t_0,\tau]$, it includes all
links $l\in \setL_s^{(e)}$ for which $ I_{\tau,l}=1$. Hence, for $t\in [t_{0},\tau]$, w.p.1
$$\theta_{L}(t) \geq \frac{\sum_{l\in \setL_{s}^{(e)} }  I_{\tau,l}}{L}.$$
Letting 
$$c_{\theta,1}=0.25 (1-e^{-0.5}),$$
 and using the above and \eqref{eq:limit:p:I:t0:1}, we obtain
\beqa
\lim_{L\to \infty } P \Big[  \inf_{t\in[t_{0},\tau]} \theta_{L}(t) > c_{\theta,1} e^{-\tau} \Big]=1.
\eeqa
Using similar arguments, we also have that 
\beqa
\lim_{L\to \infty } P \Big[\sup_{t\in[t_{0},\tau]} \theta_{L}(t) < 0.5- c_{\theta,2} e^{-\tau}\Big]=1,
\eeqa
for some constant $c_{\theta,2}>0$ independent of $\tau$, $z$, and $L$.
Combining the last two limits, we obtain the first limit in the lemma, as required. The second limit also follows 
by replacing $\tau$ with $t_{0}$ in the above discussion, completing the proof.
 \end{proof}


\begin{lemma}\label{lemma:cond:prob}
Consider two active links $l$ and $l'$ at time $t$, i.e., suppose $l,l' \in \theta_{L}(t)$. We have
  \begin{align*}
  P\Big[\mathbf{1}_{\theta_{L},l}(t+\varepsilon)=1 \  \big| \mathbf{1}_{\theta_{L},l'}(t+\varepsilon)=1, \setH(t)\Big]
\leq \varepsilon +O(z\varepsilon^2)
\end{align*}
where $\mathbf{1}_{\theta_{L},l}(t+\varepsilon)$ is the indicator function that at time 
$t+\varepsilon$, link $l$ is not transmitting, as defined in Appendix~\ref{ss:sec:diff:rho:L}, and $\setH(t)$
is the history of densities as defined in Appendix~\ref{subsec:diff:evol}.
\end{lemma}
\begin{proof}
By \eqref{eq:p:ind:r:l:varep} and \eqref{eq:abs:el1}, we have 
\begin{align}
\Big|  P\big[\mathbf{1}_{\theta_{L},l}(t+\varepsilon)=1 \ | \setH(t)  \big]-\varepsilon \Big|=O(z\varepsilon^2).
\label{ineq:p:l:stop}
\end{align}
We also have that 
\begin{align}
&   P\Big[\mathbf{1}_{\theta_{L},l}(t+\varepsilon)=1 \text{ and } \mathbf{1}_{\theta_{L},l'}(t+\varepsilon)=1 \ \big| \setH(t)\Big]
\nonumber \\ &\leq P\Big[ \text{link $l$ and $l'$ stop transmitting before $t+\varepsilon$} \ \big| \setH(t) \Big]
\nonumber \\ & = [\varepsilon-O(\varepsilon)][\varepsilon-O(\varepsilon)]\leq \varepsilon^2.
\label{ineq:p:ll':stop}
\end{align}
The first inequality follows from the fact that having 
$$\mathbf{1}_{\theta_{L},l}(t+\varepsilon)=1,$$
 by definition, 
requires link $l$ to stop transmitting before time $t+\varep$. The equality follows form \eqref{eq:prob:stop:xh} and  
the fact that packet transmission times are i.i.d over time and over links. 

Using the definition of conditional probabilities, \eqref{ineq:p:l:stop} and \eqref{ineq:p:ll':stop}, we obtain
\begin{align*}
&    P\Big[\mathbf{1}_{\theta_{L},l}(t+\varepsilon)=1 \  \big| \mathbf{1}_{\theta_{L},l'}(t+\varepsilon)=1 \ \big| \setH(t) \Big]
\nonumber \\ & \leq \frac{\varepsilon^2}{\varep-O(z\varep^2)}
 = \varep+O(z\varep^2),
\end{align*}
as required.

\end{proof}


\begin{lemma}\label{lemma:diff:rho:fg}
  For the density $\theta^{(fg)}_{L}(t)$, we have that
\begin{align}
 \Delta[\theta_{L}^{(fg)}(t)] &= \Big[R_{L}(t)-3z\theta_{L}^{(fg)}(t)+\theta^{(f)}(t) 
\nonumber \\ &\qquad  \qquad+e^{(fg)}(t)+B^{(fg)}(t)\Big],
\end{align}
where
\begin{align}
&  |e^{(fg)}(t)|=O(z^{2}\varepsilon)+  \ev[X^{(fg)}_r(t) \ | \setH(t)], 
\nonumber \\
&    \ev[X^{(fg)}_r(t) \ | \setH(t) ]\leq \hat{c}_{n}\theta_{L,r}(t) \big[ z +O(z^2  \varep)\big],
\nonumber \\
 & \ev\Big[B^{(fg)}(t)^{2}\ \big|\setH(t)\Big]\leq O(z L^{-\zeta})+O(z^{3}\varep)    ,
  \end{align}
and $\hat{c}_{n}>0$ is a constant independent of $L$ and $z$ as defined in Appendix~\ref{ss:sec:diff:rho:L}.f.
\end{lemma}

\begin{proof}
   The events that can contribute to the change in $\theta^{(fg)}_{L}(t)$ from time $t$ to time $t+\varep$
are the following events:
\begin{itemize}
\item A critical link at time $t$ associated with density $R_{L} (t)$ is inactive at time $t+\varep$.
\item An active link at time $t$ associated with density $\theta^{(f)}_{L}(t)$ is inactive at time $t+\varep$.
\item Either of three links at time $t$ associated with density $\theta^{(fg)}_{L}(t)$ is active at time $t$.
\item Multiple transitions (state-changes) from time $t$ to time $t+\varep$ done by one link or a link and the links in its $r_{n}$-neighbourhood
 that can increase or decrease the number of events associated with $\theta^{(fg)}_{L}(t+\varep)$.
\item An inactive link that at time $t$ is within $r_{n}$-neighbourhood of another inactive link associated with density
$\theta_{L,r}(t)$ is active at time $t+\varep$.

\end{itemize}
Considering the above events and taking similar steps leading to \eqref{eq:diffeq:1}-\eqref{eq:B:hat:ex}, we obtain the statement of the lemma, as required.  
\end{proof}


\begin{lemma}\label{lemma:diff:rho:f}
  For the density $\theta^{(f)}_{L}(t)$, we have that
\begin{align}
 \Delta[\theta_{L}^{(f)}(t)] &= \Big[2z\theta_{L}^{(fg)}(t)-(1+z)\theta^{(f)}(t) 
\nonumber \\ &\qquad  \qquad+e^{(f)}(t)+B^{(f)}(t)\Big],
\end{align}
where
\begin{align}
&  |e^{(f)}(t)|=O(z^{2}\varepsilon)+  \ev[X^{(f)}_r(t) \ | \setH(t)], 
\nonumber \\
&    \ev[X^{(f)}_r(t) \ | \setH(t) ]\leq \hat{c}_{n}\theta_{L,r}(t) \big[ z +O(z^2  \varep)\big]
\nonumber \\
& \ev\Big[B^{(f)}(t)^{2}\ \big|\setH(t)\Big]\leq O(z L^{-\zeta})+O(z^{3}\varep).    
\end{align}
and $\hat{c}_{n}>0$ is a constant independent of $L$ and $z$ as defined in Appendix~\ref{ss:sec:diff:rho:L}.f.
\end{lemma}

\begin{proof}
The events that can contribute to the change in $\theta^{(f)}_{L}(t)$ from time $t$ to time $t+\varep$
are the following events:
\begin{itemize}
\item Either of the two of the three inactive links at time $t$ associated with the events defined by $\theta^{(fg)}_{L}(t)$, e.g.,
 link $f$ or $g$ in Fig.~\ref{fig:lattice_net:1}, is active at time $t+\varep$.
\item An active link at time $t$ associated with density $\theta^{(f)}_{L}(t)$ is inactive at time $t+\varep$.
\item An inactive link at time $t$ associated with density $\theta^{(f)}_{L}(t)$ is active at time $t+\varep$.
\item Multiple transitions (state-changes) from time $t$ to time $t+\varep$ done by one link or a link and the links in its $r_{n}$-neighbourhood
 that can increase or decrease the number of events associated with $\theta^{(f)}_{L}(t+\varep)$.
\item An inactive link that at time $t$ is within $r_{n}$-neighbourhood of another inactive link associated with density
$\theta_{L,r}(t)$ is active at time $t+\varep$.

\end{itemize}

Considering the above events and taking similar steps leading to \eqref{eq:diffeq:1}-\eqref{eq:B:hat:ex}, we obtain the statement of the lemma, as required.   
\end{proof}


\begin{lemma}\label{lemma:diff:rho:tilde}
For the density $\tilde{\theta}_{L}(t)$, we have that
\begin{align}
 \Delta[\tilde{\theta}_{L}(t)] &= \Big[ (\theta_{L}(t)-R_{L}(t)-\theta^{(f)}(t)) -z\tilde{\theta}_{L}(t)
\nonumber \\ &\qquad  \qquad \qquad \qquad +\tilde{e}(t)+\tilde{B}(t)\Big],
\end{align}
where
\begin{align}
&  |\tilde{e}(t)|\leq O(z^{2}\varepsilon)+  \ev[\tilde{X}_r(t) \ | \setH(t) ] + \ev[\tilde{X}_h(t)\ | \setH(t)], 
 \nonumber \\ &   \ev[\tilde{X}_r(t) \ | \setH(t) ]\leq \hat{c}_{n}\theta_{L,r}(t) \big[ z +O(z^2  \varep)\big],
\nonumber \\ &    \ev[\tilde{X}_h(t)\ | \setH(t) ]\leq \hat{c}_{n}\theta_{L,h}(t) \big[ 1 +O(z  \varep)\big],
\nonumber \\ &
 \ev\Big[\tilde{B}(t)^{2}\ \big|\setH(t)\Big]\leq O(z L^{-\zeta})+O(z^{3}\varep),
  \end{align}  
and $\theta_{L,h}(t)$ is defined in Appendix~\ref{sec:all:densities}, and $\hat{c}_{n}$ is a constant independent of
$L$ and $z$ as defined in Appendix~\ref{ss:sec:diff:rho:L}.f.
\end{lemma}

\begin{proof}
  Recall that $\tilde{\theta}_{L}(t)$ (see Appendix~\ref{sec:typ:event} and Appendix~\ref{ss:sec:diff:rho:L}) represents the density of links that are inactive 
at time $t$ as a result of an ordinary event of type-I and type-II. Considering this, the following are the events that contribute to the
 change in $\tilde{\theta}_{L}(t)$ from time $t$ to time $t+\varep$:
\begin{itemize}
\item An active link $l$ at time $t$ is inactive at time $t+\varep$ where link $l$ at time $t$ 1) is
 not associated with densities $R_{L}(t)$ or $\theta^{(f)}(t)$, and 2) at time $t$ is not 
within $\frac{r_{n}}{2}$-neighbourhood of any inactive link that senses the channel as idle at time $t$.
\item An inactive link at time $t$ associated with density $\tilde{\theta}_{L}(t)$ is active at time $t+\varep$.
\item Multiple transitions (state-changes) from time $t$ to time $t+\varep$ done by one link or a link and the links in its $r_{n}$-neighbourhood
 that can increase or decrease the number of events associated with $\tilde{\theta}_{L}(t+\varep)$.
\item An inactive link that at time $t$ is within $r_{n}$-neighbourhood of another inactive link associated with density
$\theta_{L,r}(t)$ is active at time $t+\varep$.
\end{itemize}

As defined in Appendix~\ref{sec:all:densities}, $\theta_{L,h}(t)$ is the fraction of links that are inactive and sense the channel as idle
at time $t$. Considering this definition, we note the following. 
In the first event considered above, we included active links not associated with densities
$R_{L}(t)$ and $\theta^{(f)}(t)$ since these active links are already accounted for in the difference equations for 
$\theta^{(fg)}(t)$ and $\theta^{(f)}(t)$, respectively. We further included those of such active links not within $\frac{r_{n}}{2}$-neighbourhood of 
inactive links that sense the channel as idle to avoid rare events\footnote{We need to avoid rare events since the definitions of densities $\tilde{\theta}_{L,1}(t)$
and $\tilde{\theta}_{L,2}(t)$ requires active links stop transmitting without leading to rare events.}. 
By the definition of $\hat{c}_{n}$, the density of links within $\frac{r_{n}}{2}$-neighbourhood of 
inactive links that sense the channel as idle is not larger than
$$\hat{c}_{n}\theta_{L,h}(t).$$
However, if an active link $l$ is 
 within $\frac{r_{n}}{2}$-neighbourhood of 
an inactive link $l'$ that senses the channel as idle at time $t$, and 1) link $l'$ does sense the channel as busy in the time interval
$[t',t+\varep]$, for some $t'$,
$$t<t' < t+\varep,$$
 and 2) link $l$
stops transmitting at some time $t''$, 
$$t'<t''\leq t+\varep,$$
 and stays idle by time $t+\varep$, then we may have an ordinary event. These sequence of events are accounted
 for by multiple transitions event 
considered above.

Considering the above observations and taking similar steps leading to \eqref{eq:diffeq:1}-\eqref{eq:B:hat:ex}, we obtain the statement in the lemma, as required.
\end{proof}



\begin{lemma}\label{lemma:rare}
Consider any finite time-interval $[t_{1},t_{2}]$, where $0<t_{1}<t_{2}<\infty$. Then, for $z>0$,
 \begin{align*}
&\lim_{L\to \infty}   P\bigg[ \sup_{t\in [t_{1},t_{2}]} \theta_{L,r}(t) <c_{r}z^{-2} \bigg]=1
 \end{align*}
 where $c_{r}>0$ is a constant independent of $z$ and $L$.
\end{lemma}
\begin{proof}
Recall that in Section~\ref{sec:rare:new}, we defined $\theta_{L,r}(t)$ to denote the density (fraction) of links that are
inactive at time $t$, and that
in whose $r_{n}$-neighbourhood, there is another inactive link that senses the channels as idle at time $t$ such that both inactive links have remained idle until time $t$ after a rare event
with which both links are involved. 

In addition, recall that $\theta_{L,h}(t)$ is the density of links that are inactive and sense the channel as idle at time $t$. By 
\eqref{ineq:diff:ineq:rho:h} in the proof 
of Lemma~\ref{lemma:rholh}, we have that 
\begin{align}
  \Delta(\theta_{L,h}(t)) \leq \big[5-(z+5) \theta_{L,h}(t)\big]\varep+e_{h}(t)+B_{h}(t)
\label{ineq:diff:ineq:rho:h:2}
\end{align}
where $e_{h}(t)$ and $B_{h}(t)$ are properly bounded in \eqref{eq:rho:h:error:t} and \eqref{eq:rho:h:B:t}, respectively.

To obtain a difference inequality for $\theta_{L,r}(t)$, we study the change in $\theta_{L,r}(t)$ from time $t$ to time $t+\varep$. 
Consider a link $l\in \theta_{L,h}(t)$. If an active link $l'$ within $2r_{n}$-neighbourhood
of link $l$ stops transmitting at some time $t'>t$, then at most link $l$, link $l'$, and interfering links of link $l'$, which may also find the channel
as idle, may be counted as the links belonging to $\theta_{L,r}(t')$. We chose the $2r_{n}$-neighbourhood instead of $r_{n}$-neighbourhood towards obtaining an upperbound for the rate by which $\theta_{L,r}(t)$ increases. Therefore, since links stops transmitting with unit rate, and that there 
are at most $\hat{c}_{n}'$ links in the $2r_{n}$-neighbourhood of any link, where $\hat{c}_{n}'>0$ is a constant, the term
$$ 6 \hat{c}_{n}'  \theta_{L,h}(t)$$
plus an error term gives the maximum rate by which $\theta_{L,r}(t)$ increases due to active links at time $t$ that 
are within $2r_{n}$-neighbourhood of inactive links in $\theta_{L,h}(t)$ and that are idle by time $t+\varep$.

In addition, since each inactive link $l$ in $\theta_{L,r}(t)$ with rate $z$ tries to become active, the term
$$z\theta_{L,r}(t)$$ 
plus an error term gives the minimum rate by which $\theta_{L,r}(t)$ decreases due to inactive links in $\theta_{L,r}(t)$ start transmitting. 
Finally, we must consider the events in which multiple transitions occur from time $t$ to time $t+\varep$ by one link or a link and the 
links within its $r_{n}$-neighbourhood. These events can contribute to $\theta_{L,r}(t+\varep)$.

Considering the maximum and minimum rates obtained above and the events
with multiple transitions, we can
use a similar approach as taken in 
Appendix~\ref{subsec:diff:evol}, to show that
\begin{align}
\Delta[\theta_{L,r}(t)]&\leq   6 \hat{c}_{n}'\theta_{L,h}(t)-z \theta_{L,r}(t)+e_{r}(t)+B_{r}(t). \label{ineq:diff:ineq:rho:r}
\end{align}
where 
\begin{align}
|e_{r}(t)|=O(z^{2}\varepsilon),   \label{eq:e:r:bound:lem}
\end{align}
and
\begin{align}
  \ev[B_{r}(t)^{2}|\mathcal{H}(t)]=O(zL^{-\zeta})+O(z^{3}\varep). \label{eq:B:r:bound:lem}
\end{align}

 For a fixed $z$, we can obtain the deterministic differential inequality counterparts of \eqref{ineq:diff:ineq:rho:h:2} 
and \eqref{ineq:diff:ineq:rho:r} by letting
$$e_{n}(t)=B_{n}(t)=e_{r}(t)=B_{r}(t) \equiv 0,$$
 letting 
$$\varepsilon \to 0,$$
 and replacing $\theta_{L,h}(t)$ and 
$\theta_{L,r}(t)$
 with $x_{h}(t)$ and $x_{r}(t)$, respectively:
 \begin{align}
&\frac{d}{dt} x_{h}\leq 5-(z+5) x_{h}, \nonumber \\ 
&   \frac{d}{dt} x_{r}\leq 6\hat{c}_{n}' x_{h} -z x_{r}. \label{ineq:diff:ineq:rho:r:cont}
 \end{align}

By the proof of Lemma~\ref{lemma:rholh}, we can choose $t_{1}'$ independent of $z$ and $L$ with
$0<t_{1}'<t_{1}$ such that
\begin{align}
\sup_{t\in[t_{1}',t_{2}]}  x_{h}(t)\leq c_{\theta,h}' z^{-1},
\end{align}
where $c_{\theta,h}'>0$ is a constant independent of $z$ and $L$. Using this upperbound for $x_{h}(t)$ and the differential inequality for $x_{r}(t)$ in 
\eqref{ineq:diff:ineq:rho:r:cont}, we can show that
\begin{align}
\sup_{t\in [t_{1},t_{2}]}   x_{r}(t) < c_{r}'z^{-2}
\end{align}
for some constant $c_{r}'>0$ independent of $L$ and $z$.

Therefore, the solution to the deterministic differential inequality counterparts of \eqref{ineq:diff:ineq:rho:r}
is less than $c_{r}' z^{-2}$ over the entire interval $[t_{1},t_{2}]$. 
Using this result, the bounds given in \eqref{eq:e:r:bound:lem} and \eqref{eq:B:r:bound:lem}, which both approach 
zero as $L$ approaches infinity, and the methods of 
Appendix~\ref{subsec:deter:eqs} and~\ref{subsec:final:step}, we obtain the statement of the lemma for 
$$c_{r}=c_{r}'+1,$$
 as required.
\end{proof}


\begin{lemma}\label{lemma:rholh}
 Consider any finite time-interval $[t_{1},t_{2}]$, where $0<t_{1}<t_{2}<\infty$. Then, for $z>0$,
 \begin{align*}
&\lim_{L\to \infty}   P\bigg[ \sup_{t\in [t_{1},t_{2}]} 
\nonumber \\ & \qquad \quad 
\max \Big(\theta_{L,h}(t), \theta^{(fg)}_{L}(t), \theta^{(f)}_{L}(t) ,\tilde{\theta}_{L}(t) \Big) <c_{\theta}z^{-1} \bigg]=1
 \end{align*}
 where $c_{\theta}>0$ is a constant independent of $z$ and $L$.
\end{lemma}
\begin{proof}
   We first provide the proof sketch for the statement of lemma only containing $\theta_{L,h}(t)$, i.e., we show that
   \begin{align}
    \lim_{L\to \infty}   P\bigg[ \sup_{t\in [t_{1},t_{2}]} \theta_{L,h}(t) <c_{\theta,h}z^{-1} \bigg]=1 
\label{eq:lim:prob:rho:h:bound}
   \end{align}
where $c_{\theta,h}>0$ is a constant independent of $z$ and $L$.
Taking similar steps, we obtain the same limit probabilities for $\theta^{(fg)}_{L}(t)$, $\theta^{(f)}_{L}(t)$, and
$\tilde{\theta}_{L}(t)$ with 
corresponding constants $c_{\theta}^{(fg)}>0$, $c_{\theta,h}^{(f)}>0$, and $\tilde{c}_{\theta}>0$. Since all of these limit probabilities state that 
events occur with probability one in the limit, the intersection of these events also occurs with probability one in the same limit 
as $L$ approaches infinity. Hence, using
these limit probabilities, we can
obtain the limit probability in the lemma with 
$$c_{\theta}=\max \big( c_{\theta,h}, c_{\theta}^{(fg)} ,c_{\theta,h}^{(f)} , \tilde{c}_{\theta} \big).$$

We now proceed to provide the proof-sketch for \eqref{eq:lim:prob:rho:h:bound}.
 By definition, since $\theta_{L,h}(t)$ is the fraction of links that are inactive and sense the channel as idle, and $\theta_{L}(t)$ 
is the fraction of links that are active, both at time $t$, we have
 \begin{align}
\theta_{L,h}(t)+\theta_{L}(t)\leq 1.   
\label{ineq:sum:rhos}
 \end{align}

Consider the change in $\theta_{L,h}(t)$ from time $t$ to time $t+\varepsilon$.
Since the attempt rate is 
$z$,  using the same 
arguments provided in Appendix~\ref{subsec:diff:evol}, we can show that the term 
$$z \theta_{L,h}(t)\varepsilon$$
 plus an error term gives 
the fraction of links that belong to $\theta_{L,h}(t)$ and that are active at time $t+\varep$. Hence, $z \theta_{L,h}(t)\varepsilon$ plus
the error term serves as a lowerbound\footnote{We have a lowerbound since an inactive link in $\theta_{L,h}(t)$ that is active 
at time $t+\varep$ may make other (interfering) inactive links in $\theta_{L,h}(t)$ sense the channel as busy at time $t+\varep$. In such a case,
by defintion,  $\theta_{L,h}(t+\varep)$ will not account for these inactive links.} for the rate by which $\theta_{L,h}(t)$ decreases from time $t$ to time $t+\varep$ due to inactive links in $\theta_{L,h}(t)$ that start transmitting.

Similarly, the term 
$$5\theta_{L}(t)\varepsilon$$
 plus an error term
 gives the maximum increase in $\theta_{L,h}(t)$ due to active links in $\theta_{L}(t)$ that stop transmitting. This follows since 1) each active link stops transmitting with unit rate, and 2) 
at most five links may sense the channel as idle
when an active link stops transmitting, i.e., the link itself and its four interfering links.
Finally, we must consider the events in which multiple transitions occur from time $t$ to time $t+\varep$ by one link or a link and the 
links within its $r_{n}$-neighbourhood. These events can contribute to $\theta_{L,h}(t+\varep)$.

 Considering the obtained rates and the events with multiple transitions that contribute to the change in $\theta_{L,h}(t)$ from time $t$ to time $t+\varep$, we can use a similar approach as taken in 
Appendix~\ref{subsec:diff:evol} to show that
\beqa
\Delta(\theta_{L,h}(t))   \leq -z \theta_{L,h}(t) \varepsilon +5 \theta_{L}(t) \varepsilon +e_{h}(t)+B_{h}(t)
\label{ineq:delta:rho:h}
\eeqa
where 
\begin{align}
|e_{h}(t)|=O(z^{2}\varepsilon),  \label{eq:rho:h:error:t}
\end{align}
and 
\begin{align}
  \ev[B_{h}(t)^{2}|\mathcal{H}_{h}(t)]=O(zL^{-\zeta})+O(z^{3}\varep) \label{eq:rho:h:B:t}
\end{align}
where $\setH_{h}(t)$ is the history of $\theta_{L,h}(t)$ from time zero up to and including time $t$.
Using \eqref{ineq:sum:rhos} and \eqref{ineq:delta:rho:h}, we have that 
\begin{align}
  \Delta(\theta_{L,h}(t)) \leq \big[5-(z+5) \theta_{L,h}(t)\big]\varep+e_{h}(t)+B_{h}(t)
\label{ineq:diff:ineq:rho:h}
\end{align}

For a fixed $z$, we can obtain the deterministic differential inequality counterpart of the above stochastic difference inequality by letting
 $$e_{h}(t)=B_{h}(t) \equiv 0,$$
 taking the limit of $\varepsilon$ approaching zero, and replacing $\theta_{L,h}(t)$
 with $x_{h}(t)$:
$$\frac{d}{dt} x_{h}\leq 5-(z+5) x_{h}.$$
Any solution to the above inequality is bounded by a function of the form
 $$\frac{5}{z+5}+C_{h} e^{-(z+5)t},$$
 for some constant $C_{h}\geq 0$
that depends only on the initial condition on $x_{h}(0)$. Since 
as the initial condition 
$$0\leq x_{h}(0)=\theta_{L,h}(0)\leq 1 ,$$
 we have 
$$0 \leq C_{h} < 1.$$
 Therefore, considering the assumption in the lemma that $0< t_{1}<t_{2}<\infty$, we can
 find a constant $c_{\theta,h}'>0$ independent of $L$ and $z$ such that
 \begin{align}
\sup_{t\in [t_{1}, t_{2}]}   x_{h}(t)\leq \sup_{t\in [t_{1}, t_{2}]} \Big[ \frac{5}{z+5}+e^{-(z+5)t}\Big] <c_{\theta,h}' z^{-1}.
 \end{align}

Therefore, the solution to the deterministic differential inequality counterpart of \eqref{ineq:diff:ineq:rho:h}
is less than $c_{\theta,h}' z^{-1}$ over the entire interval $[t_{1},t_{2}]$. 
Using this result, the bounds given in \eqref{eq:rho:h:error:t} and \eqref{eq:rho:h:B:t}, which both approach zero as 
$L$ approaches infinity, and the methods of 
Appendix~\ref{subsec:deter:eqs} and~\ref{subsec:final:step}, we then can show that for 
$$c_{\theta,h} =c_{\theta,h}'+1$$
\eqref{eq:lim:prob:rho:h:bound} holds,
 as required.
\end{proof}

\begin{lemma}\label{lemma:R}
  Consider the lattice interference graph $G_{L}$ and a time $t$, $t_{0} < t < \infty$. 
Suppose we have that 
\begin{align}
  0<\delta_{1}\leq \delta_{L}(t)\leq \delta_{2}<0.5,\label{ineq:delta:lemma:R}
\end{align}
for some constants $\delta_{1}$ and $\delta_{2}$. Moreover, suppose
\begin{align}
  \eta(L)+5\theta_{L,h}(t) < 0.5 \delta_{L}(t)
\end{align}
where $\eta(L)$ is defined in Lemma~\ref{lemma:area} with the property that 
$$\lim_{L\to \infty} \eta(L)=0,$$
and $\theta_{L,h}(t)$ is the density of links that are inactive and sense the channel as idle at time 
$t$, as defined in Appendix~\ref{sec:all:densities}.
Then, there exist constants $z_{R}'$ and $L_{R}'$ such that for $z>z_{R}'$ and $L>L_{R}'$, w.p.1, we have 
\begin{align}
R_{L}(t)\geq c_{R} \delta_{L}(t)^{3}, \label{eq:prob:Rl:del}  
\end{align}
for some constant $c_{R}>0$, independent of $t$, $L$, $z$, and $\tau$.
\end{lemma}

\begin{proof}

By the assumption in the lemma, we have that 
\begin{align}
  \eta(L)+5\theta_{L,h}(t) < 0.5 \delta_{L}(t).
\end{align}
Using this inequality, Lemma~\ref{lemma:B:eps}, and the definitions given in Section~\ref{sec:assumptionss}, we have w.p.1 
\beqa
 \frac{\sqrt{2}\sum_{\setC\in \setC_{L}^{(nd)}(t,z)} \ell(\setC)}{L\delta_{L}(t)}   >1 \label{eq:select:z:l}.
\eeqa
Let $\ell(t)$ be the average boundary-length of non-dominating clusters at time $t$, i.e., let
\begin{align}
  \ell(t)=\frac{\sum_{\setC\in\setC_{L}^{(nd)}(t,z)} \ell(\setC)}{\# \setC_{L}^{(nd)}(t,z) }.
\label{eq:def:ell:t}
\end{align}
Using this definition and \eqref{eq:select:z:l}, we have that
\begin{align}
 \# \setC_{L}^{(nd)}(t,z) \ell(t)> \frac{1}{\sqrt{2}}L \delta_{L}(t)  .
\label{ineq:setC:ell:low}
\end{align}

Since the total number of links in $G_L$ is $L$, and the coverage area $A_{l}$ of each link is at most two (see Appendix~\ref{sec:definitions}), 
we find $2L$ as an upperbound for the
total area covered by the links in the clusters. Hence, we have that
\begin{align}
  \sum_{\setC \in \setC_{L}^{(nd)}(t,z)} A(\setC) \leq 2L. \label{ineq:sum:A:up}
\end{align}
By Assumption~\ref{assum:1}, for any $z\geq 1$ and $t>t_0$, there exists $L_{A1}$ such that if $L>L_{A1}$, w.p.1, we have that 
\begin{align}
  \sum_{\setC \in \setC_{L}^{(nd)}(t,z)} A(\setC)  \geq 0.5 c_{a}   \sum_{\setC \in \setC_{L}^{(nd)}(t,z)} \ell(\setC)^{2}.
\label{ineq:assum:1:init}
\end{align}
In the rest, we assume that 
$$L>L_{A1},$$
and 
$$z>z_{R}'=1.$$
By \eqref{ineq:sum:A:up} and \eqref{ineq:assum:1:init}, for $L>L_{A1}$, we obtain
\begin{align}
   \sum_{\setC \in \setC_{L}^{(nd)}(t,z)} \ell(\setC)^{2} \leq \frac{4L}{c_{a}}.  \label{ineq:sum:ell2:up}
\end{align}
By Jenson's inequality and \eqref{eq:def:ell:t}, we also have that 
\begin{align}
\frac{\sum_{\setC \in \setC_{L}^{(nd)}(t,z)} \ell(\setC)^{2} }{\#\setC_{L}^{(nd)}(t,z)}\geq  
\bigg[ \frac{\sum_{\setC \in \setC_{L}^{(nd)}(t,z)} \ell(\setC) }{\#\setC_{L}^{(nd)}(t,z)}\bigg]^{2}=\ell(t)^{2},
\end{align}
which along with \eqref{ineq:sum:ell2:up} leads to 
\begin{align}
  \#\setC_{L}^{(nd)}(t,z)\ell(t)^{2}\leq \frac{4L}{c_{a}}. \label{ineq:setC:ell:up}
\end{align}
Using \eqref{ineq:setC:ell:low} and \eqref{ineq:setC:ell:up}, we have that 
\begin{align}
  \ell(t)\leq \frac{4\sqrt{2}}{c_r}  \frac{1}{ \delta_{L}(t)}. \label{ineq:ell:delta:b}
\end{align}
Using the above inequality and \eqref{ineq:setC:ell:low}, we have that
\begin{align}
\#\setC_{L}^{(nd)}(t,z)\geq \frac{Lc_{a}}{8} \delta_{L}(t)^{2}. \label{ineq:setC:up:fin}
\end{align}

We use the above bounds to find a lowerbound for $R_{L}(t)$. Recall that $R_{L}(t)$ is the density of critical
events (see Appendix~\ref{sec:chang:event}). Since each critical event maps to a bump of one step (see Appendix~\ref{sec:randomness}), $R_{L}(t)$ is also the density of bumps of one step. By Assumption~\ref{assum:2}, w.p.1, for a given $\ell$,
for non-dominating clusters of length $\ell$, we have that the $\liminf$ of average number of bumps of one step per non-dominating cluster is lower-bounded by $\frac{c_{1}}{\ell}$. Therefore, for $L>L_{\ell}$, where $L_{\ell}$ is a constant dependent on $\ell$, we have that
\begin{align}
  \frac{\sum_{\setC\in \setC_{L}^{(nd)}(t,z,l)} N_{\setC}^{(b)} (n=1) }{\#\setC_{L}^{(nd)}(t,z,\ell)} \geq \frac{0.5 c_{1}}{\ell} \label{ineq:assum:nbumps}
\end{align}
where $c_{1}$ is a positive constant independent of $z$, $t$, $\tau$, or $L$.

Consider all non-dominating clusters of boundary length $k\ell(t)$ or less where 
$$k=2.$$
 By \eqref{ineq:delta:lemma:R} and \eqref{ineq:ell:delta:b}, 
$\ell(t)$ is finite and 
bounded. Hence, there are only a finite number of values for the boundary length of the considered clusters. This implies that we can 
find a constant $L_{A2}$ such that if $L>L_{A2}$, then \eqref{ineq:assum:nbumps} holds for all $\ell\leq k \ell(t)$. In the rest, we assume 
that 
$$L>L_{R}'=\max(L_{A1},L_{A2}) $$
so that \eqref{ineq:ell:delta:b}, \eqref{ineq:setC:up:fin}, and \eqref{ineq:assum:nbumps} hold for $\ell\leq k \ell(t)$.

By \eqref{ineq:ell:delta:b} and Markov inequality,
for the total number of non-dominating clusters with boundary length $k\ell(t)$ or less, we have
\begin{align}
  \sum_{\ell\leq k \ell(t)} \# \setC_{L}^{(nd)}(t,z,\ell) \geq \Big(1-\frac{1}{k}\Big) \ \#\setC_{L}^{(nd)}(t,z). \label{ineq:no:c:ell:less}
\end{align}

Using \eqref{ineq:ell:delta:b}, \eqref{ineq:setC:up:fin}, \eqref{ineq:assum:nbumps}, and \eqref{ineq:no:c:ell:less}, for the total number of bumps of one step on the 
boundary of non-dominating clusters with boundary-length of $k\ell(t)$ or less, we have 
\begin{align}
  \sum_{\ell\leq k \ell(t)} \sum_{\setC\in \setC_{L}^{(nd)}(t,z,l)} N_{\setC}^{(b)} (1) & \geq   \sum_{\ell\leq k \ell(t)}  \# \setC_{L}^{(nd)}(t,z,\ell) 
 \frac{0.5c_{1}}{\ell}
\nonumber \\
&\geq  \frac{0.5c_{1}}{k\ell(t)} \sum_{\ell\leq k \ell(t)}  \# \setC_{L}^{(nd)}(t,z,\ell) 
\nonumber \\ &
\geq \frac{c_1 c_{a}^{2} }{k64\sqrt{2} } \Big(1-\frac{1}{k}\Big) L \delta_{L}(t)^{3}.
\end{align}
Thus, for $z>z_{R}'$ and $L>L_{R}'$, we have
\begin{align}
  R_{L}(t)\geq   \frac{1}{L}\sum_{\ell\leq k \ell(t)} \sum_{\setC\in \setC_{L}^{(nd)}(t,z,l)} N_{\setC}^{(b)} (1)\geq  
c_{R}\delta_{L}(t)^{3}
\end{align}
where 
$$c_{R}=\frac{c_1 c_{a}^{2} }{2^{8} \sqrt{2} }$$
is a constant independent of $L$, $z$, $\tau$, and $t$. This completes the proof of the lemma.
\end{proof}


\begin{lemma}\label{lemma:conditional}
  Consider a non-negative r.v. $x$, and an event $A$ such that 
$P(A)\geq 1-\eps_{A}$, where $0\leq \eps_{A} <1$. We have 
\begin{align}
  \ev[x|A] \leq \frac{\ev[x]}{1-\eps_{A}} . \nonumber
\end{align}
\end{lemma}
\begin{proof}
  Define $A^c$ to be the complement of the event $A$. For the expected value of r.v. $x$, we can write
  \begin{align}
    \ev[x]=\ev[x| A]P(A)+\ev[x|A^c][1-P(A)].\nonumber
  \end{align}
Using the non-negativity of r.v. $x$ and the assumption for the event $A$, we have 
  \begin{align}
    \ev[x|A]&=\frac{ \ev[x]-\ev[x| A^c] [1-P(A)] }{P(A)}\leq \frac{ \ev[x]}{P(A)} 
   \nonumber \\ &\leq \frac{\ev[x]}{1-\eps_{A}},\nonumber
  \end{align}
completing the proof.

\end{proof}

\begin{lemma}\label{lemma:t1}
  Suppose the following holds for the initial condition $\bfx_{t_{0}}=\bfx(t_{0})$ of the ODE of \eqref{eq:diff:cont}. First, suppose
  \begin{align}
    0<c_{x_{1},1}<x_{1}(t_{0})<c_{x_{1},2}<0.5 , \label{ineq:assum:x1:lemma:m}
  \end{align}
where $c_{x_{1},1}$ and $c_{x_{1},2}$ are constants independent of $z$. Second, suppose for $2\leq i \leq 4$,
\begin{align}
  x_{i}(t_{0})=O(z^{-1}).\label{ineq:assum:xs:lemma:m}
\end{align}
Then, there exists a constant $z_{0}'$ such that the following holds for $z>z_{0}'$. First, we have that the ODE of \eqref{eq:diff:cont} has a well-defined unique solution $\mathbf{x}(t)$, i.e., we have that 
$\mathbf{x}(t)$ exists, is unique, and $\mathbf{x}(t) \in \setD_{w}$ for 
all $t\in [t_{0},\tau]$. Second, there exists a $t_{1}$, $t_{0}\leq t_{1}<2$, independent of $z$, such that
\begin{align}
 \sup_{t\in[t_{1},\tau]} \| \mathbf{x}(t)-\mathbf{y}(t)\|=O(z^{-1}), \ \text{as $z\to \infty$}, \label{eq:B:x:y:z-1}
\end{align}
where $\mathbf{y}(t)$ is the solution to \eqref{eq:diff:y} with the initial condition
 $$\mathbf{y}(t_{1})=\mathbf{x}(t_{1})$$ with the property that
\begin{align}
 \sup_{t\in[t_{1},\tau]} \left|\frac{d}{dt}\bfy(t)\right|=O(1)\mathbf{1}_{4\times 1},\ \text{as $z \to \infty$},  \label{eq:b:der:yt} 
\end{align}
and for $t\geq t_{1}$
\begin{align}
  y_{1}(t)> 0.5 - \frac{C_{y}}{\sqrt{t}} \label{ineq:y1:bound:m:s}
\end{align}
where $C_{y}>0$ is a constant independent of $t_{1}$ and $z$.

\end{lemma}


\begin{proof}

By Lemma~\ref{lemma:x:cond}, there exists a constant $z_{0}$ such that for $z>z_{0}$, there exists a $t_{1}$ independent of $z$ with 
$$t_{0} \leq  t_{1} <2$$
 such that under the assumptions in the lemma,
 $\mathbf{x}(t)$ is unique and well-defined over the interval $[t_{0},t_{1}]$, i.e.,
 \begin{align}
   \bfx(t)\in\setD_{w}, \ t\in[t_{0},t_{1}].    
 \end{align}
 Lemma~\ref{lemma:x:cond} also states that for $z>z_{0}$, we have
  \begin{align}
    0<c_{x_{1},1}'<x_{1}(t_{1})<c_{x_{1},2}'<0.5 ,\label{ineq:init:x1:t1:an}
  \end{align}
where $c_{x_{1},1}'$ and $c_{x_{1},2}'$ are constants independent of $z$. By Lemma~\ref{lemma:x:cond}, we also have that
\begin{align}
  x_{i}(t_{1})=O(z^{-1}),\qquad   2\leq i \leq 4,
\label{ineq:init:x2:4:t1}
\end{align}
and 
\begin{align}
&  \Big|x_{2}(t_{1})-\frac{1}{3z} c_{R}\delta_{\mathbf{x}}^{3}(t_{1}) \Big|=O(z^{-2}),
\label{eq:x2:t1}
\\ &
  \Big|x_{3}(t_{1})-\frac{2}{3z} c_{R}\delta_{\mathbf{x}}^{3}(t_{1}) \Big|=O(z^{-2}).
\label{eq:x3:t1}
\\ &
  \Big|x_{4}(t_{1})-\frac{1}{z}\big[x_{1}(t_{1})-c_{R}\delta_{\mathbf{x}}(t_{1})^{3}\big] \Big|=O(z^{-2}). \label{eq:x4:t1}
\end{align}

Having that for $z>z_{0}$, $\bfx(t)$ uniquely exists over the interval $[t_{0},t_{1}]$, it remains to prove that
 $\bfx(t)$ as the solution to 
\begin{align}
  \frac{d}{dt}\bfx=\bfA\bfx+f(\bfx), \ \bfx(t_{1})=\bfx_{t_{1}} \label{eq:diff:cont:l}
\end{align}
is also unique and well-defined over the interval $[t_{1},\tau]$ given an initial value at time $t_{1}$, and that 
\eqref{eq:B:x:y:z-1}-\eqref{ineq:y1:bound:m:s} hold. We note that 
general theorems on existence and uniqueness of solutions for ODEs require Lipschitz continuous functions
on the RHS of ODEs.
However, the RHS of \eqref{eq:diff:cont:l} is only locally Lipschtiz continuous. One approach to address the
 lack of Lipschitz continuity is through adding a 
term to the RHS of \eqref{eq:diff:cont:l} to make sure that the RHS is Lipschtiz continuous. To do so, we define the
following ODE with an initial condition at time $t_{1}$, where $t_{1}$ is specified in Lemma~\ref{lemma:x:cond}:
\begin{align}
&  \frac{d}{dt}\mathbf{x}_{g}=\bfA\mathbf{x}_{g}+f(\mathbf{x}_{g})+g(\mathbf{x}_{g}),
\nonumber 
\\
 &  \mathbf{x}_{g}(t_{1})=\mathbf{x}(t_{1})\label{eq:cont:diff:g}
\end{align}
where $\bfx_{g}(t)$ is given by
$$ \bfx_{g}(t)= 
\begin{bmatrix}
x_{g,1}(t) \\ x_{g,2}(t) \\ x_{g,3}(t) \\ x_{g,4}(t)  
\end{bmatrix}
, $$
and we choose $g(\mathbf{x})$ to be a function that satisfies the following:
\begin{itemize}
\item For  $\|\mathbf{x}\|\geq 2 $,
\begin{align}
  \|f(\mathbf{x})+g(\mathbf{x})\|=0. \label{eq:def:g}
\end{align}
\item Defining 
\begin{eqnarray}
  \mathcal{D}=\{\mathbf{x}: x_{i} \in [-0.6,0.6], 1\leq i \leq 4\}, \label{eq:def:setD}
\end{eqnarray}
we have that
\begin{align}
 \|g(\mathbf{x})\|=0 , \  \mathbf{x}\in \mathcal{D} \label{eq:def:g:D}
\end{align}
\item The function $(f+g)(\mathbf{x})$ is continuously differentiable (on $\mathbb{R}^{4}$) and Lipschitz continuous with Lipschitz constant 
$L_{f+g}$ that is independent of $z$.
\item The Lipschitz continuous function $(f+g)(\bfx)$ is such that 
for any $\bfx_{1}$ and $\bfx_{2}$ 
\begin{align}
  (f+g)(\bfx_{1})-(f+g)(\bfx_{2}) = \mathbf{L}_{f+g}(\bfx_{1},\bfx_{2}) (\bfx_{1}-\bfx_{2})
\nonumber 
\end{align}
where each element $ [\mathbf{L}_{f+g}(\bfx_{1},\bfx_{2})]_{ij}$ of matrix $ \mathbf{L}_{f+g}(\bfx_{1},\bfx_{2})$ is such that
\begin{align}
  \left|  [\mathbf{L}_{f+g}(\bfx_{1},\bfx_{2})]_{ij} \right|<L_{f+g}. \label{ineq:Lfgij:B}
\end{align}
\end{itemize}

The third condition in the above ensures that for a given initial condition at time $t_{1}$, the solution $\bfx_{g}(t)$ to
 the ODE in \eqref{eq:cont:diff:g} exists for all $t\geq t_{0}$ and that
it is unique (e.g. see \cite{perko:book}, Chapter~3, Theorem~3.). We also note that 
a function $g(\bfx)$ satisfying the fourth condition exists since
$f(\mathbf{x})$ is a (polynomial) function of only $x_{1}$.


Consider the ODE for $\bfy(t)$ as given in \eqref{eq:diff:y} with the initial condition at time $t_{1}$ such that
\begin{align}
  \bfy(t_{1})=\bfx_{g}(t_{1})=\bfx(t_{1}). \label{def:init:con:yt1:xgt1}
\end{align}
Considering the constraints on the initial condition $\bfy(t_{1})$ as imposed by the above
 equality and \eqref{ineq:init:x1:t1:an}-\eqref{eq:x4:t1}, it follows from Lemma~\ref{lemma:y:diff:B} that
\begin{align}
\sup_{t\in [t_{1},\tau]} \left|  \frac{d}{dt}\bfy(t)\right| =O(1)\mathbf{1}_{4\times1}. \label{eq:b:der:yt:1}
\end{align}

By the ODE for $\bfy(t)$, we also have that $y_{1}(t)$ is given by the following ODE
\begin{align}
  \frac{d}{dt}y_{1}=\frac{2}{3}c_{R}\delta_{\mathbf{y}}^{3}(t)=\frac{2}{3} c_{R} (0.5-y_{1})^{3} \label{eq:ode:y1}
\end{align}
 with the initial condition that
\begin{align}
  y_{1}(t_{1})=x_{1}(t_{1}). \label{eq:init:y:x}
\end{align}
Using this ODE for $y_{1}(t)$, we obtain for $t\geq t_{1}$
 \begin{eqnarray}
   y_{1}(t)=0.5-\al (1+\be (t-t_{1}))^{-\frac{1}{2}},
   \label{eq:y1:exact:sol}
 \end{eqnarray}
where
\begin{align}
  \al=0.5-y_{1}(t_{1}), \label{eq:al}
\end{align}
and 
\begin{align}
  \be=\frac{4c_{R}}{3}\al^{2}.\label{eq:be}
\end{align}
By \eqref{ineq:init:x1:t1:an} and \eqref{eq:init:y:x}, we have that 
\begin{align}
 & 0<\al_1<\al<\al_2 <0.5, 
\nonumber \\
&  0<\be_1<\be<\be_2<C_R, \label{ineq:al:be:cons}
\end{align}
where $\al_1$, $\al_2$, $\be_1$, and $\be_2$ are constants independent of $z$. Using these constants, 
we can find a constant $C_{y}>0$ independent of $t_{1}$ and $z$ such that
\begin{align}
  y_{1}(t)> 0.5 - \frac{C_{y}}{\sqrt{t}}, \ t\geq t_{1}. \label{ineq:B:t:y1:f}
\end{align}

We can also rewrite the ODE for $y_{1}(t)$ given in \eqref{eq:ode:y1} as
\begin{align}
  \frac{d}{dt}y_{1}& =-y_{1}+ 3zy_{2}+zy_{3}+zy_{4} \nonumber \\
&                     \quad -3zy_{2}+c_{R}\delta_{\mathbf{y}}^{3}(t) \label{y1:term2} \\
 &                     \quad -z y_{3}+ \frac{2}{3}c_{R}\delta_{\mathbf{x}}^{3}(t) \label{y1:term3} \\
&                     \quad -zy_{4}+y_{1}- c_{R}\delta_{\mathbf{y}}^{3}(t) . \label{y1:term4} 
\end{align}
By Lemma~\ref{lemma:y:diff:B}, we have that the $\sup$ of expressions in \eqref{y1:term2}-\eqref{y1:term4} over $[t_{1},\tau]$
 all are $O(z^{-1})$. Hence, for $t\in[t_{1},\tau]$, we have
\begin{eqnarray}
  \frac{d}{dt}y_{1}=-y_{1}+zy_{3}+zy_{4}+e_{y_{1}} \label{eq:dy1}
\end{eqnarray}
where
\begin{align}
\sup_{t\in [t_{1},\tau]} |e_{y_{1}}(t)|=O(z^{-1}).\label{eq:ey1:B}  
\end{align}

Since by Lemma~\ref{lemma:y:diff:B}, for $t \in[t_{1},\tau]$ we have $\mathbf{y}(t)\in \mathcal{D}_{w}$, and $\setD_{w}\subset \mathcal{D}$, 
by the property of $g(\cdot)$ given in \eqref{eq:def:g:D}, we have that 
$$\|g(\mathbf{y}(t))\|=0 , \ t \in[t_{1},\tau].$$
 Using this equality, the ODE in \eqref{eq:dy1}, and 
the differential equations for $y_{2}(t)$, $y_{3}(t)$, and $y_{4}(t) $, 
as given in \eqref{eq:diff:y}-\eqref{eq:def:tildef}, we obtain that for $t \in[t_{1},\tau]$
\begin{eqnarray}
  \frac{d}{dt}\mathbf{y}=\bfA \mathbf{y}+f(\mathbf{y})+g(\mathbf{y})+\mathbf{e}_{\mathbf{y}},
\label{eq:diff:y:g:e}
\end{eqnarray}
where $\mathbf{e}_{\mathbf{y}}(t)$ accounts for the error term $e_{y_{1}}(t)$, and by \eqref{eq:ey1:B}, we have
\begin{eqnarray}
\sup_{t\in[t_{1},\tau]}  \|\mathbf{e}_{\mathbf{y}}(t)\|=O(z^{-1}). \label{eq:error:bound:ey}
\end{eqnarray}
Therefore, $\mathbf{y}(t)$ is the solution to a perturbed version of the ODE in \eqref{eq:cont:diff:g}. 
We next use this result to complete the proof of the lemma.

Define $\bfe_{g}(t)$ as
\begin{align}
  \bfe_{g}(t)=\bfx_{g}(t)-\bfy(t) .
\end{align}
Considering the properties given for the chosen function $g(\cdot)$ in the beginning of the proof, by Lemma~\ref{lemma:yh:xh}, 
we have that
\begin{align}
\sup_{t\in[t_{1},\tau] } \|\bfe_{g}(t) \| =O(z^{-1}) .\label{ineq:eg:lemm}
\end{align}
Since by Lemma~\ref{lemma:y:diff:B}, $\mathbf{y}(t)\in\mathcal{D}_{w}$ where $\mathcal{D}_{w}$ is 
strictly inside $\mathcal{D}$, by \eqref{ineq:eg:lemm}, we have that for $z>z_{1}$, where 
$z_{1}>z_{0}$ and is a sufficiently large constant, 
\begin{align}
  \mathbf{x}_{g}(t)\in \mathcal{D}, \ t\in[t_{1},\tau] . \label{ineq:xg:in:D}
\end{align}
 Hence, for $z>z_{1}$, by \eqref{eq:def:g:D} and \eqref{ineq:xg:in:D}, 
 \begin{align}
g(\mathbf{x}_{g}(t))=0   , \ t\in [t_{1},\tau],
 \end{align}
which along with \eqref{eq:cont:diff:g} for $z>z_{1}$ leads to
\begin{eqnarray}
\frac{d}{dt}\mathbf{x}_{g}= \bfA \mathbf{x}_{g}+f(\mathbf{x}_{g}) , \ t\in [t_{1},\tau]  \label{eq:xg:x:equ}.
\end{eqnarray}

In addition, since $y_{1}(t)$ is an increasing function of $t$ and less than $0.5$ as given by \eqref{eq:y1:exact:sol}, 
by \eqref{ineq:eg:lemm}, we can choose a sufficiently large constant $z_{2}$, $z_{2}>z_{1}$, such that for $z>z_{2}$
and any $t\in[t_{1},\tau]$, 
\beqa
\frac{1}{2}(0.5-\al_{2}) < x_{g,1}(t) < 0.5  \label{eq:05x1:bound}
\eeqa
where $\al_{2}$ is a constant defined in \eqref{ineq:al:be:cons}. Using the above bounds for 
$x_{g,1}(t)$ and that $\bfx_{g}(t)$ satisfies the ODE in \eqref{eq:xg:x:equ} with the constraints at time $t_{1}$
given in \eqref{ineq:init:x1:t1:an}-\eqref{ineq:init:x2:4:t1}, we have that 
\begin{align}
0<x_{g,i}(t), \ 2\leq i \leq 4, \ t\in[t_{1},\tau],
\end{align}
and
\begin{align}
  \sup_{t\in[t_{1},\tau]} x_{g,i}(t)=O(z^{-1}),  \ 2\leq i \leq 4.
\end{align}
 Therefore, by the above bound and \eqref{eq:05x1:bound}, for $z>z_{3}$ where $z_{3}>z_{2}$ is a sufficiently large constant, we have
\begin{align}
  \bfx_{g}(t)\in \setD_{w}, \ t\in[t_{1},\tau]. \label{ineq:xg:in:Dw}
\end{align}

Since $\bfx_{g}(t)$ is uniquely defined over the interval $[t_{1},\tau]$, by \eqref{eq:cont:diff:g}, \eqref{eq:xg:x:equ}, and \eqref{ineq:xg:in:Dw}, we have that for $z>z_{3}$, $\bfx_{g}(t)$
serves as a well-defined solution for the original ODE in \eqref{eq:diff:cont:l}.
Moreover, since the
RHS of the ODE in \eqref{eq:diff:cont:l} is continuously differentiable with respect to $x_{i}$'s, $ 1\leq i \leq 4$, we are ensured
 if a solution to \eqref{eq:diff:cont:l} exists, it has to be unique (e.g., see Theorem~1.13 in \cite{markley:book}). Therefore, 
$\mathbf{x}_{g}(t)$ serves as a unique solution to \eqref{eq:diff:cont:l} for the interval $[t_{1},\tau]$. By \eqref{ineq:eg:lemm} and 
\eqref{ineq:xg:in:Dw}, 
it then follows that for 
$$z>z_{0}'=z_{3},$$ 
the ODE in \eqref{eq:diff:cont:l} has a well-defined unique solution over the interval $[t_{1},\tau]$, and that
  \begin{align}
\sup_{t\in[t_{1},\tau] } \| \bfx(t)-\bfy(t) \|=O(z^{-1}) ,
\end{align}
which along with the results in \eqref{eq:b:der:yt:1} and \eqref{ineq:B:t:y1:f} completes the proof of the lemma.
  
\end{proof}


\begin{lemma}\label{lemma:rho:stoch}
Given that event $\setE(L,z)$ as defined in \eqref{eq:def:E:Ltau} occurs, for any $\eps_{\theta}>0$, we have that 
\begin{align*}
&\liminf_{z\to \infty} \ \liminf_{L\to \infty}  
 \nonumber \\ 
& \qquad \qquad P\bigg[ \sup_{\frac{t_{0}}{\varepsilon}\leq n \leq \frac{\tau}{\varepsilon}} \| 
\bftheta(n\varep) -\retheta(n\varep)\| \leq \eps_{\theta} \bigg]=1.
\end{align*}
\end{lemma}

\begin{proof}
Consider the system at discrete times in the set 
$$\{n \varepsilon, n\geq t_{0}\varepsilon^{-1}\}.$$
 For these times, by definition, 
$\boldsymbol{\theta}_{L}(n\varepsilon)$
provides a solution to the difference equation in \eqref{eq:difference:eqs}, and
$\boldsymbol{\theta}_{\mathbf{e}=\mathbf{0},\mathbf{B}=\mathbf{0}}(n\varepsilon)$ is the solution to the same difference equation
 with
 $$\|\mathbf{e}(t=n\varepsilon)\|\equiv 0, \ \|\mathbf{B}(t=n\varepsilon)\|\equiv 0,$$
and the initial condition that 
\begin{align}
\retheta(t_{0}) =  \bftheta(t_{0}) .\label{eq:init:reB:rho}
\end{align}

Throughout this proof, by the assumption in the lemma, we assume that event $\setE(L,z)$ occurs. As a result,
all expectations and probabilities are stated conditioned on $\setE(L,z)$.
In the rest, to simplify the notation, where appropriate we write $n$ to indicate the time instant $t=n\varepsilon$; hence,
we write $\boldsymbol{\theta}_{L}(n)$ to mean $\boldsymbol{\theta}_{L}(t=n\varepsilon)$.

We first provide the proof-sketch. Define the sequence
$\{\Upsilon_{n}, n\geq t_{0}\varepsilon^{-1} \}$ as
$$\Upsilon_{n}=\varepsilon\sum_{m=0}^{n-t_{0}\varepsilon^{-1}} \Big\| |\mathbf{e}(n-m)|+|\mathbf{B}(n-m)| \Big\|$$
where $\bfe(n)=\mathbf{e}(t=n\varep)$ and $\bfB(n)=\mathbf{B}(t=n\varep)$ are defined in \eqref{eq:difference:vecs}.
Define
\begin{eqnarray}
  \mathbf{e}_{\boldsymbol{\theta}}(n)= \boldsymbol{\theta}_{L}(n)-\boldsymbol{\theta}_{\mathbf{e}=\mathbf{0},\mathbf{B}=\mathbf{0}}(n).
\label{eq:def:e:rho}
\end{eqnarray}
By \eqref{eq:init:reB:rho}, we have
\begin{align}
  \mathbf{e}_{\boldsymbol{\theta}}(n=t_{0}\varep^{-1})=\mathbf{0}.\label{eq:init:erho:n:l}
\end{align}

In the first part of the proof, we show that if 
\beqa
 \qquad \sup_{ \frac{t_{0}}{\varepsilon} \leq n \leq \frac{\tau}{\varepsilon}} \Upsilon_{n} \leq  \eps_{1}
\label{eq:upsilon:assum}
\eeqa
where $\eps_{1}>0$ can be chosen arbitrarily small, and if
\beqa 
\|\boldsymbol{\theta}_{\mathbf{e}=\mathbf{0},\mathbf{B}=\mathbf{0}}(t=n\varep)
-\mathbf{x}(t=n\varepsilon)\|<\eps_{\theta}, \ \frac{t_{0}}{\varepsilon}\leq n \leq \frac{\tau}{\varepsilon}
\label{ineq:rho:e:B:0:y:eps2}
\eeqa
where $\bfx(t)$ is the solution to \eqref{eq:diff:cont}-\eqref{eq:init:xt0:rhot0} and $\eps_{\theta}>0$ can be chosen arbitrarily small, then
for $z>z_0'$, where $z_0'$ is defined in Lemma~\ref{lemma:t1}, we have
\begin{eqnarray}
\sup_{\frac{t_{0}}{\varepsilon} \leq n \leq \frac{\tau}{\varepsilon}}  \|\mathbf{e}_{\boldsymbol{\theta}}(n)\|\leq \eps_{\theta}.
\end{eqnarray}

In the second part of the proof, to complete the proof, we show that \eqref{eq:upsilon:assum} holds with probability 
approaching one as we first let $L$ approach infinity and then let $z$ appraoch infinity. We also show that 
\eqref{ineq:rho:e:B:0:y:eps2} holds for sufficiently large $L$.

\underline{\emph{First part:}} Suppose  \eqref{eq:upsilon:assum} and \eqref{ineq:rho:e:B:0:y:eps2} hold for arbitrarily small $\eps_{1}>0$ and $\eps_{\theta}>0$.
We choose $\eps_{\theta}$ arbitrarily small such that $0<\eps_{\theta}<1$, and let
\beqa
\eps_{1}=\eps_{\theta}\frac{1}{4\hat{c} e^{L_{c}\tau}}, \label{eq:eps1:def}
\eeqa
where $\hat{c}$ is a constant
independent of $z$ that will be determined later (in the following inductive proof).
 We next show that
\begin{eqnarray}
  \|\mathbf{e}_{\boldsymbol{\theta}}(n)\|\leq  4 \hat{c} e^{L_{c} n\varepsilon  }
\Upsilon_{n},\qquad \frac{t_{0}}{\varepsilon}\leq n \leq \frac{\tau}{\varepsilon}, \label{eq:e:induct}
\end{eqnarray}
which along with \eqref{eq:upsilon:assum} and \eqref{eq:eps1:def} leads to 
\begin{eqnarray}
\sup_{\frac{t_{0}}{\varepsilon} \leq n \leq \frac{\tau}{\varepsilon}}  \|\mathbf{e}_{\boldsymbol{\theta}}(n)\|\leq \eps_{\theta}.
\end{eqnarray}

What remains to complete the first part is to show that
\eqref{eq:e:induct} indeed holds. We first note the following.
 Subtracting the difference equation associated with 
$\boldsymbol{\theta}(n)_{\mathbf{e}=\mathbf{0},\mathbf{B}=\mathbf{0}}$ from that of $\bftheta(n)$ given in \eqref{eq:difference:eqs}
and using \eqref{eq:def:e:rho}, we obtain
\begin{align}
  \Delta [\mathbf{e}_{\mathbf{\theta}}(n)]&=\bfA \mathbf{e}_{\mathbf{\theta}}(n)+f(\boldsymbol{\theta}(n))-
f(\boldsymbol{\theta}(n)_{\mathbf{e}=\mathbf{0},\mathbf{B}=\mathbf{0}}) 
\nonumber \\ & +\mathbf{e}(n)+\mathbf{B}(n),
\nonumber
\end{align}
If $\|\mathbf{e}_{\mathbf{\theta}}(n)\|\leq 1$ and $\|\boldsymbol{\theta}_{L}(n)\|\leq 2$, we can write
\begin{eqnarray}
  \Delta [\mathbf{e}_{\mathbf{\theta}}(n)]=\bfA \mathbf{e}_{\mathbf{\theta}}(n)+\mathbf{L}(n) \mathbf{e}_{\mathbf{\theta}}(n)+\mathbf{e}(n)+\mathbf{B}(n),
\label{eq:e:rho:ode}
\end{eqnarray}
where $\mathbf{L}(n) $ is a $4\times 4$ matrix such that
\begin{align}
|\mathbf{L}_{ij}(n)|\leq L_{c}, \ 1\leq i,j \leq 4 , \label{ineq:Lij:B:n}
\end{align}
for some constant $L_{c}>0$. This follows from the definition of $f(\cdot)$ as a polynomial function, as defined in \eqref{eq:difference:vecs}.
Therefore, as long as $\boldsymbol{\theta}(n)$ and $ \mathbf{e}_{\mathbf{\theta}}(n)$ 
are properly bounded, the evolution of $ \mathbf{e}_{\mathbf{\theta}}(n)$ can be locally linearized using the matrix
$\mathbf{L}(n)$ as in \eqref{eq:e:rho:ode}.

 We next prove inductively that for all 
$$n \in \{t_{0}\varepsilon^{-1} \leq n \leq \tau\varep^{-1}\},$$
 1) inequality \eqref{eq:e:induct} holds, 2)
equality \eqref{eq:e:rho:ode} holds, and 3) $\boldsymbol{\theta}(n) \in \mathcal{D}$ where $\mathcal{D}$
is defined in \eqref{eq:def:setD}.

For $k=t_{0}\varepsilon^{-1}$, by \eqref{eq:init:erho:n:l}, we have the
 degenerate case of 
$$\mathbf{e}_{\boldsymbol{\theta}}(k)=\mathbf{0}$$
 for which
\eqref{eq:e:induct} holds. Moreover, since the event $\setE(L,z)$ occurs, we have that 
$\bftheta(t=t_{0})\in \mathcal{D}$, or in other words,
$\bftheta(k=t_{0}\varepsilon^{-1})\in \mathcal{D}$. This implies that 
$$\|\bftheta(k=t_{0}\varepsilon^{-1}) \|\leq 2.$$ 
Since $\mathbf{e}_{\boldsymbol{\theta}}(k=t_{0}\varepsilon^{-1})=\mathbf{0}$ 
and $\|\bftheta(k=t_{0}\varepsilon^{-1}) \|\leq 2$, we also have that
equality \eqref{eq:e:rho:ode} holds. Hence, for $k=t_{0}\varepsilon^{-1}$, 
the claims of the induction hold.

Next, suppose for some $k$ where
$$ t_{0}\varepsilon^{-1} \leq k\leq \tau\varepsilon^{-1}-1, $$
 the claims of the induction 
hold for all 
$$n \in \{t_{0}\varepsilon^{-1},t_{0}\varepsilon^{-1}+1, \cdots, k\}.$$
 We show that the claims hold for
 time $n=k+1$. By the induction assumption, \eqref{eq:e:rho:ode} holds for all 
$$n \in \{ t_{0}\varepsilon^{-1},t_{0}\varepsilon^{-1}+1 \cdots, k\}.$$
For such an $n$, using the definition of $\Delta(\cdot)$ given in \eqref{eq:def:Delta:h}, we can rewrite \eqref{eq:e:rho:ode} as
\begin{eqnarray}
  \mathbf{e}_{\boldsymbol{\theta}}(n+1)=\hat{\bfA}(n)\mathbf{e}_{\boldsymbol{\theta}}(n)
+\varepsilon \big[\mathbf{e}(n)+\mathbf{B}(n)\big]
\end{eqnarray}
where 
\begin{align}
  \hat{\bfA}(n)=\mathbf{I}+\varepsilon [\bfA+\mathbf{L}(n)].
\end{align}
Using the above and \eqref{eq:init:erho:n:l}, we have that\footnote{For $m=0$, we define the product term in \eqref{eq:e:expr} to be one.} 
\begin{align}
&  \mathbf{e}_{\boldsymbol{\theta}}(n+1)
\nonumber \\ & 
\ =\varepsilon \sum_{m=0}^{n-t_{0}\varepsilon^{-1}} \Big[ \prod_{j=0}^{m-1} \hat{\bfA}(n-j)\Big]  \Big[\mathbf{e}(n-m)+\mathbf{B}(n-m)\Big].
\label{eq:e:expr}
\end{align}
Using \eqref{ineq:Lij:B:n}, we then have that 
\begin{align}
|  \mathbf{e}_{\boldsymbol{\theta}}(n+1)| \leq \varepsilon \sum_{m=0}^{n-t_{0}\varepsilon^{-1}}  \hat{\bfA}^{m} 
 \Big[|\mathbf{e}(n-m)|+|\mathbf{B}(n-m)|   \Big] \label{eq:e:expr:B}
\end{align}
where 
\begin{align}
  \hat{\bfA}=\mathbf{I}+\varepsilon [\bfA + L_{c} \mathbf{1}_{4\times 4}].
\end{align}

All eigenvalues of $\hat{\bfA}$ are distinct and real, and for small $\varepsilon$, they are
 close to one. Moreover, we have that the largest 
eigenvalue of $\hat{\bfA}$ is less than $1+4\varepsilon L_{c}$. Using this and the eigenvalue
decomposition for $\hat{\bfA}$ to write it as $\hat{\bfA}=\mathbf{U} \mathbf{D} \mathbf{U}^{-1}$, where 
$\mathbf{D}$ is a diagonal matrix containing the eigenvalues of $\hat{\bfA}$ and $\mathbf{U}$
is the matrix with its columns as the eigenvectors of $\hat{\bfA}$, we can show that each element $\hat{\bfA}_{ij}^{m}$ of the matrix
$\hat{\bfA}^{m}$ satisfies the following
\begin{eqnarray}
|  \hat{\bfA}_{ij}^{m}|\leq \hat{c} (1+4\varepsilon L_{c})^{m}\leq \hat{c} e^{4L_{c}\varepsilon m },
\label{ineq:Aijm}
\end{eqnarray}
for some constant $\hat{c}>0$. This constant was used earlier 
to determine $\eps_{1}$ in \eqref{eq:eps1:def}.

Using the bound in \eqref{ineq:Aijm}, we have that
\begin{eqnarray}
&&  \left\| \hat{\bfA}^{m}\Big[|\mathbf{e}(n-m)|+|\mathbf{B}(n-m)|\Big] \right\|
\nonumber \\ && \qquad \leq 4 \hat{c} e^{L_{c}\varepsilon m} \Big\| |\mathbf{e}(n-m)|+|\mathbf{B}(n-m)| \Big\|.
\label{eq:AB:product}
\end{eqnarray}
Hence, by \eqref{eq:e:expr:B} and \eqref{eq:AB:product}, we have that
\begin{eqnarray}
 \| \mathbf{e}_{\mathbf{\theta}}(n+1) \| \leq 
4\hat{c} e^{L_{c}\varepsilon (n+1)} \Upsilon_{n}.
\label{ineq:e:eps2}
\end{eqnarray}
Since, by the induction assumption,
$$n \in \{t_{0}\varepsilon^{-1},t_{0}\varepsilon^{-1}+1, \cdots, k\},$$
we can use choose $n=k$ in \eqref{ineq:e:eps2}. Therefore, we obtain that \eqref{eq:e:induct} holds for $n=k+1$.

Moreover, by \eqref{eq:upsilon:assum}, 
\eqref{eq:eps1:def}, and the inequality \eqref{ineq:e:eps2} for $n=k$, we have that
\begin{align}
 \big\| \mathbf{e}_{\mathbf{\theta}}(k+1) \big\| \leq \eps_{\theta}<1.
\label{ineq:e:eps2:2}
\end{align}
In addition, by \eqref{ineq:rho:e:B:0:y:eps2} and
 the error bound in \eqref{ineq:e:eps2:2} along with definition of $\mathbf{e}_{\mathbf{\theta}}(n)$ in \eqref{eq:def:e:rho}, we have that
 $$ \big\|\boldsymbol{\theta}_{L}(t=(k+1)\varep)- \mathbf{x}(t=(k+1)\varepsilon) \big\|\leq 2\eps_{\theta}.$$
 Under the assumption in the lemma that $\setE(L,z)$ occurs, by Lemma~\ref{lemma:t1}, for $z>z_0'$, we have
$$\mathbf{x}(t=(n+1)\varepsilon)\in \mathcal{D}_{w}$$
where $\setD_{w}$ is strictly inside $\setD$. 
This and the above inequality imply that for small $\eps_{\theta}$, we have
 $$\boldsymbol{\theta}_{L}(k+1) \in \mathcal{D}.$$
This inequality implies
that $\|\boldsymbol{\theta}(k+1)\|\leq 2$, which along with
\eqref{ineq:e:eps2:2} confirms that \eqref{eq:e:rho:ode} holds for $n=k+1$. Therefore, we have that for 
$n=k+1$ all claims of the induction hold, completing the inductive proof.


\underline{\emph{Second part:}} By the non-negativity of the norm function, the sequence $\{\Upsilon_{n}\}$ is a submartingle. By Doob's inequality, for any $\eps_{1}>0$, we then have that 
\begin{eqnarray}
P\Big[\sup_{\frac{t_{0}}{\varepsilon} \leq n \leq \frac{\tau}{\varepsilon}} \Upsilon_{n} > \eps_{1}\ \big| \setE(L,z) \Big]
\leq\eps^{-1}_{1} \ev \big[\Upsilon_{\frac{\tau}{\varepsilon}} \ | \setE(L,z)\big] .
\label{ineq:doob:ups}
\end{eqnarray}

By \eqref{eq:es:bound}, we have that
\begin{eqnarray}
  \| \mathbf{e}(n) \| = O(z^{-1}).
  \label{eq:exp:norm:e}
\end{eqnarray}
In addition, by 
\eqref{eq:Bs:hat:ex}, the following inequality 
for two given random
vectors $\mathbf{a}$ and $\mathbf{b}$:
 $$\ev\big[\|\mathbf{a}+\mathbf{b} \|\big] \leq \ev\big[\|\mathbf{a}\|\big]+\ev\big[\|\mathbf{b}\|\big],$$
and that for a given r.v. $x$, $\ev[x]\leq \sqrt{\ev[x^{2}]}$, we have that
\begin{align}
&  \ev\Big[\|\mathbf{B}(n) \|  \ \big| \setE(L,z) \Big] \leq
\Big[O(z^{0.5}L^{-0.5\zeta}) + O(z^{1.5}\varep^{0.5}) \Big].
\nonumber \\
  \label{eq:exp:B:t:norm}
\end{align}

From \eqref{eq:exp:norm:e} and \eqref{eq:exp:B:t:norm}, we have that
\begin{align}
  \ev\Big[\Upsilon_{\frac{\tau}{\varepsilon}}\ \big| \setE(L,z) \Big] &=\varepsilon \Big[\frac{\tau-t_{0}}{\varepsilon}+1\Big]
 \Big[O(z^{-1})+O(z^{0.5}L^{-0.5\zeta})
\nonumber \\ & \qquad \qquad \qquad \qquad   + O(z^{1.5}\varep^{0.5})  \Big].
\end{align}
Thus,
$$\limsup_{z\to \infty}\ \limsup_{L\to \infty} \ \ev \big[\Upsilon_{\frac{\tau}{\varepsilon}}  \ | \setE(L,z) \big] =0 .$$
Taking the limit of both sides of \eqref{ineq:doob:ups}, we obtain that for any $\eps_{1}>0$
\begin{align}
\limsup_{z\to \infty} \limsup_{L\to \infty} 
P \Big[ \sup_{\frac{t_{0}}{\varepsilon}\leq n \leq \frac{\tau}{\varepsilon}} \Upsilon_{n} >\eps_{1} \ | \setE(L,z) \Big] =0 ,
\label{eq:lim:sup:ups}
\end{align}
which means that 
\begin{align}
 \liminf_{z\to \infty} \liminf_{L\to \infty} 
P \Big[ \sup_{\frac{t_{0}}{\varepsilon}\leq n \leq \frac{\tau}{\varepsilon}} \Upsilon_{n} \leq \eps_{1} \ | \setE(L,z) \Big] =1.
\label{eq:lim:sup:ups}
\end{align}
The above limit states that \eqref{eq:upsilon:assum} holds with probability approaching one in the limit as first $L$ approaches infinity and then $z$ approaches
infinity, as required.

We next show that \eqref{ineq:rho:e:B:0:y:eps2} holds for sufficiently large $L$. By the discussion in Appendix~\ref{subsec:final:step}
 leading to \eqref{eq:dis:rho:e:b:0}, 
having that the event $\setE(L,z)$ occurs, we can show that
\begin{align}
  \sup_{\frac{t_{0}}{\varep} \leq n \leq \frac{\tau}{\varep} }\|\retheta(n\varepsilon)-\mathbf{x}(n\varepsilon)\|=O(\varepsilon), \
\text{as $L \to \infty$}.
\end{align}
Therefore , for any $\eps_{\theta}>0$, we 
can choose $L$ sufficiently large so that \eqref{ineq:rho:e:B:0:y:eps2} holds, completing the second step and 
the proof of the lemma.
  
\end{proof}



\begin{lemma}\label{lemma:cont:disc}
Suppose $z>1$. For $t\geq 0$, define $$n_{t} =\Big\lfloor\frac{t}{\varep}  \Big\rfloor.$$ 
We have that 
\begin{align}
 \lim_{L\to \infty} \sup_{t\in[t_{0},\tau] } \|\boldsymbol{\theta}_{L}(t) -\bftheta(n_{t}\varep)\|=0 \ (\text{in prob.})
\end{align}

\end{lemma}

\begin{proof}

Consider the evolution of the system in the time interval $[n\varepsilon, (n+1)\varepsilon)$, where $n \in \{0,1,2,\cdots\}$. 
Suppose for a given $n$, there exists some $t\in [n\varepsilon, (n+1)\varepsilon)$ such that
\begin{eqnarray}
  \|  \boldsymbol{\theta}_{L}(t)-  \boldsymbol{\theta}_{L}(n\varepsilon)\|>\eta_{2},
\label{ineq:dis:rho:dis}
\end{eqnarray}
where $\eta_{2}>0$.

Since $\boldsymbol{\theta}_{L}(t)$ is a four dimensional vector, the above inequality
 implies that within the time interval $[n\varepsilon, t]$, at least
one of the densities in the vector $\bftheta(t)$ must change by an amount of $\frac{\eta_{2}}{2}$. Therefore, a fraction $c\eta_{2}$, where 
$c>0$ is some constant, of links must change
their states at least once, from active to inactive or vice versa, from time $n\varepsilon$ to time $t$.
 Define $x_{l}(n)\in\{0,1\}$ as follows. Let
$x_{l}(n)=1$ if link $l$ changes its state at least once from time $n\varepsilon$ to
 time $(n+1)\varepsilon$, and we let $x_{l}(n)=0$, otherwise.

Based on the above discussion, \eqref{ineq:dis:rho:dis} implies that
\begin{eqnarray}
  \sum_{l\in \mathcal{L}} x_{l}(n) \geq cL \eta_{2}. \label{ineq:sum:x:beg}
\end{eqnarray}
We note that $x_{l}$'s may be correlated random variables. To address this difficulty, 
we focus on the times epochs after time $n\varepsilon$ that one link changes it state. Since packet transmission times
and back-off periods are continuous r.v.'s, w.p.1., no two links can change their state at the
same time. This implies that the amount of time from time $n\varepsilon$ until the first link changes
its state is well-defined, which we denote it by $s_{1}$. Hence, the first state change
occurs at time $n\varepsilon+s_{1}$. Similarly, let $s_{i}$, $i\geq2$, denote the amount time from the $(i-1)$th state
change to the $i$th state change. Hence, the $i$th state change after time $n\varepsilon$ occurs at time
$$n\varepsilon+\sum_{j=1}^{i}s_{j}.$$ 

Without loss of generality in the rest we assume that $cL\eta_{2}$ is an integer.
Assuming so, the condition in \eqref{ineq:sum:x:beg} implies that the $cL\eta_{2}$th state 
change after time $n\varepsilon$ must occur up to time $(n+1)\varepsilon$, and therefore
\begin{align}
  \sum_{i=1}^{cL\eta_{2}}s_{i}\leq \varepsilon.
\label{ineq:sum:s11}
\end{align}
Consider the time $n\varepsilon+\sum_{j=0}^{j=i-1}s_{j}$. At this time, there are a number of links that sense the channel as idle and try to access
the channel with rate $z$, and there are a number of links that are busy with packet transmission and 
with unit rate stop transmitting. Other remaining links in the network do not affect the time until the next
state change since their interfering neighbours are transmitting. Moreover, since packet transmission
 times and back-off periods are exponentially 
distributed, the time until the next state change is the minimum of all of these associated exponential r.v.'s. In addition,
since back-off timers and packet transmission 
times are independent of each other, given the state of system at time $n\varep+\sum_{j=0}^{j=i-1}s_{j}$, the time until the next state
 change is the minimum of 
a collection of independent exponential r.v.'s.
As a result, the r.v. $s_{i}$ is also an exponential r.v. with a rate that depends only
on system history up to time $n\varepsilon+\sum_{j=0}^{j=i-1}s_{j}$. Since there are at most $L$ links, for $z>1$, the r.v. 
$s_{i}$ becomes an exponential r.v. with maximum rate of $Lz$. Hence, independent of the system history up to time $n\varepsilon$
and the previous times where links states changed, r.v. $s_{i}$ is stochastically dominated by an exponential r.v. with rate
$Lz$ and we have that
\begin{align}
\ev\Big[s_{i} \ \Big|\{s_{i-1}, s_{i-2}, ..., s_{1}\}, \setH_{s}(n\varepsilon)\Big]\geq \frac{1}{Lz}
\label{ineq:ev:si:h}
\end{align}
where $\setH_{s}(t)$ is the history of the system up to and including time $t$.

Define $\upsilon(L)$ as
\begin{align}
\upsilon(L) =z\sqrt{L}=o(zL), \ \text{as $L\to \infty$}.
\label{ineq:def:ups}
\end{align}
Multiplying both sides of \eqref{ineq:sum:s11} by $-\upsilon(L)$, taking exponential function of both sides, and then applying Markov inequality, we obtain
\begin{eqnarray}
  P\Big[\sum_{i=1}^{cL\eta_{2}}s_{i}\leq \varepsilon \ \big| \setH_{s}(n\varepsilon) \Big] \leq 
\frac{\ev\big[e^{-\upsilon\sum_{i=1}^{cL\eta_{2} }s_{i} } \ | \setH_{s}(n\varepsilon) \big]}
{e^{-\upsilon \varepsilon}}
\label{ineq:prob:sum:x:b}
\end{eqnarray}
In the above and the following to simplify the notation, we have used $\upsilon$ as a short notation for $\upsilon(L)$.
Using \eqref{ineq:ev:si:h}, and noting that for an exponential r.v. $s$ with mean $\ev[s]$
$$\ev[e^{-\upsilon s}]=(1+\upsilon\ev[s])^{-1},$$
we can show that 
\begin{align}
  \ev\big[e^{-\upsilon\sum_{i=1}^{cL\eta_{2}} s_{i} }  \ | \setH_{s}(n\varepsilon) \big]\leq (1+\frac{\upsilon}{Lz})^{-cL\eta_{2} }.
\end{align}

 For a fixed $z$ and sufficiently large $L$, by \eqref{ineq:def:ups}, we have that 
\begin{align}
  (1+\frac{\upsilon}{Lz})^{-cL\eta_{2} }=  (1+\frac{\upsilon}{Lz})^{ \frac{Lz}{\upsilon} ( -\frac{c\upsilon\eta_{2}}{z}) }
 <\left[\frac{e}{2}\right]^{-\frac{c\upsilon\eta_{2}}{z}}
\end{align}
Hence, by the above and \eqref{ineq:prob:sum:x:b}, for sufficiently large $L$, 
\begin{align}
    P\Big[\sum_{i=1}^{cL\eta_{2}} s_{i}\leq \varepsilon  \ \big| \setH(n\varepsilon) \Big] <e^{-\upsilon(\frac{c\eta_{2}}{z}- \varepsilon)}=
e^{-\frac{\upsilon}{z}(c\eta_{2}- z\varepsilon)}.
\end{align}
Using \eqref{ineq:def:ups} and noting that by \eqref{eq:lim:z:varepsilon} for a fixed $z$
$$\lim_{L\to \infty } z\varepsilon=0,$$
 we can choose $L$ sufficiently large so that
\begin{align}
P\Big[\sum_{i=1}^{cL\eta_{2}}s_{i}\leq \varepsilon \ \big| \setH_{s}(n\varepsilon) \Big] <
e^{-\frac{c\upsilon\eta_{2}}{2z} }=e^{-\frac{c\sqrt{L}\eta_{2}}{2} }.
\end{align}

The discussion leading to \eqref{ineq:sum:s11} states that in order to have \eqref{ineq:dis:rho:dis} for
 some $t\in[n\varepsilon, (n+1)\varepsilon)$, the inequality 
in \eqref{ineq:sum:s11} should hold. Using the probability bound in the above, for sufficiently large $L$, we thus have that 
\begin{align}
&P\Big[  \exists t \in[n\varepsilon, (n+1)\varepsilon), \text{ s.t. }
\nonumber \\
& \qquad \|  \boldsymbol{\theta}_{L}(t)-  \boldsymbol{\theta}_{L}(n\varepsilon)\|>\eta_{2} \ \big| \setH_{s}(n\varepsilon) \Big]
<e^{-\frac{c\sqrt{L}\eta_{2}}{2} }.
\end{align}
Note that the above probability bound is independent of the system history up to and 
including time $n\varepsilon$. Hence, considering the event that 
in \emph{none} of the intervals $[n\varepsilon, (n+1)\varepsilon)$, 
$0 \leq n \leq  \tau\varepsilon^{-1}-1 \}$, the inequality of \eqref{ineq:dis:rho:dis} holds for 
some $t \in[n\varepsilon, (n+1)\varepsilon)$, we have 
\begin{eqnarray}
P\Big[\sup_{0\leq n \leq \frac{\tau}{\varepsilon}-1} \ \sup_{ n\varepsilon\leq t \leq (n+1)\varepsilon  } 
\| \boldsymbol{\theta}_{L}(t)-\boldsymbol{\theta}_{L}(n\varepsilon)\|\leq \eta_{2} \Big]
\nonumber \\ \
\geq   \Big(1- e^{-\frac{c\sqrt{L}\eta_{2}}{2}} \Big)^{\frac{\tau}{\varepsilon} }.
\label{ineq:fin:prob:error}
\end{eqnarray}
Since $\varepsilon=\lceil L^{-(1-\zeta)} \rceil$, it follows that
\begin{eqnarray}
\lim_{L\to \infty}    \Big(1- e^{-\frac{c\sqrt{L}\eta_{2}}{2}} \Big)^{\frac{\tau}{\varepsilon} }=1.
\label{ineq:p:lim:one}
\end{eqnarray}
Using this limit and the bound in \eqref{ineq:fin:prob:error}, we obtain the statement in the lemma, as required.

\end{proof}


\begin{lemma}\label{lemma:x:cond}
Under the assumptions in Lemma~\ref{lemma:t1},
there exists a constant $z_{0}$ such that following holds for $z>z_{0}$. There exists a $t_{1}$ independent of $z$, where 
$$t_{0}\leq t_{1}<2,$$
 such that
$\mathbf{x}(t)$ uniquely exists over the interval $[t_{0},t_{1}]$, and
for all $t\in[t_{0},t_{1}]$, we have $\mathbf{x}(t)\in\mathcal{D}_{w}$. 
Moreover, we have
  \begin{align}
    0<c_{x_{1},1}'<x_{1}(t_{1})<c_{x_{1},2}'<0.5 ,\label{ineq:init:x1:t1:l}
  \end{align}
where $c_{x_{1},1}'$ and $c_{x_{1},2}'$ are constants independent of $z$, and for $2\leq i \leq 4$,
\begin{align}
  x_{i}(t_{1})=O(z^{-1}).
\label{ineq:init:x2:4:t1:l}
\end{align}
In addition,
\begin{align}
&  \Big|x_{2}(t_{1})-\frac{1}{3z} c_{R}\delta_{\mathbf{x}}(t_{1})^{3}\Big|=O(z^{-2}),
\label{eq:x2:t1:l}
\\ &
  \Big|x_{3}(t_{1})-\frac{2}{3z} c_{R}\delta_{\mathbf{x}}(t_{1})^{3} \Big|=O(z^{-2}).
\label{eq:x3:t1:l}
\\ &
  \Big|x_{4}(t_{1})-\frac{1}{z}\big[x_{1}(t_{1})-c_{R}\delta_{\mathbf{x}}(t_{1})^{3}\big] \Big|=O(z^{-2}). \label{eq:x4:t1:l}
\end{align}

\end{lemma}

\begin{proof} 

By the assumption in the lemma on $\bfx_{t_{0}}=\bfx(t_{0})$, for large $z$,
\begin{align}
  \|\bfx(t_{0})\|<1 \label{ineq:xt0:assump:l}.
\end{align}
In the rest, since the lemma is stated for $z>z_{0}$, we assume $z_{0}$ is chosen sufficiently large
 so that the above inequality holds for $z>z_{0}$.

Consider the ODE of \eqref{eq:diff:cont} and suppose $z>z_{0}$. 
Since the RHS of this ODE is a locally Lipschitz continuous function of $\bfx$ and continuously differentiable
with respect to $x_{i}$'s, $1\leq i \leq 4$, we have that starting
from the bounded $\mathbf{x}_{t_{0}}$, see \eqref{ineq:xt0:assump:l}, the solution exists and is unique until at least where
we have $\|\bfx(t) \| \geq 2$ for some $t\geq t_{0}$ (see e.g., Theorem 1.9 and Theorem 1.13 in \cite{markley:book}). Define time $t_{b}$ to be the first time after $t_{0}$ that 
$$\| \bfx(t_{b})\|\geq 2.$$
By the definition of $t_{b}$, $\bfx(t)$ uniquely exists over the interval
$[t_{0},t_{b})$, and we have that
\begin{align}
  \| \bfx(t)\|\leq 2, \ t\in[t_{0},t_{b}). \label{ineq:abs:x:tb}
\end{align}
If time $t_{b}$ does not exist, define $t_{b}=\infty$. By \eqref{ineq:xt0:assump:l}, 
we have $$t_{b}>t_{0}.$$ 
By the definition of $f(\bfx)$ in \eqref{eq:diff:cont} and the inequality \eqref{ineq:abs:x:tb}, we also have that
\begin{align}
  |f(\bfx)| \leq c \mathbf{1}_{4\times 1}, \ t\in[t_{0},t_{b}), \label{ineq:abs:fx:tb}
\end{align}
for some constant $c>0$.
By the ODE of \eqref{eq:diff:cont}, we also have that for all $t\in[t_{0},t_{b})$,
\begin{align}
  \bfx(t)=\bfx(t_{0})e^{\bfA(t-t_{0})}+\int_{t_{0}}^{t} e^{\bfA(t-t')} f(\bfx(t'))dt'. \label{eq:bx:sol:int:ex}
\end{align}
Since all eigenvalues of matrix $\bfA$ are real, distinct, and non-positive, we can find a constant $c_{\bfA}>0$ independent of
$z$ such that for $t' \leq t$
\begin{align}
  \left| e^{\bfA(t-t')}\right|\leq c_{\bfA} \mathbf{1}_{4 \times 4}.\label{ineq:eA:x}
\end{align}
Considering a time $t\in [t_{0},t_{b})$, by \eqref{ineq:abs:x:tb}-\eqref{ineq:eA:x}, we have
\begin{align}
  |\bfx(t)| &\leq  2 c_{\bfA} \mathbf{1}_{4 \times 1} \mathbf{1}_{4 \times 4} +\int_{t_{0}}^{t}\left| e^{\bfA(t-t')}\right|
 \left| f(\bfx(t')) \right|dt'
\nonumber \\ 
& \leq 8  c_{\bfA} \mathbf{1}_{4 \times 1} + c  c_{\bfA}(t-t_{0})  \mathbf{1}_{4 \times 4} \mathbf{1}_{4 \times 1} 
\nonumber \\
&
\leq 4 c c_{\bfA} (t-t_{0}+2c^{-1})\mathbf{1}_{4\times 1} . \label{eq:bx:sol:int:ex}
\end{align}

The above inequality states that each element of $\bfx(t)$ independent of $z$ increases at most linearly with time $t$ before reaching
time $t_{b}$. Therefore, by \eqref{ineq:xt0:assump:l}, there must exist 
a time interval 
$$I_{0}=[t_{0},t_{i}]$$
 such that
\begin{align}
  \|\mathbf{x}(t)\| \leq 2, \ t\in I_{0}
\end{align}
where $t_{i}$ can be chosen independent of $z$ such that
\begin{align}
  t_{0}<t_{i}<2.
\end{align}
 As a result, 
\begin{align}
|x_{i}(t)|\leq 2,  \ t\in I_{0}. \label{ineq:xi:less:2:l}
\end{align}
Using the above inequality, we have that 
\begin{align}
\left|  \delta_{\bfx}^{3}(t) \right|\leq 8, \ t\in I_{0} \label{ineq:hat:delta:B:x}
\end{align}
where $\delta_{\bfx}(t)$ is defined in \eqref{def:hat:delta:bfx}.

By the ODE for $x_{2}(t)$, as given by \eqref{eq:diff:cont}, we have
\begin{align}
  \frac{d}{dt} x_{2}=c_{R}\delta_{\bfx}^{3}-3zx_{2}. \label{eq:diff:x2:lem:in}
\end{align}
Using the above ODE, \eqref{ineq:hat:delta:B:x}, and the assumption in the lemma that \eqref{ineq:assum:xs:lemma:m} holds, 
i.e., that $x_{2}(t_{0})=O(z^{-1})$, 
we can show that
\beqa
\sup_{t\in I_{0}} |x_{2}(t)|=O(z^{-1}) . \label{eq:x2:bness}
\eeqa
Similarly, we can show that
\beqa
\sup_{t\in I_{0} } |x_{3}(t)|=O(z^{-1}), \ \ \sup_{t\in I_{0}} |x_{4}(t)|=O(z^{-1}). \label{eq:x3:4:bness}
\eeqa

Using \eqref{ineq:xi:less:2:l}, \eqref{eq:x2:bness}, and \eqref{eq:x3:4:bness} in the ODE for $x_{1}(t)$ given in
\eqref{eq:diff:cont}, we 
have that
\begin{align}
\sup_{t\in I_{0}}\left|\frac{d}{dt}x_{1} \right|=O(1),\ \text{ as $z\to \infty$}.   \label{ineq:x1:bound:der:l}
\end{align}
 Hence, for any $\eps_{1}>0$, we can find a sub-interval
 \begin{align}
   I_{1}=[t_{0},t_{i}'] \subset I_{0} \label{eq:def:I1:I0}
 \end{align}
 with 
$$t_{0}< t_{i}' <t_{i},$$
 and $t_{i}'$ independent of $z$ such that 
$$|x_{1}(t)-x_{1}(t_{0})|<\eps_{1}, \  t \in I_{1}.$$
In particular, since by assumption in the lemma
$$0<c_{x_{1},1}<x_{1}(t_{0})<c_{x_{1},2}<0.5,$$
we can choose $\eps_{1}$ sufficiently small such that for all $t\in I_{1}$
\begin{align}
  0<c_{x_{1},1}'<x_{1}(t)<c_{x_{1},2}'<0.5,  \label{ineq:x1:I1:B}
\end{align}
where $c_{x_{1},1}'$ and $c_{x_{1},2}'$ are constants independent of $z$.

 Using \eqref{ineq:x1:I1:B}, the definition of $\delta_{\bfx}(t)$ given in \eqref{def:hat:delta:bfx}, 
the assumption in the lemma that $x_{2}(t_{0})=O(z^{-1})$, and the ODE
for $x_{2}(t)$, given in \eqref{eq:diff:x2:lem:in}, we 
 can show that
 \begin{align}
&   0<x_{2}(t), \ t\in I_{1}, \nonumber \\
& \sup_{t\in I_{1}} x_{2}(t)=O(z^{-1}). \label{ineq:obj:x2}
 \end{align}
Similarly, for $3\leq i \leq 4$, we obtain that
\begin{align}
&   0<x_{i}(t),  \ t\in I_{1}, \nonumber \\
& \sup_{t\in I_{1}}x_{i}(t)=O(z^{-1}). \label{ineq:obj:x4}
\end{align}
Considering \eqref{ineq:x1:I1:B}, \eqref{ineq:obj:x2}, and \eqref{ineq:obj:x4}, for $z>z_{0}$,
we have found an interval $I_{1}\subset I_{0}\subset [t_{0},t_{b})$
independent of $z$ within which $\mathbf{x}(t)$ uniquely exists and is well-defined, i.e., we have 
that 
\begin{align}
  \mathbf{x}(t)\in\mathcal{D}_{w}  , \ t\in I_{1}. \label{eq:x:in:Dw:l}
\end{align}

We next show that \eqref{eq:x2:t1:l}-\eqref{eq:x4:t1:l} hold. Here, we provide the proof for \eqref{eq:x2:t1:l}. The proofs for 
\eqref{eq:x3:t1:l} and \eqref{eq:x4:t1:l} follow from similar
lines. Consider the ODE for $x_{2}(t)$ given in \eqref{eq:diff:x2:lem:in}, which can be used to express $x_{2}(t)$
as a function of $c_{R}\delta_{\bfx}^{3}(t)$ for $t\in I_{1}$ as
\begin{align}
  x_{2}(t)=x_{2}(t_{0})e^{-3zt}+\int_{t_{0}}^{t}e^{-3z(t-t')} c_{R}\delta_{\bfx}^{3}(t') dt'. \label{eq:int:x4:lem}
\end{align}

Recall that $I_{1}=[t_{0},t_{i}']$ where $t_{i}'$, $t_{i}'>t_{0}$, is chosen independently of $z$. This allows us to choose a sufficiently small
  constant 
$\eta>0$ and choose $z_{0}$ sufficiently large so that for $z>z_{0}$ and $k=\sqrt{z}$, we have
\begin{align}
t_{0}+\eta+ \frac{k}{z}<t_{i}'. \label{ineq:t0:eta:ti'}
\end{align}
Consider a time $t\in(t_{0}+\eta+ \frac{k}{z},t_{i}']\subset I_{1}$. Using \eqref{eq:int:x4:lem}, we have
\begin{align}
  x_{2}(t)=\int_{t-\frac{k}{z}}^{t}e^{-3z(t-t')} c_{R}\delta_{\bfx}^{3}(t')  dt' +e_{x_{2}}(t), \label{eq:int:x2:error:l}
\end{align}
where by \eqref{ineq:xt0:assump:l} and \eqref{ineq:hat:delta:B:x}, we have
\begin{align}
  e_{x_{2}}(t)&= x_{2}(t_{0})e^{-3z(t-t_{0})}+\int_{t_{0}}^{t-\frac{k}{z}}e^{-3z(t-t')} c_{R}\delta_{\bfx}^{3}(t') dt' 
\nonumber \\ 
&<e^{-3z(t-t_{0})}+\frac{8c_{R}}{3z}\big(e^{-3k}-e^{-3z(t-t_{0})} \big). \label{eq:error:int:boound:l}
\end{align}

By Mean-Value Theorem, we have that
\begin{align}
&  \int_{t-\frac{k}{z}}^{t}e^{-3z(t-t')} c_{R}\delta_{\bfx}^{3}(t')  dt'
\nonumber \\ &=\int_{t-\frac{k}{z}}^{t}e^{-3z(t-t')} \Big[c_{R} \delta_{\bfx}^{3}(t) + 
(t'-t) c_{R} \frac{d}{dt} \big[\delta_{\bfx}^{3}(t'')\big]  \Big] dt',
\label{eq:int:delta:m:l}
\end{align}
where $t''$ is a function of $t'$ and $t$, and we have that 
$$t''\in(t',t) \subset I_{1}\subset I_{0}.$$
 Using \eqref{def:hat:delta:bfx}, \eqref{ineq:hat:delta:B:x}, 
and \eqref{ineq:x1:bound:der:l}, we have that
\begin{align}
\sup_{t\in I_{0}}\left|  \frac{d}{dt} \delta_{\bfx}^{3}(t) \right| =O(1),\  \text{as $z\to \infty$},
\end{align}
which can be used along with \eqref{eq:int:delta:m:l} to show that
\begin{align}
& \sup_{t\in (t_{0}+\eta+ \frac{k}{z},t_{i}']} \Big|  \int_{t-\frac{k}{z}}^{t}e^{-3z(t-t')} c_{R}\delta_{\bfx}^{3}(t')  dt' -\frac{1}{3z} c_{R}\delta_{\bfx}^{3}(t) \Big|
\nonumber \\ & \qquad \ 
<O(1) e^{-3k}\Big[z^{-1}+k z^{-2}\Big]+ O(1)(z)^{-2}. \label{ineq:x2:delta3:bound}
\end{align}

Since $k=\sqrt{z}$, combining \eqref{eq:int:x2:error:l}, \eqref{eq:error:int:boound:l}, and \eqref{ineq:x2:delta3:bound}, we obtain
that
\begin{align}
  \sup_{t\in(t_{0}+\eta+\frac{k}{z},t_{i}']} \left| x_{2}(t)-\frac{1}{3z}c_{R}\delta_{\bfx}^{3}(t) \right|=O(z^{-2}). \label{eq:x4:bound:lf}
\end{align}
Since $t_{i}'$ does not depend on $z$, we can choose $\eta$ sufficiently small and $z_{0}$ sufficiently large such that for $z>z_{0}$, we have
$$t_{1}=t_{0}+2\eta \in (t_{0}+\eta+\frac{k}{z},t_{i}']. $$

Hence, by \eqref{eq:x4:bound:lf}, we have that \eqref{eq:x2:t1:l} holds for $t_{1}$
as defined above. Similarly, we have that
\eqref{eq:x3:t1:l} and \eqref{eq:x4:t1:l} hold. Moreover, by the definition of $t_{1}$, we have
that for $z>z_{0}$
$$[t_{0}, t_{1}]\subset I_{1}.$$
 Hence, by \eqref{ineq:x1:I1:B}-\eqref{ineq:obj:x4}, we have
\eqref{ineq:init:x1:t1:l}-\eqref{ineq:init:x2:4:t1:l} hold. Moreover, since $[t_{0}, t_{1}]\subset I_{1}$, 
by \eqref{eq:x:in:Dw:l}, we have that for $z>z_{0}$, $\bfx(t)$ is uniquely defined over the interval $[t_{0},t_{1}]$, and that 
for all $t\in[t_{0},t_{1}] $, $\bfx(t)\in \setD_{w}$,
completing the proof.
\end{proof}



\begin{lemma}\label{lemma:yh:xh}
Suppose $\bfy(t)$ and $\bfx_{h}(t)$ are such that we have
\begin{align}
  \frac{d}{dt}\mathbf{y}&=\bfA\mathbf{y}+h(\mathbf{y})+\mathbf{e}_{\mathbf{y}},
\nonumber \\
    \frac{d}{dt}\mathbf{x}_{h}&=\bfA\mathbf{x}_{h}+h(\mathbf{x}_{h}),
\end{align}
with the initial condition at time $t_1<\tau$ that 
$$\bfx_{h}(t_{1})=\bfy(t_{1})=\bfy_{t_{1}}$$
for some constant vector $\bfy_{t_{1}}$, and where
\begin{align}
 \sup_{t\in [t_{1},\tau]}  \| \mathbf{e}_{\mathbf{y}}(t) \|=O(z^{-1}). \label{eq:ey:B:assum}
\end{align}

Moreover, suppose the function $h(\cdot)$ is a Lipschitz continuous function with the property that for any $\mathbf{x}_{1}$
and $\bfx_{2}$
\begin{eqnarray}
h(\mathbf{x}_{1}) -h(\mathbf{x}_{2})=\mathbf{L}_{h}(\mathbf{x}_{1} ,\mathbf{x}_{2} ) \big[\mathbf{x}_{1} -\mathbf{x}_{2} \big]
  \label{eq:diff:lip:h}
\end{eqnarray}
where $\mathbf{L}_{h}(\mathbf{x}_{1} ,\mathbf{x}_{2} )$ is a $4\times 4$ matrix with elements $[\mathbf{L}_{h} (\mathbf{x}_{1} ,\mathbf{x}_{2} )]_{ij}$ such that
\begin{align}
   |[\mathbf{L}_{h}(\mathbf{x}_{1} ,\mathbf{x}_{2} )]_{ij}| \leq L_{h},
\label{ineq:Lm:b}
\end{align}
for some constant $L_{h}>0$. Let
 \begin{eqnarray}
  \mathbf{e}_{h}(t)=\mathbf{y}(t)-\mathbf{x}_{h}(t). \label{eq:def:eg}
\end{eqnarray}
We then have that
\begin{align}
\sup_{t\in[t_{1},\tau]}  \| \mathbf{e}_{h}(t) \| =O(z^{-1}).
\end{align}
\end{lemma}

\begin{proof}
Using \eqref{eq:diff:lip:h} and \eqref{eq:def:eg}, we have that
\begin{eqnarray}
  \label{eq:diff:lip:fg}
h(\mathbf{y}(t)) -h(\mathbf{x}_{h}(t))=\mathbf{L}_{h}(t)\mathbf{e}_{h}(t)
\end{eqnarray}
where by \eqref{ineq:Lm:b}, the matrix $\mathbf{L}_{h}(t)$ is such for the absolute value of its element $[\mathbf{L}_{h}(t)]_{ij}$, we have
\begin{align}
  \big| [\mathbf{L}_{h}(t)]_{ij} \big| \leq L_{h}. \label{ineq:lh:ij}
\end{align}

Subtracting the ODE for $\bfx_{g}(t)$ from that of $\bfy(t)$, as given in the lemma, we obtain
\begin{eqnarray}
    \frac{d}{dt}\mathbf{e}_{h}=\bfA\mathbf{e}_{h}+\mathbf{L}_{h}  \mathbf{e}_{h}+\mathbf{e}_{\mathbf{y}}=
\big[\bfA+\mathbf{L}_{h} \big] \mathbf{e}_{h}+\mathbf{e}_{\mathbf{y}}.
\label{eq:diff:e:g:L}
\end{eqnarray}
By the initial conditions in the lemma, we have 
$$\bfx_{h}(t_{1})=\bfy(t_{1}),$$
 and hence,
  $$\mathbf{e}_{h}(t_{1})=0.$$
Using this initial condition for $\mathbf{e}_{h}(t)$ and the ODE in \eqref{eq:diff:e:g:L}, we have
\begin{align}
  \mathbf{e}_{h}(t)= \bfV(t)\int_{t_{1}}^{t} \bfV(t')^{-1} \mathbf{e}_{\mathbf{y}}(t') dt'
\label{eq:eg:exa:sol}
\end{align}
where $\bfV(t)$ is the fundamental matrix solution for the ODE in \eqref{eq:diff:e:g:L} and is given by
\begin{align}
  \bfV(t)= e^{\int_{t_{1}}^{t} [\bfA+\mathbf{L}_{h}(r)] dr}.
\end{align}
Plugging this into \eqref{eq:eg:exa:sol}, we have
\begin{align}
  \mathbf{e}_{h}(t)=  \int_{t_{1}}^{t} e^{\int_{t'}^{t} [\bfA+\mathbf{L}_{h}(r)] dr } \mathbf{e}_{\mathbf{y}}(t') dt'  
\label{eq:eg:detailed}.
\end{align}

Since all eigenvalues of matrix $\bfA$ are real, distinct, and non-positive, we can find a constant $c_{\bfA}>0$ independent 
of $z$ such that for $t'\leq t$
\begin{align}
   \left| e^{(t-t') \bfA } \right| \leq c_{\bfA} \mathbf{1}_{4\times 4}
\end{align}
where the inequality is component-wise.
Using the assumption in \eqref{eq:ey:B:assum}, we then have that
\begin{align}
& \sup_{t\in[t_{1},\tau]} \left| \int_{t_{1}}^{t} e^{\int_{t'}^{t} \bfA dr } \mathbf{e}_{y}(t') dt' \right| 
\nonumber \\ & 
\qquad \qquad \leq
 \sup_{t\in[t_{1},\tau]} \int_{t_{1}}^{t} \left| e^{(t-t') \bfA } \right| \left| \mathbf{e}_{y}(t') \right| dt'
\nonumber \\ & \qquad \qquad 
\leq (\tau-t_{1})O(z^{-1})\mathbf{1}_{4\times 1}=O(z^{-1})\mathbf{1}_{4\times 1}.
\label{eq:int:eA}
\end{align}

By \eqref{ineq:lh:ij}, we also have that
\begin{align}
  \big|\mathbf{L}_{h}(t)\big| \leq L_{h}\mathbf{1}_{4\times 4},
\end{align}
which can be used to show that independent of $z$
\begin{align}
\sup_{t\in[t_{1},\tau]}|  e^{\int_{t'}^{t} \mathbf{L}_{h}(r) dr } | &\leq e^{\int_{t'}^{\tau}  L_{h}\mathbf{1}_{4\times 4} dr }
=e^{(\tau-t') L_{h}\mathbf{1}_{4\times 4}}
\nonumber \\ & 
\leq c_{h} e^{4L_{h}(\tau-t')}\mathbf{1}_{4\times 4}
\end{align}
where $c_{h}>0$ is a constant independent of $z$.
Therefore, by \eqref{eq:ey:B:assum}, we have that
\begin{align}
&\sup_{t\in[t_{1},\tau]} \left|  \int_{t_{1}}^{t} e^{\int_{t'}^{t} \mathbf{L}_{h}(t) dr } \mathbf{e}_{y}(t') dt' \right|
\nonumber \\
& \qquad \qquad \leq 
c_{h} \int_{t_{1}}^{\tau}  e^{4L_{h}(\tau-t')}\mathbf{1}_{4\times 4}
  \left| \mathbf{e}_{y}(t')  \right| dt'
\nonumber \\ & \qquad \qquad 
=O(z^{-1})e^{4L_{h}(\tau-t_1)}  \mathbf{1}_{4\times 1}
\nonumber \\ & \qquad \qquad 
=O(z^{-1})\mathbf{1}_{4\times 1}.
\label{eq:int:Lh}
\end{align}

Using \eqref{eq:eg:detailed}, \eqref{eq:int:eA}, and \eqref{eq:int:Lh}, we obtain that
\begin{align}
\sup_{t\in[t_{1},\tau]}  \left| \mathbf{e}_{h}(t) \right|=O(z^{-1})\mathbf{1}_{4\times 1},
\end{align}
which implies the statement in the lemma, completing the proof.
 
\end{proof}


\begin{lemma}\label{lemma:y:diff:B}
  Suppose for some $t_{1}<\tau$, we have $\bfy(t_{1}) \in \setD_{w}$ such that
  \begin{align}
    0<c_{y_{1},1}'<y_{1}(t_{1})<c_{y_{1},2}'<0.5 ,\label{ineq:init:x1:t1}
  \end{align}
where $c_{y_{1},1}'$ and $c_{y_{1},2}'$ are constants independent of $z$. In addition, suppose
for $I=\{t_{1}\}$, we have
\begin{align}
 & \sup_{t\in I } \ \max\big[y_{2}(t), y_{3}(t),y_{4}(t)\big]=O(z^{-1}),
\label{ineq:init:y2:4:t1:lemma}
\\
&  \sup_{t\in I } \Big|y_{2}(t_{1})-\frac{1}{3z} c_{R}\delta_{\mathbf{y}}^{3}(t) \Big|=O(z^{-2}),
\label{eq:y2:t1:lemma}
\\ &
\sup_{t\in I }   \Big|y_{3}(t)-2 y_{2}(t) \Big|=O(z^{-2}).
\label{eq:y3:t1:lemma}
\\ &
\sup_{t\in I }   \Big|y_{4}(t)-\frac{1}{z}\big[y_{1}(t)-c_{R}\delta_{\mathbf{y}}(t)^{3}\big] \Big|=O(z^{-2}). \label{eq:y4:t1:lemma}
\end{align}
We then have that $\mathbf{y}(t)\in \mathcal{D}_{w}$, for $t\in[t_{1},\tau]$, and 
\begin{align}
  \sup_{t\in[t_{1},\tau]}\left | \frac{d}{dt} \bfy \right |=O(1)\mathbf{1}_{4\times 1}.
\end{align}
Moreover, we have that
 \eqref{ineq:init:y2:4:t1:lemma}-\eqref{eq:y4:t1:lemma} hold for
$$I= [t_{1},\tau].$$
\end{lemma}

\begin{proof}

Using the ODE for $y_{1}(t)$ in \eqref{eq:diff:y}, we obtain
 \begin{eqnarray}
   y_{1}(t)=0.5-\al' (1+\be' (t-t_{1}))^{-\frac{1}{2}},
   \label{eq:y1:exact:sol:lemm}
 \end{eqnarray}
where
\begin{align}
  \al'=0.5-y_{1}(t_{1}), \label{eq:al:lemm}
\end{align}
and 
\begin{align}
  \be'=\frac{4c_{R}}{3}(\al')^{2}.\label{eq:be:lemm}
\end{align}

 Having determined $y_{1}(t)$, we then can use \eqref{eq:diff:y} and $y_{1}(t)$ to find $y_{2}(t)$, $y_{3}(t)$, and $y_{4}(t)$. Since
by \eqref{eq:y1:exact:sol:lemm},
\begin{align}
  0< y_{1}(t)<0.5, \label{ineq:y1:const}
\end{align}
 we can use the constraints on the initial conditions given by \eqref{ineq:init:y2:4:t1:lemma} at time $t=t_{1}$ to show
 that $y_{2}(t)$, $y_{3}(t)$, and $y_{4}(t)$  are all positive
 for $t\in[t_{1},\tau]$, and that \eqref{ineq:init:y2:4:t1:lemma} holds for $I=[t_{1},\tau]$.
Therefore, $\mathbf{y}(t)\in \mathcal{D}_{w}$, for $t\in[t_{1},\tau]$. Having \eqref{ineq:y1:const} and that 
\eqref{ineq:init:y2:4:t1:lemma} holds for $I=[t_{1},\tau]$, it follows from \eqref{eq:diff:y} that
\begin{align}
  \sup_{t\in[t_{1},\tau]}\left | \frac{d}{dt} \bfy \right |=O(1)\mathbf{1}_{4\times 1}. \label{eq:diffy:O:lim}
\end{align}

We next prove that \eqref{eq:y4:t1:lemma} holds for $I=[t_{1},\tau]$. The proof for \eqref{eq:y2:t1:lemma} and \eqref{eq:y3:t1:lemma} 
follows from similar steps. By \eqref{eq:diffy:O:lim}, we have
\begin{align}
  \sup_{t\in[t_{1},\tau] } \Big|\frac{d}{dt} y_{1}(t)\Big|=O(1), \ \text{as $z\to \infty$}\label{eq:der:bound:y1}.
\end{align}  
Now consider the ODE for $y_{4}(t)$ given by \eqref{eq:diff:y}. We have that
\begin{eqnarray}
  \frac{d}{dt}y_{4}=\big[y_{1}-c_{R}[0.5-y_{1}]^{3} \big]-zy_{4}. \label{eq:diff:ode:y4}
\end{eqnarray}
Define $y_d(t)$ as
\begin{align}
  y_d(t)=y_{1}-c_{R}[0.5-y_{1}]^{3}.
\end{align}
By \eqref{ineq:y1:const} and \eqref{eq:der:bound:y1}, we have that
\begin{align}
\sup_{t\in[t_{1},\tau] } \left|  \frac{d}{dt}y_{d}\right |=O(1), \ \text{as $z \to \infty$} \label{eq:der:bound:yd}.
\end{align}
By Mean-Value Theorem, for $t\in[t_{1},\tau]$, we have
\begin{align}
  y_{d}(t)=y_{d}(t_{1})+(t-t_{1}) \frac{d}{dt}y_{d}(t_{i}) \label{eq:mvt:yd}
\end{align}
where $t_{i}$ is a function of $t_{1} $ and $t $  such that $t_{i}\in (t_{1},t)$.

 Knowing $y_{1}(t)$ and using \eqref{eq:diff:ode:y4}, we can find 
$y_{4}(t)$ as
\begin{align}
  y_{4}(t)=y_{4}(t_{1})e^{-z(t-t_{1})}+\int_{t_{1}}^{t}e^{-z(t-t')} y_{d}(t')dt'. \label{eq:int:x4}
\end{align}
To obtain the statement in the lemma, we use \eqref{eq:int:x4}. By \eqref{eq:mvt:yd} and \eqref{eq:int:x4}, we have
\begin{align}
  y_{4}(t)& =y_{4}(t_{1})e^{-z(t-t_{1})} +\int_{t_{1}}^{t}e^{-z(t-t')} y_{d}(t)dt' + e_{y_{4},1}(t)
\label{eq:y4:exp}
\end{align}
where
\begin{align}
e_{y_{4},1}(t)=  \int_{t_{1}}^{t}e^{-z(t-t')}  
(t'-t) \frac{d}{dt}y_{d}(t'') dt',
\end{align}
in which $t''$ is a function of $t'$ and $t$ such that $t''\in(t',t)$. Using \eqref{eq:der:bound:yd}, we have
that
\begin{align}
\sup_{t\in[t_{1},\tau] }\big|e_{y_{4},1}(t)\big|=O(z^{-2}). \label{eq:bound:yd:t:t1}
\end{align}

We also can write 
\begin{align}
&  y_{4}(t_{1})e^{-z(t-t_{1})} +\int_{t_{1}}^{t}e^{-z(t-t')} y_{d}(t)dt'
\nonumber \\ & \qquad \qquad = e^{-z(t-t_{1})} \big[ y_{4}(t_{1}) - \frac{1}{z}y_{d}(t)\big]+\frac{1}{z}y_{d}(t).
\label{eq:y4:int:i1}
\end{align}
Moreover, using \eqref{eq:der:bound:yd} and \eqref{eq:mvt:yd}, and noting that $xe^{-xz}\leq z^{-1}e^{-1}$, we have
that
\begin{align}
& \sup_{t\in[t_{1},\tau]} \left|  e^{-z(t-t_{1})} \frac{1}{z}y_{d}(t)-  e^{-z(t-t_{1})} \frac{1}{z}y_{d}(t_{1}) \right|
\nonumber \\ & \qquad \qquad \leq \sup_{t\in[t_{1},\tau]} e^{-z(t-t_{1})}  \frac{1}{z} (t-t_{1}) O(1)=O(z^{-2}).
\label{eq:yd:i1}
\end{align}

Using \eqref{eq:y4:exp}, \eqref{eq:bound:yd:t:t1}, \eqref{eq:y4:int:i1}, and \eqref{eq:yd:i1}, we have
\begin{align}
&\sup_{t\in[t_{1},\tau]}\left|  y_{4}(t)-\frac{1}{z}y_{d}(t)\right| \nonumber \\
& \qquad  \leq \sup_{t\in[t_{1},\tau]} e^{-z(t-t_{1})} \Big[ y_{4}(t_{1}) - \frac{1}{z}y_{d}(t_{1})\Big] +O(z^{-2}).
\end{align}
Using the above and the assumption in the lemma given in \eqref{eq:y4:t1:lemma} for $I=\{t_{1}\}$, we then have that
\begin{align}
\sup_{t\in[t_{1},\tau]}   \left|  y_{4}(t)-\frac{1}{z}y_{d}(t)\right|=O(z^{-2}),
\end{align}
as required.
  
\end{proof}

\begin{lemma}\label{lemma:B:eps}
Consider the lattice $G_{L}$, attempt rate $z>0$, and a given time $t$, $0<t<\infty$. We have, w.p.1,
  \begin{align}
\frac{\sum_{\setC\in \setC_{L}^{(nd)}(t,z)} \ell(\setC)}{\sqrt{2}L} \geq \delta_{L}(t)-\eta(L)-5\theta_{L,h}(t)
  \nonumber
  \end{align}
where $\eta(L)>0$ and 
$$\lim_{l\to \infty} \eta(L)=0,$$
and $\theta_{L,h}(t)$ is the density of links that are inactive and sense the channel as idle at time 
$t$, as defined in Appendix~\ref{sec:all:densities}.
\end{lemma}
\begin{proof}
  The lemma follows from Lemma~\ref{lemma:area} and Lemma~\ref{lemma:bound:Ac}.
\end{proof}


\begin{lemma}\label{lemma:area}
Consider the lattice interference graph $G_{L}$, attempt rate $z>0$, and a time $t$ where $0<t<\infty$. Recall that $\theta_{L}(t)$ represent both density and the set 
of active links at time $t$. Let
$$A^{c}_{L}(t)=n^{2}-\sum_{l\in\theta_{L}(t)} A_{l}$$
 be the total area of the lattice not covered by the union of the coverage areas of
 active links at time $t$, where $A_{l}$ is defined in Appendix~\ref{sec:definitions}. Then, w.p.1,
\beqa
\left| \frac{A^{c}_{L}(t)} {2 L}-\delta_{L}(t) \right| <\eta(L)
\label{eq:lim:leps:Ac}
\eeqa
where $ \eta(L)$ only depends on $L$, and
$$ \lim_{L\to \infty } \eta(L)=0.$$
\end{lemma}

\begin{proof} 
By the definition of $A^{c}_L(t)$, the area
covered by active links at time $t$ is $n^{2}-A^{c}_{L}(t)$. Separating the contribution
 of active links on the boundary $\partial G_{L}$ of the lattice $G_L$ from those of active links
inside the lattice with $A_{l}=2$, we obtain
\begin{align}
&  n^{2}-A^{c}_{L}(t) =\sum_{l\in \big(\partial G_{L} \cap \theta_{L}(t)\big)  }A_{l}
+\sum_{l: \ l\notin \partial G_{L} \cap l\in\theta_{L}(t) } \!\!\!\! A_{l}
\nonumber \\
&= O(\sqrt{L})+2 \Big|\big\{l: \ l \notin \partial G_{L} \cap l\in\theta_{L}(t)\big\}\Big|
 .\label{eq:inter:lemma:len}
\end{align}
The term $O(\sqrt{L})$ accounts for the contribution of links in $\partial G_{L}$, whose total number is less than $4\sqrt{L}$, and
whose coverage area is less than $2$.
Similarly, we have that 
\beqa
 \theta_{L}(t)=\frac{1}{L}\Big|\big\{l\notin \setB(G_{L}) \cap l \in \theta_{L}(t)\big\}\Big|+O(1/\sqrt{L})  \label{eq:inter:2:lemma:len}
\eeqa
 Dividing
\eqref{eq:inter:lemma:len} by $2L$, and using \eqref{eq:def:delta:l:n} and \eqref{eq:inter:2:lemma:len}, we have
\begin{align}
&\Big| \frac{1}{2L} \big(n^{2}-A^{c}_{L}(t)\big) - \theta_{L}(t)\Big|=
\nonumber \\ & \qquad \qquad 
\Big| \frac{1}{2L} \big(n^{2}-A^{c}_{L}(t)\big) - 0.5+\delta_{L}(t)\Big| = O(1/\sqrt{L}).
\end{align}
Noting that $ n^{2}/L =1-O(1/\sqrt{L})$ since $L=(n+1)^{2}$, from the above, we obtain 
$$\Big| \frac{ A^{c}_{L}(t)}{2L} -\delta_{L}(t)\Big| = O(1/\sqrt{L}), $$
which implies the statement in the lemma, as required.
\end{proof}


\begin{lemma}\label{lemma:bound:Ac} 
Consider the lattice $G_{L}$, attempt rate $z>0$, and a given time $t$, $0<t<\infty$. We have, w.p.1,
  \begin{align}
     A^{c}_{L}(t) \leq  \sqrt{2} \sum_{\setC\in \setC_{L}^{(nd)}(t,z)}\ell(\setC)   +10 L\theta_{L,h}(t) \nonumber
  \end{align}
where $A^{c}_{L}(t)$ is defined in Lemma~\ref{lemma:area}, $\setC_{L}^{(nd)}(t,z)$ is the set of non-dominating
clusters at time $t$ as defined in Section~\ref{sec:assumptionss}, and $\theta_{L,h}(t)$ is the density of links that are inactive and sense the channel as idle at time 
$t$, as defined in Appendix~\ref{sec:all:densities}.
\end{lemma}
\begin{proof}

Define the set $\setL^{(h)}(t)=\{l_{h,i},\  1\leq i\leq i_{h} \}$ to be the set of all inactive links
 that sense the channel as idle at time $t$. By definition,
$L\theta_{L,h}(t)$ is the total number of these links, and hence
$$i_{h}=L\theta_{L,h}(t).$$ 

Consider the following process. First, suppose $i_{h}>0$; we later also consider the case of $i_{h}=0$. Start from the first link $l_{h,1}$ in $\setL^{(h)}(t)$ and 
make it active, which results to a new set of active links. We then consider the next link $l_{h,2}$. If this link does not cause interference
to the new set of active links, we make link $l_{h,2}$ active. We continue the same process
for all $i=3,..., i_{h}$. By this process, 1) we reduce the uncovered area $A_{L}^{c}(t)$ by at most 
\begin{align}
2L\theta_{L,h}(t)  \label{eq:area:rem:p}
\end{align}
since link coverage area is at most $2$, 2) we might have a new set of clusters, which we denote by 
$$\setC_{L}'(t,z)=\{\setC'_i, 1\leq i\leq i_{max}'\},$$
and 3) the length of cluster 
boundaries changes at most by $4\sqrt{2}L\theta_{L,h}(t)$ since the coverage area of each newly added link can change or contribute at most
$4\sqrt{2}$ to the length of cluster boundaries. Therefore,
\begin{align}
  \Big| \sum_{\setC'\in \setC_{L}^{(s)}(t,z)} \ell(\setC') -\sum_{\setC\in \setC_{L}^{(nd)}(t,z)} \ell(\setC) \Big|\leq 4\sqrt{2}L\theta_{L,h}(t)
\label{ineq:p:n:boundC}
\end{align}
where $\setC_{L}^{(s)}(t,z)$ is a subset of clusters in $\setC_{L}'(t,z)$ that have the same type as the non-dominating clusters in 
$\setC_{L}^{(nd)}(t,z)$.
For the case where $i_{h}=0$, we have that $\theta_{L,h}(t)=0$, and the above inequality trivially holds by letting the new set of clusters 
to be the same as the original set of clusters at time $t$, i.e, by letting $ \setC_{L}'(t,z)=\setC_{L}(t,z)$.

After the above process, we cannot add any further links to the set of active links, and
the remaining uncovered area will be the union of areas around new cluster boundaries.
 For instance, consider the uncovered area between the two clusters in Fig.~\ref{fig:lattice_net:1}. By inspection, we find that
 this area is less than
 \begin{align}
 c_{u}\frac{\sqrt{2}}{2} \ell(\setC)  \label{eq:up:area:2cs}
 \end{align}
 where $c_{u}$ is a constant such that $0<c_u<2$, and
$\setC$ can be either of the clusters in the figure.

Considering the area removed in the defined process and the area between new cluster boundaries along with 
\eqref{eq:area:rem:p} and \eqref{eq:up:area:2cs}, we
have that
  \begin{align}
     A^{c}_{L}(t) \leq  \sqrt{2} \sum_{\setC' \in \setC_{L}^{(s)}(t,z) }  \ell(\setC')   +2 L\theta_{L,h}(t). \nonumber
  \end{align}
Using this inequality and \eqref{ineq:p:n:boundC}, we obtain the inequality in lemma.
\end{proof}



 \end{document}